\documentclass[11pt,letterpaper]{article}
\usepackage[left=1in,top=1in,right=1in,bottom=1in]{geometry}
\usepackage{amsmath} % use this to split displayed equations etc
\usepackage{enumitem} %allows compact lists

\usepackage{tikz}  %needed for \textcolor as well as tikz
\usetikzlibrary{arrows}
\usetikzlibrary{patterns,snakes}
\usetikzlibrary{decorations.shapes}
\tikzstyle{overbrace text style}=[font=\tiny, above, pos=.5, yshift=5pt]
\tikzstyle{overbrace style}=[decorate,decoration={brace,raise=5pt,amplitude=3pt}]
\usetikzlibrary{shapes.geometric}

\newsavebox{\fmbox}

\usepackage{thmtools}
\usepackage{thm-restate}
\usepackage{verbatim}
\usepackage{amsmath}
\usepackage{amssymb}
\usepackage{placeins} %needed for \FloatBarrier later

\usepackage{algorithm,algorithmicx}
\usepackage{algpseudocode}
 \usepackage{mathtools}

\usepackage{comment}

\newcount\Comments
\definecolor{darkgreen}{rgb}{0,0.7,0}
\newcommand{\kibitz}[2]{\ifnum\Comments=1\textcolor{#1}{#2}\fi}

\newtheorem{theorem}{Theorem}[section]
\newtheorem{definition}{Definition}

\newtheorem{corollary}[theorem]{Corollary}
\newtheorem{observation}[theorem]{Observation}
\newtheorem{lemma}[theorem]{Lemma}

\newtheorem{proposition}[theorem]{Proposition}

\usepackage{authblk}

\newcommand{\rset}{\mbox{{\normalfont I\hspace{-.4ex}R}}}
\newcommand{\np}{\mbox{{\normalfont NP}}}

\newcommand{\ppad}{\mbox{{\normalfont PPAD}}}

\newcommand{\ppa}{\mbox{{\normalfont PPA}}}

\newcommand{\pls}{\mbox{{\normalfont PLS}}}
\newcommand{\cls}{\mbox{{\normalfont CLS}}}
\newcommand{\tfnp}{\mbox{{\normalfont TFNP}}}

\newcommand{\leaf}{\mbox{{\sc Leaf}}}

\newcommand{\true}{\mbox{{\sc true}}}
\newcommand{\false}{\mbox{{\sc false}}}

\newcommand{\hide}[1]{}
\newcommand{\xb}{{\bf x}}
\newcommand{\yb}{{\bf y}}
\newcommand{\zb}{{\bf z}}
\newcommand{\eb}{{\bf e}}
\newcommand{\lplus}{\ensuremath{A_+}}
\newcommand{\lminus}{\ensuremath{A_-}}
\newcommand{\sensor}{\ensuremath{{\cal S}}}

\newcommand{\delt}{\ensuremath{\delta^{\rm tiny}}}

\newcommand{\delnt}{\ensuremath{\tilde{\delta}^{\rm tiny}}}
\newcommand{\pmeg}{\ensuremath{p^{\rm large}}}
\newcommand{\pgig}{\ensuremath{p^{\rm huge}}}

\def\tucker{{\sc Tucker}}
\def\twodtucker{2D-{\sc Tucker}}

\def\ch{{\sc Consensus-halving}}

\def\ns{{\sc Necklace-splitting}}
\def\vhdt{{\sc Variant high-D Tucker}}
\def\nvhdt{{\sc New variant high-D Tucker}}

\def\dhs{{\sc Discrete Ham Sandwich}}

\newcommand{\y}{100}
\newcommand{\x}{11$N$+1}

\renewcommand{\epsilon}{\varepsilon}

\newlength{\boxwidth}
\makeatletter
\DeclareRobustCommand{\qed}{%
  \ifmmode % if math mode, assume display: omit penalty etc.
  \else \leavevmode\unskip\penalty9999 \hbox{}\nobreak\hfill
  \fi
  \quad\hbox{\qedsymbol}}
\newcommand{\openbox}{\leavevmode
  \hbox to.77778em{%
  \hfil\vrule \vbox to.675em{\hrule width.6em\vfil\hrule}%
  \vrule\hfil}}
\newcommand{\qedsymbol}{\openbox}

\newenvironment{proof}[1][\proofname]{\par \normalfont
  \topsep6\p@\@plus6\p@ \trivlist
  \item[\hskip\labelsep\bfseries\itshape #1.]\ignorespaces }{%
  \qed\endtrivlist }

\newcommand{\proofname}{Proof}

\author[1]{Aris Filos-Ratsikas\thanks{Most of this work was performed when the author was at the University of Oxford.}}
\author[2]{Paul W. Goldberg}
\affil[1]{Department of Computer Science, 	{\'E}cole polytechnique f{\'e}d{\'e}rale de Lausanne {\tt aris.filosratsikas@epfl.ch}}
\affil[2]{Department of Computer Science, University of Oxford \ \ \ \ \ \ \ \  
	{\tt Paul.Goldberg@cs.ox.ac.uk}}

\date{%
	\today
}

\begin{document}

\title{The Complexity of Splitting Necklaces and Bisecting Ham Sandwiches}

\maketitle

\begin{abstract}
\noindent
We resolve the computational complexity of two problems known as \ns\ and
\dhs\, showing that they are PPA-complete.
For \ns, this result is specific to the important special
case in which two thieves share the necklace.
We do this via a PPA-completeness result for an approximate version of the \ch\
problem, strengthening our recent result that the problem is PPA-complete
for inverse-exponential precision. At the heart of our construction is
a smooth embedding of the high-dimensional M\"{o}bius strip in the \ch\ problem.
These results settle the status of PPA as a class that captures the complexity
of ``natural'' problems whose definitions do not incorporate a circuit.
\end{abstract}

\begin{paragraph}{Keywords:}
Computational complexity; Tucker's Lemma; TFNP; fair division
\end{paragraph}
\section{Introduction}\label{sec:intro}

The complexity classes \ppa\ and \ppad\ were introduced in a seminal paper of Papadimitriou \cite{Pap} in 1994,
in an attempt to classify several natural problems in the class \tfnp\ \cite{MP91}.
\tfnp\ is the class of \emph{total search problems} in \np\ for which a solution
exists for every instance, and solutions can be efficiently verified.
Various important problems were subsequently proven to be complete for the class \ppad, such as the
complexity of many versions of Nash equilibrium
\cite{DGP09,CDT,EGG06,Mehta14,R15,CDO15}, market equilibrium
computation \cite{CSVY08,CDDT09,VY11,CPY13,SSB}, and others \cite{DQS09,KPRST13}.
As evidence of computational hardness,
\ppa-completeness is stronger than \ppad-completeness, i.e., $\ppad\subseteq\ppa$.
Indeed, Je{\v r}{\'a}bek~\cite{J12} shows that it indicates cryptographic
hardness in a strong sense: \cite{J12} gives a randomised reduction from FACTORING to \ppa-complete problems.
This is not known for \ppad-complete problems.
For more details, and the significance of \ppa-completeness, we refer
the reader to the related discussion in~\cite{FG17}.
\ppa\ is the class of problems reducible to \leaf\ (Definition~\ref{def:leaf}), and a \ppa-complete
problem is polynomial-time equivalent to \leaf.
\begin{definition}\label{def:leaf}
An instance of the problem \leaf\ consists of an undirected graph $G$ whose vertices
have degree at most 2; $G$ has $2^n$ vertices represented by bitstrings of length $n$;
$G$ is presented concisely via a circuit that takes as input a vertex and outputs
its neighbour(s). We stipulate that vertex $0^n$ has degree 1. The challenge is
to find some other vertex having degree 1.
\end{definition}

Complete problems for the class \ppa\ seemed to be much more elusive than \ppad-complete ones,
especially when one is interested in ``natural'' problems, where ``natural'' here has the very
specific meaning of problems that do not explicitly contain a circuit in their definition.
Besides Papadimitriou~\cite{Pap}, other papers asking about the possible existence of
natural \ppa-complete problems include \cite{Grigni01,CD09,DGP09,DEFKQX}.
In a recent precursor~\cite{FG17} to the present paper we identified the first example of such a problem,
namely the approximate \ch\ problem,
dispelling the suspicion that such problems might not exist.
In this paper we build on that result and settle the complexity of two natural and important problems whose complexity
status were raised explicitly as open problems in Papadimitriou's paper itself, and in many other papers
beginning in the 1980s.
Specifically, we prove that \ns\ (with two thieves, see Definition~\ref{def:ns})
and \dhs\ are both \ppa-complete.

\begin{definition}[Necklace Splitting]\label{def:ns}
In the $k$-\ns\ problem there is an open necklace with $ka_i$ beads of colour $i$, for $1\leq i\leq n$.
An ``open necklace'' means that the beads form a string, not a cycle.
The task is to cut the necklace in $(k-1)\cdot n$ places and partition the resulting
substrings into $k$ collections, each containing precisely $a_i$ beads of colour $i$, $1\leq i\leq n$.
\end{definition}
In Definition~\ref{def:ns},
$k$ is thought of as the number of thieves who desire to split the necklace in such
a way that the beads of each colour are equally shared.
In this paper, usually we have $k=2$ and we refer to this special case as \ns.

\begin{definition}[Discrete Ham Sandwich]\label{def:dhs}
In the \dhs\ problem, there are $n$ sets of points in $n$ dimensions
having integer coordinates (equivalently one could use rationals).
A solution consists of a hyperplane that splits each set of points into
subsets of equal size (if any points lie on the plane, we are allowed to place
them on either side, or even split them arbitrarily).
\end{definition}
In Definition~\ref{def:dhs}, each point set represents an ingredient of the sandwich,
which is to be cut by a hyperplane in such a way that all ingredients are equally split.

The \emph{necklace-splitting problem}  was introduced in a 1982
paper of Bhatt and Leiserson~(\cite{BL82}, Section 5), where it arose in the
context of VLSI circuit design (the version defined in \cite{BL82} is the 2-thief
case proved \ppa-complete in the present paper).
In 1985 and 1986, the 2-thief case was shown to have guaranteed solutions
(as defined in Definition~\ref{def:ns})
by Goldberg and West~\cite{GW85} and Alon and West~\cite{AW86} and then in 1987, Alon \cite{Alon87} proved 
existence of solutions for $k$ thieves as well.
Early papers that explicitly raise its complexity-theoretic status as
an open problem are Goldberg and West~\cite{GW85} and Alon~\cite{Alon88,Alon90}.
Subsequently, the necklace-splitting problem was found to be closely related to
``paint-shop scheduling'', a line of work in which several papers such as \cite{Meunier08,MS09,MN12}
explicitly mention the question of the computational complexity of necklace-splitting.
Meunier~\cite{Meunier08} notes that the search for a {\em minimum} number of cuts admitting a fair division
(which may be smaller than the number $(k-1)n$ that is guaranteed to suffice)
is NP-hard, even for a subclass of instances of the 2-thief case.
(That is a result of Bonsma et al.~\cite{BEH06}, for the ``paint shop problem
with words'', equivalent to 2-thief \ns\ with 2 beads of each colour.)

In \cite{FG17}, we showed \ns\ to be computationally equivalent to
$\varepsilon$-\ch\ for inverse-polynomial precision parameter $\varepsilon$, but
the \ppa-completeness of $\varepsilon$-\ch\ was only shown for inverse-exponential $\varepsilon$.
\cite{FG17} established \ppad-hardness of \ns, applying the main result of \cite{FFGZ18}.
In this paper, we prove that
$\varepsilon$-\ch\ is \ppa-complete for $\varepsilon$ inversely polynomial,
thus obtaining the desired \ppa-completeness of \ns. While some structural parts 
of our reduction are extensions of those presented
in \cite{FG17}, obtaining the result for inverse-polynomial precision
is much more challenging,
as the construction needs to move to a high-dimensional space (rather than the
two-dimensional space which is sufficient for the result in \cite{FG17}). We highlight
the main new techniques that we have developed in this paper in Section \ref{sec:overview},
where we provide an overview of the reduction.
Our \ppa-completeness result gives a convincing negative answer to Meunier
and Neveu's questions~\cite{MN12} about possible polynomial-time solvability
or membership of \ppad\ for \ns; likewise it runs counter to Alon's cautious optimism
at ICM 1990~(\cite{Alon90}, Section 4) that the problem may be solvable in polynomial time.

%Our \ppa-completeness result corrects a misapprehension in the literature~\cite{Pap,MN12}
%that \ns\ belongs to \ppad, rules out the possible \ppad-completeness of 2-thief
%\ns\ pointed out as an open problem by Meunier and Neveu \cite{MN12}, and gives an even more devastating
%negative answer to their other question about a possible polynomial-time algorithm.

The \emph{Ham Sandwich Theorem} \cite{ST42} is of enduring and widespread interest
due to its colourful and intuitive statement, and its relevance and applications in
topology, social choice theory, and computational geometry.
Roughly, it states that given $d$ measures in Euclidean $d$-space, there exists a
hyperplane that cuts them all simultaneously in half.
Early work on variants and applications of the theorem focused on non-constructive
existence proofs and mostly did not touch on the algorithmics.
A 1983 paper by Hill~\cite{Hill-amm83} hints at possible interest in the
corresponding computational challenge, in the context of a related land division problem.
The computational problem (and its complexity) was first properly studied
in a line of work in computational geometry beginning in the 1980s,
for example~\cite{EW86,LMS-stoc92,LMS-dcg94,M94}.
The problem envisages input data consisting of $d$ sets of $n$ points in Euclidean
$d$-space, and asks for a hyperplane that splits all point sets in half.
(The problem \dhs~(Definition~\ref{def:dhs}) as named in \cite{Pap} is essentially this,
with $d$ set equal to $n$ to emphasise that we care about the high-dimensional case.)
In this work in computational geometry, the emphasis has been on efficient algorithms
for small values of $d$; Lo et al.~\cite{LMS-dcg94} improve the dependence on $d$ but it is still exponential,
and the present paper shows for the first time that we should {\em not} expect to improve
on that exponential dependence.
More recently, Grandoni et al.~\cite{GRSZ14} apply the ``Generalized Ham
Sandwich Theorem'' to a problem in multi-objective optimisation and note that
a constructive proof would allow a more efficient algorithm to emerge.
The only computational hardness result we know of is Knauer et al.~\cite{KTW11}
who obtain a $W[1]$-hardness result for a constrained version of the problem;
\cite{KTW11} points out the importance of the computational complexity of the general problem.
The \ppa-completeness result of the present paper is the first hardness result
{\em of any kind} for \dhs, and as we noted, is a strong notion of computational
hardness. Karpic and Saha~\cite{KS17} showing a form of equivalence between the Ham Sandwich
Theorem and Borsuk-Ulam, explicitly mention the possible \ppa-completeness of
\dhs\ as an ``interesting and challenging open problem''.

We prove the \ppa-completeness of \dhs\ via a simple reduction from \ns.
Ours is not the first paper to develop the close relationship between the two problems:
Blagojevi\'{c} and Sober\'{o}n~\cite{BS17} shows a generalisation,
where multiple agents may share a ``sandwich'', dividing it into convex pieces.
Further papers to explicitly point out their computational complexity as open problems include
Deng et al.~\cite{DFK} (mentioning that both problems ``show promise to be complete for \ppa''),
Aisenberg et al.~\cite{ABB}, and Belovs et al.~\cite{BIQSY17}.\\

\noindent \textbf{Further Related Work:} The class \tfnp\ was defined in \cite{MP91} and several of its subclasses were studied over the years, such as PPA, PPAD and PPP \cite{Pap}, \pls\ \cite{JPY} and \cls\ \cite{DP11}; here we focus on the most recent results. As we mentioned earlier, in \cite{FG17} we identified the first natural complete problem for \ppa, the approximate \ch\ problem. In a recent paper, Sotiraki et al. \cite{SZZ18} identified the first natural problem for the class PPP, the class of problems whose totality is established by an argument based on the pigeonhole principle. For the class CLS, both Daskalakis et al. \cite{DTZ18} and Fearnley et al. \cite{FGMS18} identified complete problems (two versions of the Contraction Map problem, where a metric or a meta-metric are given as part of the input). In the latter paper, the authors define a new class, namely EOPL (for ``End of Potential Line''), and show that it is a subclass of CLS. Furthermore, they show that two well-known problems in CLS, the P-Matrix Linear Complementarity Problem (P-LCP), and finding a fixpoint of a piecewise-linear contraction map (\textsc{PL-Contraction}) belong to the class. The \textsc{End of Potential Line} problem of \cite{FGMS18} is closely related to the \textsc{End of Metered Line} of \cite{HY17}.

%\subsection{Techniques}
%
%The precursor~\cite{FG17} to the present paper established that $\varepsilon$-\ch\ is
%\ppa-complete for inverse-exponential $\varepsilon$; the present
%paper is mostly devoted to establishing the same thing for inverse-polynomial $\varepsilon$.
%While some structural parts of our reduction are extensions of those presented
%in \cite{FG17}, obtaining the result for inverse-polynomial precision
%is much more challenging,
%as the construction needs to move to a high-dimensional space (rather than the
%two-dimensional space which is sufficient for the result in \cite{FG17}).
%As in \cite{FG17}, our starting-point is $2D$-\tucker~\cite{ABB}.
%\begin{itemize}
%\item We apply a novel form of the snake-embedding technique to map an instance of
%$2D$-\tucker\ to $n$ dimensions, giving rise to a significantly more complicated
%variant of the Tucker problem, which will be used in the reduction to $\varepsilon$-\ch.
%\item
%At the heart of our construction lies a new coordinate system for an $n$-simplex
%having 2 facets identified to embed an $n$-dimensional M{\"o}bius strip,
%and a smooth transformation (in particular, distances are polynomially related)
%between it and a coordinate system whose values are encoded
%using solution cuts in \ch\ (Figure~\ref{fig:dir-numbers} illustrates.)
%\item
%Additional features of the reduction (and their proofs) become substantially
%more complex technically, relative to~\cite{FG17}.
%\end{itemize}
%Section~\ref{sec:overview} discusses the new ideas in more detail.

\section{Problems and Results}\label{sec:probres}

We present and discuss our main results, and in Section~\ref{sec:overview}
we give an overview of the proof and new techniques, in particular with respect
to the precursors~\cite{FFGZ18,FG17} to this paper.

\begin{definition}[$\varepsilon$-Consensus Halving \cite{SS03,FG17}]\label{def:ch}
An instance $I_{CH}$ incorporates, for $1\leq i\leq n$, a non-negative measure $\mu_i$
of a finite line interval $A=[0,x]$, where each $\mu_i$ integrates to 1 and $x>0$ is part of the input.
We assume that $\mu_i$ are step functions represented in a standard
way, in terms of the endpoints of intervals where $\mu_i$ is constant,
and the value taken in each such interval.
We use the bit model (logarithmic cost model) of numbers.
$I_{CH}$ specifies a value $\varepsilon\geq 0$ also using the bit model.
We regard $\mu_i$ as the value function held by agent $i$ for subintervals of $A$.

A solution consists firstly of a set of $n$ {\em cut points} in $A$ (also given in
the bit model of numbers).
These points partition $A$ into (at most) $n+1$ subintervals, and the second
element of a solution is that each subinterval is labelled $\lplus$ or $\lminus$.
This labelling is a correct solution provided that for each $i$,
$|\mu_i(\lplus)-\mu_i(\lminus)|\leq\varepsilon$, i.e.\ each agent has a value
in the range $[\frac{1}{2}-\frac{\varepsilon}{2},\frac{1}{2}+\frac{\varepsilon}{2}]$ for the
subintervals labelled $\lplus$ (hence also values the subintervals
labelled $\lminus$ in that range).
\end{definition}

We assume without loss of generality that in a valid solution, labels $\lplus$ and $\lminus$ alternate.
We also assume that the alternating label sequence begins with label
$\lplus$ on the left-hand side of $A$ (i.e. $\lplus$ denotes the leftmost label in a \ch\ solution).

The \ch\ problem of Definition~\ref{def:ch} is a computational version of the {\em Hobby-Rice theorem}~\cite{HR65}.
Most of the present paper is devoted to proving the following theorem.

\begin{theorem}~\label{thm:main}
$\varepsilon$-\ch\ is \ppa-complete for some inverse-polynomial $\varepsilon$. %,
%hence by~\cite{FG17}, \ns\ is \ppa-complete, even in the $k=2$ special case of $k$-\necklace.
\end{theorem}
As mentioned in the introduction, in \cite{FG17} it was proven that 2-thief
\ns\ and $\varepsilon$-\ch\, for $\varepsilon$ inversely-polynomial
are computationally equivalent, i.e. they reduce to each other in polynomial time.
Therefore, by \cite{FG17} and the ``in \ppa'' result proven in Section~\ref{sec:hsinppa},
we immediately get the following corollary.

\begin{theorem}\label{thm:necklace}
\ns\ is \ppa-complete when there are $k=2$ thieves.
\end{theorem}

If we knew that $k$-\ns\ belonged to \ppa\ for other values of $k$, we could of course make
the blanket statement ``\ns\ is \ppa-complete''. Alas, the proofs that \ns\ is a total
search problem for $k>2$ \cite{Alon87,Meunier14} do {\em not} seem to boil down to
the parity argument on an undirected graph!
That being said, we do manage to establish membership of \ppa\ for $k$ being a power of 2
(essentially an insight of \cite{Alon87}), see Section~\ref{sec:inppa} of the Appendix for 
the details and a related discussion. 
Of course, the $k=2$ result strongly suggests that $k$-\ns\ is a hard problem for other values of $k$.

As it happens, the \ppa-completeness of \dhs\ follows straightforwardly,
and we present that next. The basic idea of Theorem~\ref{thm:dhs} of embedding
the necklace in the moment curve appears already in \cite{RS07,Matou} and \cite{LGMN17}, p.48.

\begin{theorem}\label{thm:dhs}
\dhs\ is \ppa-complete.
\end{theorem}

\begin{proof}
Inclusion in \ppa\ is shown in Section~\ref{sec:hsinppa} of the Appendix.
For \ppa-hardness, we reduce from 2-thief \ns\ which is \ppa-complete by Theorem \ref{thm:necklace}.

The idea is to embed the necklace into the \emph{moment curve}
$\gamma=\{(\alpha,\alpha^2,\ldots,\alpha^n) : \alpha\in[0,1] \}$.  
Assume all beads lie in the unit interval $[0,1]$.
A bead having colour $i\in[n]$ located at $\alpha\in[0,1]$ becomes a point 
mass of ingredient $i$ of the ham sandwich located at $(\alpha,\alpha^2,\ldots,\alpha^n)\in\rset^n$.
It is known that any hyperplane intersects the moment curve $\gamma$ in at most $n$ points,
(e.g. see \cite{Matou}, Lemma 5.4.2), therefore a solution
to \dhs\ corresponds directly to a solution to \ns, where the two thieves splitting the
necklace take alternating pieces.
(In the $k=2$ case, we may assume without loss of generality that they do
in fact take alternating pieces).
\end{proof}
A limitation to Theorem~\ref{thm:dhs} is that the coordinates may be exponentially
large numbers; they could not be written in unary.
We leave it as an open problem whether a unary-coordinate version is also \ppa-complete.
As defined in \cite{Pap}, \dhs\ stipulated that each of the $n$ sets of points is of size $2n$,
whereas Definition~\ref{def:dhs} allows polynomial-sized sets.
We can straightforwardly extend \ppa-completeness to the version of  \cite{Pap} by adding ``dummy dimensions''
whose purpose is to allow larger sets of each ingredient; the new ingredients that
are introduced, consist of compact clusters of point masses, each cluster
in general position relative to the other clusters and the subspace of dimension
$n$ that contains the points of interest.

\begin{paragraph}{Notation:}
We use the standard notation $[n]$ to denote the set $\{1,\ldots,n\}$,
and we also use $\pm[n]$ to denote $\{1,-1,2,-2,\ldots,n,-n\}$.
We often refer to elements of $\pm[n]$ as ``labels'' or ``colours''.
$\lambda$ is usually used to denote a labelling function
(so its codomain is $\pm[n]$).

We let $A$ denote the domain of an instance of \ch; if that instance has complexity
$n$ then $A$ will be the interval $[0,poly(n)]$, where $poly(n)$ is some number bounded by a polynomial
in $n$. Recall by Definition \ref{def:ch} that $\mu_a$ denotes the value function, or measure, of agent $a$ on the domain $A$, in a \ch\ instance. We also associate each agent with its own cut (recall that the number of agents and cuts is supposed to be equal) and we let $c(a)$ denote the cut associated with agent $a$.

We let $p^C(n)$ be a polynomial that represents the number of ``circuit-encoders'' that we
use in our reduction (see Section \ref{sec:forward}); we usually denote it $p^C$, dropping the $n$ from the notation.

Finally, $B$ denotes the $n$-cube (or ``box'') $[-1,1]^n$.
\end{paragraph}

\begin{paragraph}{Terminology}
In an instance of \ch, a {\em value-block} of an agent $a$ denotes a sub-interval of the 
domain where $a$ possesses positive value, uniformly distributed on that interval.
In our construction, value-blocks tend to be scattered rather sparsely along the domain.
\end{paragraph}

\subsection{Overview of the proof}\label{sec:overview}

We review the ground covered by the precursors~\cite{FFGZ18,FG17} to this paper,
then we give an overview of the technical innovations of the present paper.

\subsubsection{Ideas from~\cite{FFGZ18,FG17}}

\cite{FFGZ18} established \ppad-hardness of \ch.
An arithmetic circuit can be encoded by an instance to $\epsilon$-\ch,
by letting each gate have a cut whose location (in a solution) represents the (approximate) value
taken at that gate. Agents' valuation functions ensure that values taken at
the gates behave according to the type of gate.
A ``\ppad\ circuit'' can then be represented using an instance of \ch.

\cite{FG17} noted that the search space of solutions to instances as constructed by \cite{FFGZ18},
is oriented. A radical new idea was needed to encode the
non-oriented feature of topological spaces representable by \ppa.
That was achieved by using two cuts to represent the coordinates of a point
on a triangular region faving two sides identified to form a M\"{o}bius strip.
(These cuts are the only ones that lie in a specific subinterval of the interval $A$ of
a \ch\ instance, called the ``coordinate-encoding (c-e) region''.
The two cuts are called the ``coordinate-encoding cuts''.)
Identifying two sides in this way is done by exploiting the equivalence of taking
a cut on the LHS of the c-e region, and moving it to the RHS.
In order to embed a hard search problem into the surface of a standard 2-dimensional
M\"{o}bius strip, it was necessary to work at exponentially-fine resolution,
which immediately required inverse-exponential $\epsilon$ for instances of $\epsilon$-\ch.
In \cite{FG17} we reduced from the \ppa-complete problem \twodtucker~\cite{ABB}
(Definition~\ref{def:2d-tucker-abb} below), a search problem on an exponential-sized
2-dimensional grid.

In \cite{FG17}, the rest of $A$ is called the ``circuit-encoding region'' $R$,
and the cuts occurring within $R$ do the job of performing computation on
the location of cuts in the c-e region.
The present paper retains this high-level structure (Section~\ref{sec:bbb}).
As in \cite{FG17} we use multiple copies of the circuit that performs the
computation, each in its own subregion of $R$.
Here we use $p^C(n)$ copies where $p^C$ is a polynomial; in \cite{FG17}
we used 100 copies. Each copy is called a {\em circuit-encoder}.
The purpose of multiple copies is to make the system robust; a small
fraction of copies may be unreliable: as in \cite{FG17} we have to account for
the possibility that one of the c-e cuts may occur in the circuit-encoding region,
rendering one of the copies unreliable.
We re-use the ``double negative lemma'' of \cite{FG17} that such a cut is not too harmful.
We also adapt a result of \cite{FG17} that when a cut is moved from the one end
to the other end of the c-e region, this corresponds to identifying two facets of a simplex
to form a M\"{o}bius strip.

\cite{FG17} uses a sequence of ``sensor agents'' to identify the endpoints of intervals
labelled $\lplus$ and $\lminus$ in the coordinate-encoding region, and
feed this information into the above mentioned circuit-encoders, which perform computation on those values.
As in \cite{FG17} we use sensor agents.
We obtain a simplification with respect to \cite{FG17} which is that we do not
need the gadgets used there to perform ``bit-extraction'' (converting the
position of a c-e cut into $n$ boolean values).
In \cite{FG17}, a solution to an instance of \ch\ was associated with a sequence
of 100 points in the M\"{o}bius-simplex (referred there as the ``simplex-domain''),
and the ``averaging manoeuvre'' introduced in \cite{DGP09} was applied.
In this paper, for a polynomial $p^C(n)$, we sample a sequence of $p^C$ points in a more elegant manner,
again exploiting the inverse-polynomial precision of solutions that we care about.

\subsubsection{Technical innovations}

As in \cite{FG17}, we reduce from the \ppa-complete problem \twodtucker~\cite{ABB}
(Definition~\ref{def:2d-tucker-abb}).
That computational problem uses an exponentially-fine 2D grid,
and (unlike~\cite{FG17}), in Section~\ref{sec:snake}
we apply the {\em snake-embedding} technique invented in~\cite{CDT}
(versions of which are used in~\cite{DEFKQX,DFK} in the context of \ppa)
to convert this to a grid of fixed resolution, at the expense of going from 2 to $n$ dimensions.
The new problem, \vhdt\ (Definition~\ref{def:vhdt}) envisages a $7\times 7\times\cdots\times 7$ grid.
Here, we design the snake-embedding in such a way that \ppa-completeness
holds for instances of the high-dimensional problem that obey a 
further constraint on the way the high-dimensional grid is coloured, that we
exploit subsequently.
A further variant, \nvhdt\ (Definition~\ref{def:nvhdt}) switches to a ``dual'' version
where a hypercube is divided into cubelets, and points in the hypercube are coloured
such that interiors of cubelets are monochromatic.
A pair of points is sought having equal and opposite colours and
distant by much less than the size of the cubelets.

We encode a point in $n$ dimensions using a solution to an instance of \ch\ as follows.
Instead of having just 2 cuts in the coordinate-encoding region (as in \cite{FG17}), suppose we ensure
that up to $n$ cuts lie there.
These cuts split this interval into $n+1$ pieces whose total length is constant,
so represent a point in the unit $n$-simplex (in \cite{FG17}, the unit 2-simplex).
This ``M\"{o}bius-simplex'' (Definition~\ref{def:simplex-domain}; Figure~\ref{fig:subspaces})
has the further property that two facets are identified with each other in a way that
effectively turns the simplex into an $n$-dimensional M\"{o}bius strip.\\

\noindent In Section~\ref{sec:coord-sys} we define a crucially-important coordinate transformation 
(see Figure~\ref{fig:dir-numbers}) with the following key properties
\begin{itemize}
\item the transformation and its inverse can be computed efficiently, and distances between transformed
coordinate vectors are polynomially related to distances between un-transformed vectors;
\item at the two facets that are identified with each other, the coordinates of corresponding
points are the negations of each other; our colouring function (that respects Tucker-style
boundary conditions) has the effect that antipodal points get equal and opposite
colours, and {\em no undesired solutions are introduced at these facets}.
\end{itemize}
This is the ``smooth embedding'' referred to in the abstract.\\

\noindent With the aid of the above coordinate transformation, we divide up the M\"{o}bius-simplex:
\begin{itemize}
\item{
The {\em twisted tunnel} (Definition~\ref{def:twisted}) is an inverse-polynomially thick strip, connecting the two facets that
are identified in forming the M\"{o}bius-simplex.
It contains at its centre an embedded copy of the hypercube domain of
an instance $I_{VT}$ of \nvhdt. Outside of this embedded copy,
it is ``coloured in'' (using our new coordinate system) in a way that avoids introducing
solutions that do not encode solutions of $I_{VT}$.
}
\item{ The {\em Significant Region} contains the twisted tunnel and constitutes a somewhat
thicker strip connecting the two facets. It serves as a buffer zone between the twisted tunnel
and the rest of the M\"{o}bius-simplex.
It is subdivided into subregions where each subregion has a unique set of labels, or colours, from $\pm[n]$.
(We sometimes refer to these as ``colour-regions''.)
It is shown that any solution to an instance of \ch\ constructed as in
our reduction, represents a point in the Significant Region.
}
\item If, alternatively, a set of cuts represents a point from outside the
Significant Region, then certain agents (so-called ``blanket-sensor agents'')
will observe severe imbalances between labels $\lplus$ and $\lminus$,
precluding a solution.
\end{itemize}

In \cite{FG17}, it was relatively straightforward to integrate the subset of the 2-dimensional
M\"{o}bius-simplex that corresponds with the twisted tunnel,
with the parts of the domain where the blanket-sensor agent became active
(ruling out a solution) in a way that avoided introducing solutions that fail to
encode solutions of \tucker. In the present paper, that gap has to be ``coloured-in''
in a carefully-designed way (Section~\ref{sec:colour-f}, list item \ref{item-f-sr}),
and this is the role of the part of the Significant Region that is not the twisted tunnel.
The proofs that they work correctly (Sections \ref{sec:solution}, \ref{sec:bdry1}) become
more complicated.

\section{Snake embedding reduction}\label{sec:snake}

The purpose of this section is to establish the \ppa-completeness of
\nvhdt, Definition~\ref{def:nvhdt}.
The snake embedding construction was devised in~\cite{CDT},
in order to prove that $\varepsilon$-Nash equilibria are \ppad-complete
to find when $\varepsilon$ is inverse polynomial; without this trick the result
is just obtained for $\varepsilon$ being inverse exponential. We do a similar trick here.
We will use as a starting-point the \ppa-completeness of 2D-\tucker, from~\cite{ABB},
which is the following problem:

\begin{definition}\label{def:2d-tucker-abb} (Aisenberg et al.~\cite{ABB})
	An instance of \twodtucker\ consists of a labelling
	$\lambda:[m]\times[m]\rightarrow\{\pm 1,\pm 2\}$
	such that for $1\leq i,j\leq m$, $\lambda(i,1)=-\lambda(m-i+1,m)$
	and $\lambda(1,j)=-\lambda(m,m-j+1)$.
	A solution to such an instance of \twodtucker\ is a pair of vertices $(x_1,y_1)$, $(x_2,y_2)$
	with $|x_1-x_2|\leq 1$ and $|y_1-y_2|\leq 1$ such that $\lambda(x_1,y_1)=-\lambda(x_2,y_2)$.
\end{definition}

The hardness of the problem in Definition~\ref{def:2d-tucker-abb} arises when
$m$ is exponentially-large, and the labelling function is presented by
means of a boolean circuit.

We aim to prove the following is PPA-complete, even when the values $m_i$
are all upper-bounded by some constant (specifically, 7).

\begin{definition}\label{def:nd-tucker} (Aisenberg et al.~\cite{ABB})
	An instance of $n$D-\tucker\ consists of a labelling
	$\lambda:[m_1]\times\cdots\times[m_n]\rightarrow\{\pm 1,\cdots, \pm n\}$
	such that if a point ${\bf x}=(x_1,\ldots,x_n)$
	lies on the boundary of this grid (i.e., $x_i=1$ or $x_i=m_i$ for some $i$),
	then letting $\bar{\bf x}$ be the antipodal point of {\bf x}, we have $\lambda(\bar{\bf x})=-\lambda({\bf x})$.
	(Two boundary points are antipodal if they lie at opposite ends of a line
	segment passing through the centre of the grid.)
	A solution consists of two points {\bf z}, ${\bf z}'$ on this grid, having opposite labels
	($\lambda({\bf z})=-\lambda({\bf z}')$), each of whose coordinates differ (coordinate-wise) by at most 1.
	
	It is assumed that $\lambda$ is presented in the form of a circuit,
	syntactically constrained to give opposite labels to antipodal grid points.
\end{definition}

\begin{definition}\label{def:vhdt}
	An instance of \vhdt\ is similar to Definition \ref{def:nd-tucker} but whose
	instances obey the following additional constraints.
	The $m_i$ are upper bounded by the constant 7.
	We impose the further constraint that the facets of the cube are
	coloured with labels from $\pm[n]$ such
	that all colours are used, and opposite facets have opposite labels, and
	for $2\leq i\leq n$ it holds that the facet with label $i$
	(resp. $-i$) has no grid-point on that facet with label $i$ (resp. $-i$).
\end{definition}

\begin{theorem}\label{thm:snake}
	\vhdt\ is \ppa-complete.
\end{theorem}

\begin{paragraph}{Informal description of snake embedding}
	A snake-embedding consists of a reduction from $k$D-\tucker\ to $(k+1)$D-\tucker,
	which we describe informally as follows. See Figure~\ref{fig:se}.
	Let $I$ be an instance of $k$D-\tucker, on the grid $[m_1]\times\cdots\times[m_k]$.
	Embed $I$ in $(k+1)$-dimensional space, so that it lies in the grid
	$[m_1]\times\cdots\times[m_k]\times[1]$. Then sandwich $I$ between two layers,
	where all points in the top layer get labelled $k+1$, and points in the bottom layer
	get labelled $-(k+1)$, as in the left part of Figure \ref{fig:se}. We now have points in the grid 
	$[m_1]\times\cdots\times[m_k]\times[3]$, and notice that this construction
	preserves the required property that points on the boundary have labels
	opposite to their antipodal points.
	
	Then, the main idea of the snake embedding is the following. We fold this grid
	into three layers, by analogy with folding a sheet of
	paper along two parallel lines so that the cross-section is a zigzag line,
	and one dimension of the folded paper is one-third of the unfolded version,
	the other dimension being unchanged (see the right hand side of Figure \ref{fig:se}).
	In higher dimension, suppose that $m_1$ is the largest value of any $m_i$. 
	Then, we can reduce $m_1$ by a factor of about 3, while causing the final coordinate
	to go up from 3 to 9.
	By merging layers of label $k+1$ and $-(k+1)$, the thickness of 9 reduces to 7.
	{\em This operation preserves the labelling rule for antipodal boundary points.}\\
	
	\noindent However, there are two points that need extra care for the reduction to go through:
	\begin{itemize}
		\item Firstly, simply folding the layers such that their cross-sections are zigzag lines may
		introduce diagonal adjacencies between cubelets that were not present in the original
		instance in $k$-dimensions, i.e. we might end up generating adjacent cubelets with equal-and-opposite
		colours, see the left part of Figure \ref{fig:se2} for an illustration. To remedy this, we will ``copy'' (or ``duplicate'') the cubelets 
		at the folding points, essentially having three cubelets of the same colour, whose cross-sections are the 
		short vertical section in the right hand side of Figure \ref{fig:se}, see also the right hand side of Figure \ref{fig:se2} for an illustration. From now on, when referring to ``folding'', we will mean the version where we also duplicate the cubelets at the folding points, as described above.
		\item Secondly, the folding and duplicating operation only works if $m_1$ is a multiple of $3$,  as otherwise the $(k+1)$-dimensional instance may not satisfy the boundary conditions of Definition \ref{def:nd-tucker}, i.e. we might end up with antipodal cubelets that do not have equal-and-opposite colours. To ensure that $m_1$ is a multiple of $3$ before folding, we can add $1$ or $2$ additional layers of cubelets to $m_1$, (depending on whether the remainder of the division of $m_1$ by $3$ is either $2$ or $1$ respectively). These layers are duplicate copies of the outer layers of cubelets at opposite ends of the length-$m_1$ direction; if there is only one additional layer to be added, we can add on either side. Note that this operation does not generate any cubelets of equal-and-opposite labels that were not there before and the same will be true for the instance after the folding operation. See Figure \ref{fig:se3} for an illustration.
	\end{itemize}

\end{paragraph}

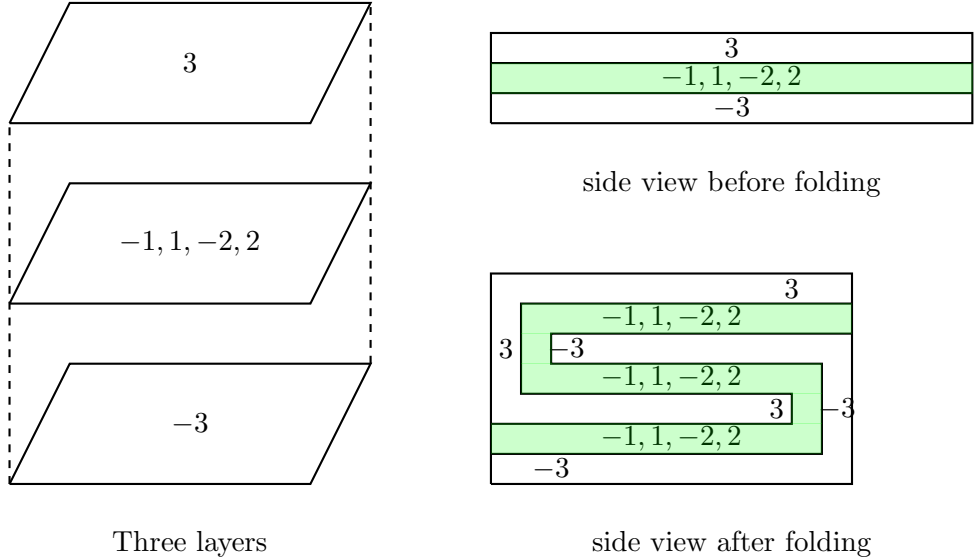
\begin{figure}
	\center{
		\begin{tikzpicture}[scale=0.8]
		\tikzstyle{xxx}=[dashed,thick]
		
		\draw[thick](0,2)--(5,2)--(6,4)--(1,4)--(0,2);
		\node at(3,3){$-3$};
		\draw[thick](0,5)--(5,5)--(6,7)--(1,7)--(0,5);
		\node at(3,6){$-1,1,-2,2$};
		\draw[thick](0,8)--(5,8)--(6,10)--(1,10)--(0,8);
		\node at(3,9){$3$};
		\draw[thick,dashed](0,2)--(0,8);\draw[thick,dashed](6,4)--(6,10);
		\node at(3,1){Three layers};
		
		\draw[thick](8,8)--(16,8)--(16,9.5)--(8,9.5)--(8,8);
		\draw[thick](8,8.5)--(16,8.5);\draw[thick](8,9)--(16,9);
		\node at(12,9.25){$3$};\node at(12,8.75){$-1,1,-2,2$};\node at(12,8.25){$-3$};
		\node at(12,7){side view before folding};
		
		\fill[fill=green, opacity=0.2] (8,8.5) rectangle (16,9);
		
		\draw[thick](8,2)--(14,2)--(14,5.5)--(8,5.5)--(8,2);
		\draw[thick](8,2.5)--(13.5,2.5)--(13.5,4)--(9,4)--(9,4.5)--(14,4.5);
		\draw[thick](8,3)--(13,3)--(13,3.5)--(8.5,3.5)--(8.5,5)--(14,5);
		\node at(13,5.25){$3$};
		\node at(11,4.75){$-1,1,-2,2$};
		\node at(8.25,4.25){$3$};\node at(9.25,4.25){$-3$};
		\node at(11,3.75){$-1,1,-2,2$};
		\node at(12.75,3.25){$3$};\node at(13.75,3.25){$-3$};
		\node at(11,2.75){$-1,1,-2,2$};
		\node at(9,2.25){$-3$};
		\node at(12,1){side view after folding};
		
		\fill[fill=green, opacity=0.2] (8,2.5) rectangle (13.5,3);
		\fill[fill=green, opacity=0.2] (8.5,3.5) rectangle (13.5,4);
		\fill[fill=green, opacity=0.2] (8.5,4.5) rectangle (14,5);
		\fill[fill=green, opacity=0.2] (13,3) rectangle (13.5,3.5);
		\fill[fill=green, opacity=0.2] (8.5,4) rectangle (9,4.5);
		
		\end{tikzpicture}
		\caption{\small{Snake embedding from 2 to 3 dimensions.}}\label{fig:se}}
\end{figure}

\begin{figure}
	\center{
		\begin{tikzpicture}[scale=0.8]
		\tikzstyle{xxx}=[dashed,thick]
		
		\draw[thick](0,2)--(6,2)--(6,5.5)--(0,5.5)--(0,2);
		\draw[thick](0,2.5)--(5.5,2.5)--(5.5,4)--(1,4)--(1,4.5)--(6,4.5);
		\draw[thick](0,3)--(5,3)--(5,3.5)--(0.5,3.5)--(0.5,5)--(6,5);
		%\node at(5,5.25){$3$};
		%\node at(6,4.75){$-1,1,-2,2$};
		%\node at(0.25,4.25){$3$};
		%\node at(1.25,4.25){$-3$};
		%\node at(3,3.75){$-1,1,-2,2$};
		%\node at(4.75,3.25){$3$};\node at(5.75,3.25){$-3$};
		%\node at(3,2.75){$-1,1,-2,2$};
		%\node at(1,2.25){$-3$};
		\node at(5.25,2.75){\small{$1$}};
		\node at(5.275,3.25){\small{$2$}};
		\node at(5.15,3.7){\small{$-1$}};
		\node at(4.5,2.7){\small{$-2$}};
		\node at(4.5,3.7){\small{$-2$}};
		\node at(3,1){side view without duplication};
		
		\fill[fill=blue, opacity=0.2] (5,3) rectangle (5.5,3.5);
		\fill[fill=blue, opacity=0.2] (4.2,3.5) rectangle (5,4);

		\draw[thick](8,2)--(14,2)--(14,5.5)--(8,5.5)--(8,2);
		\draw[thick](8,2.5)--(13.5,2.5)--(13.5,4)--(9,4)--(9,4.5)--(14,4.5);
		\draw[thick](8,3)--(13,3)--(13,3.5)--(8.5,3.5)--(8.5,5)--(14,5);
		%\node at(13,5.25){$3$};
		%\node at(11,4.75){$-1,1,-2,2$};
		%\node at(8.25,4.25){$3$};\node at(9.25,4.25){$-3$};
		%\node at(11,3.75){$-1,1,-2,2$};
		%\node at(12.75,3.25){$3$};\node at(13.75,3.25){$-3$};
		%\node at(11,2.75){$-1,1,-2,2$};
		%\node at(9,2.25){$-3$};
		\node at(13.25,2.75){\small{$1$}};
		\node at(13.25,3.25){\small{$1$}};
		\node at(13.25,3.75){\small{$1$}};
		\node at(12.5,2.7){\small{$-2$}};
		\node at(12.65,3.75){\small{$2$}};
		
		\fill[fill=brown, opacity=0.2] (13,2.5) rectangle (13.5,4);

		\node at(11,1){side view with duplication};
		
		\end{tikzpicture}
		\caption{\small{Side views of the folding operation with and without duplications of cubelets. On the left, simply folding generates equal-and-opposite labels diagonally in the shaded cubelets. On the right, the duplication of the cubelet at the folding position in three copies prevents this from happening.}}\label{fig:se2}}
\end{figure}
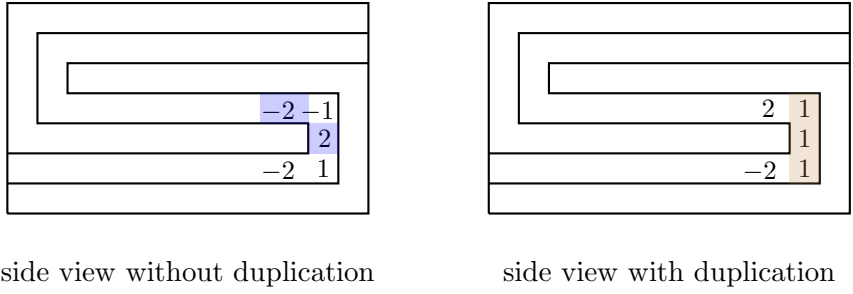

\begin{figure}
	\center{
		\begin{tikzpicture}[scale=0.8]
		\tikzstyle{xxx}=[dashed,thick]
		
		\draw[thick](0,5)--(5,5)--(6,7)--(1,7)--(0,5);

		\draw[thick](-1,5)--(6,5);
		\draw[thick](-0.75,5.5)--(6.25,5.5);
		\draw[thick](-0.5,6)--(6.5,6);
		\draw[thick](-0.25,6.5)--(6.75,6.5);
		\draw[thick](0,7)--(7,7);
		
		\draw[thick](0,7)--(-1,5);
		\draw[thick](7,7)--(6,5);
		
		\draw[thick](1,5)--(2,7);
		\draw[thick](2,5)--(3,7);
		\draw[thick](3,5)--(4,7);
		\draw[thick](4,5)--(5,7);
		
		\draw[thick, fill=green, opacity=0.2](-1,5)--(0,7) -- (1,7) -- (0,5) -- (-1,5);
		\draw[thick, fill=red, opacity=0.2](5,5)--(6,7) -- (7,7) -- (6,5) -- (5,5);
		
		\node at (0.7,5.25) {$2$};
		\node at (0.95,5.75) {$2$};
		\node at (1.2,6.25) {$1$};
		\node at (1.3,6.75) {$-2$};
		
		\node at (-0.3,5.25) {$2$};
		\node at (-0.05,5.75) {$2$};
		\node at (0.25,6.25) {$1$};
		\node at (0.4,6.75) {$-2$};
		
		\node at (4.6,5.25) {$2$};
		\node at (4.8,5.75) {$-1$};
		\node at (5.1,6.25) {$-2$};
		\node at (5.3,6.75) {$-2$};
		
		\node at (5.6,5.25) {$2$};
		\node at (5.8,5.75) {$-1$};
		\node at (6.1,6.25) {$-2$};
		\node at (6.3,6.75) {$-2$};

		\end{tikzpicture}
		\caption{\small{Extending the colouring to ensure that $m_1$ is a multiple of $3$. In the figure, the case when $m_1 \mod 3=1$ is shown, i.e. one layer needs to be added at each side.}}\label{fig:se3}}
\end{figure}
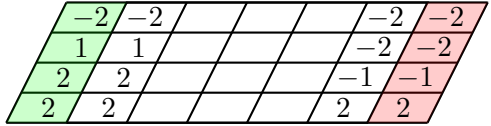

\begin{paragraph}{Formal description of snake embedding}
	Let $I$ be an instance of $k$D-\tucker\ having coordinates in ranges
	$[m_1],\ldots,[m_k]$ and label function $\lambda$.
	%We construct an instance $I'$ of $(k+1)$D-\tucker\ having
	%coordinates in ranges \aris{to figure this out}$[\lfloor m_1/3\rfloor+2],[m_2],\ldots,[m_k],[7]$
	%and label function $\lambda'$ as follows.
	Select the largest $m_i$, breaking ties lexicographically.
	Assume for simplicity in what follows that $m_1$ is largest.\\
	
	\noindent \textbf{Fixing the length to a multiple of $3$}. Let $r = m_1 \mod 3$ and let $\ell = 3-r$. Consider the instance $I^{3}$ of $kD$-\tucker\ having coordinates in ranges $[m_1'], \ldots, [m_k]$, with $m_1' \mod 3 = 0$, constructed from $I$ as follows. For any point $\xb'=(x_1',\ldots,x_k')$ in $[m_1'] \times \ldots \times [m_k]$, $\xb'$ is mapped to a point $\xb$ in $[m_1] \times \ldots \times [m_k]$ and receives a colour $\lambda'(\xb')$ such that, 
	\begin{itemize}
		\item if $\ell = 0$, then  $\xb'$ is mapped to $x=(x_1',\ldots,x_k')$ and $\lambda'(\xb') = \lambda(\xb)$, i.e. $\xb'$ is mapped to itself and receives its own label, since $m_1$ is already a multiple of $3$.
		\item If $\ell = 1$, then 
		\begin{itemize}
			\item if $x_1' \leq m_1$, $\xb'$ is mapped to $\xb=(x_1',\ldots,x_k')$ and $\lambda'(\xb') = \lambda(\xb)$.
			\item if $x_1' = m_1 + 1$, $\xb'$ is mapped to $\xb=(m_1,\ldots,x_k')$ and $\lambda'(\xb') = \lambda(\xb)$.
		\end{itemize}
		In other words, points for which the first coordinate ranges from $1$ to $x_1'$, are mapped to themselves and receive their own label, and points for which the first coordinate is $m_1+1$ are mapped to the points where the first coordinate is $m_1$, receiving the label of that point. This essentially ``duplicates'' the layer of cubelets on the right endpoint of the $m_1$-direction. See Figure \ref{fig:se3} for an illustration.
		\item If $\ell = 2$, then 
		\begin{itemize}
			\item if $x_1' =1 $, $\xb'$ is mapped to $\xb=(1,\ldots,x_k')$ and $\lambda'(\xb') = \lambda(\xb)$.
			\item if $2\leq x_1' \leq m_1+1$, $\xb'$ is mapped to $\xb=(x_1'-1,\ldots,x_k')$ and $\lambda'(\xb') = \lambda(\xb)$.
			This is similar to the mapping and labelling in the previous case, except for the fact that we need to ``shift'' the labels of the points, since we essentially introduced a copy of the layer of cubelets on the left endpoint of the $m_1$-direction. See Figure \ref{fig:se3} for an illustration.
		\end{itemize}
	\end{itemize} 
	Note that by the operation of adding $\ell$ layers as above, we do not introduce any cubelets with equal-and-opposite labels that were not present before. To avoid complicating the notation, in the following we will use $m_1$ to denote the maximum size of the first coordinate (instead of $m_1'$) and we will assume that $m_1$ is a multiple of $3$. We will also use $I$ to denote the instance of $kD$-\tucker\ where $m_1$ is a multiple of $3$, instead of $I^3$ as denoted above.\\

	\noindent \textbf{From $k$ to $k+1$ dimensions}. Starting from an instance $I$ of $kD$-\tucker, we will construct an instance $I'$ of $(k+1)D-$\tucker\ as follows. Let ${\bf x}=(x_1,\ldots,x_k)$ be a point in $[m_1]\times\ldots\times[m_k]$ with labelling function $\lambda$. We will associate each such point with a corresponding point $\xb'$ in $\left[\frac{m_1}{3}+2\right]\times\ldots\times[m_k] \times [7]$ and a label $\lambda'(\xb')$ as follows.
	\begin{itemize}
		\item
		If $x_1\leq \frac{m_1}{3}$, then ${\bf x}$ is mapped to
		${\bf x}'=(x_1,\ldots,x_k,2)$, and $\lambda'({\bf x}')=\lambda({\bf x})$.
		\item
		If $x_1=\frac{m_1}{3}+1$ \emph{(the first ``folding'' point)}, then ${\bf x}$ is mapped to the following three points in $I'$ and receives the following colours (see the shaded cubelets at the right-hand side of Figure \ref{fig:se2}):
		\begin{itemize}
			\item ${\bf x}'=(\frac{m_1}{3}+1,\ldots,x_k,2)$ \emph{(the original cubelet)} and $\lambda'({\bf x}')=\lambda({\bf x})$.
			\item ${\bf x}'=(\frac{m_1}{3}+1,\ldots,x_k,3)$ \emph{(the first copy)} and $\lambda'({\bf x}')=\lambda({\bf x})$.
			\item ${\bf x}'=(\frac{m_1}{3}+1,\ldots,x_k,4)$ \emph{(the second copy)} and $\lambda'({\bf x}')=\lambda({\bf x})$.
		\end{itemize} 
		\item If $\frac{m_1}{3}+2 \leq x_1 \leq \frac{2m_1}{3}-1$, then {\bf x} is mapped to ${\bf x}'=(\frac{2m_1}{3}+2-x_1,x_2,\ldots,x_k,4)$, with $\lambda'({\bf x}')=\lambda({\bf x})$.
		\item If $x_1=\frac{2m_1}{3}$ \emph{(the second ``folding'' point)}, then ${\bf x}$ is mapped to the following three points in $I'$ and receives the following colours:
		\begin{itemize}
			\item ${\bf x}'=(2,\ldots,x_k,4)$ \emph{(the original cubelet)} and $\lambda'({\bf x}')=\lambda({\bf x})$.
			\item ${\bf x}'=(2,\ldots,x_k,5)$ \emph{(the first copy)} and $\lambda'({\bf x}')=\lambda({\bf x})$.
			\item ${\bf x}'=(2,\ldots,x_k,6)$ \emph{(the second copy)} and $\lambda'({\bf x}')=\lambda({\bf x})$.
		\end{itemize} 
		\item
		If $\frac{2m_1}{3}+1 \leq x_1 \leq m_1$, then ${\bf x}$ is mapped to
		${\bf x}'=(x_1+2-\frac{2m_1}{3},x_2,\ldots,x_k,6)$, with $\lambda'({\bf x}')=\lambda({\bf x})$.
		\item
		Set $\lambda'(1,\ldots,1)=-k$, along with any point {\bf x} connected to it via a path
		of points that have not been labelled by the above procedure.
		\item
		Set $\lambda'(\frac{m_1}{3}+2,m_2,\ldots,m_k,7)=k$,
		along with any point {\bf x} connected to it via a path
		of points that have not been labelled by the above procedure.
	\end{itemize}
\end{paragraph}
We are now ready to prove Theorem \ref{thm:snake}.

\begin{proof} [Proof of Theorem~\ref{thm:snake}]
	First, it is not hard to check that the the composition of $O(n)$ snake-embeddings
	is a polynomial-time reduction. %We also need to argue that any solution to the high-dimensional instance
	%allows us to recover a solution to the original instance of $2D$-\tucker.
	Also note that, by the way
	the high-dimensional instances is constructed, we have not introduced any adjacencies that did not already exist,
	i.e. if there is a pair of adjacent cubelets with equal-and-opposite labels in the instance $I'$ of the high-dimensional version, this pair is present in the instance $I$ of the $2D$ version as well, and it is easy to recover it in polynomial time. Therefore, it suffices to show how to obey the additional constraint of \vhdt, namely that for $i\geq 2$, a side having
	label $i$ has no grid-points with label $i$, and similarly for $-i$.
	%First, we need to show that the composition of $O(n)$ snake-embeddings
	%is a polynomial-time reduction. Additionally, we need
	%to argue that besides being polynomial-time, any solution to
	%the high-dimensional instance allows us to recover a solution to
	%the original instance of $2D$-\tucker.
	%This can be argued by noting that from the way the high-dimensional
	%instance is constructed, there are no equal-and-opposite labels
	%adjacent to each other, in the range $\pm 3,\pm 4,\ldots$.
	
	%Then, need to reduce from this version of \tucker\ to the version
	%where we have a triangulation of $S^n$.
	%General idea is to take 2 copies of an instance of the above $n$D-\tucker,
	%negate the labels of one of them, flip it through its centroid,
	%and glue them together at the boundaries.
	%Then you have a ``flattened'' conventional \tucker-instance.
	%Triangulate by using copies of a standard triangulation of the hypercube,
	%then project to $S^n$.
	We begin as in \cite{FG17} (see Figure 1 in that paper), by taking the original $2D$ instance $I$, of size $m\times m$, and extend to an
	instance of size $3m\times m$ as follows. The original instance is embedded
	in the centre of the new instance. Each region $R$ to the sides (of size $m\times m$)
	are labelled by copying the edge of $I$ facing $R$, along an adjacent edge of $R$,
	and connecting these two edges with paths that have two straight sections and
	connect 2 points of the same label, and points along that path have that label.
	The outermost path then labels a side of the new instance of length $m$,
	so these two opposite sides get opposite labels.
	We may assume (by switching $1$'s and $2$'s if needed) that these new opposite sides are labelled $\pm 2$.
	
	The $S$-fold approach shown in Figure~\ref{fig:se} (in this paper) can be checked to retain this property.
	When we sandwich a cuboid between two layers of opposite (new) colours (call them $c$ and $-c$),
	as shown in Figure~\ref{fig:se}, we label the new facets thus formed with
	$-c$ and $c$ respectively. We label the other facets with their original labels
	(each of these facets has acquired the labels $c$ and $-c$, and no other labels).
	The folding operation has a natural correspondence between the facets of the
	unfolded and folded versions of the cuboid. It can be checked that the set
	of colours of a facet before folding is the same as the set of colours of the
	corresponding facets after folding.
\end{proof}
It is convenient to define the following problem, whose \ppa-completeness follows
fairly directly from the \ppa-completeness of \vhdt.

\begin{definition}\label{def:nvhdt}
	An instance of \nvhdt\ in $n$ dimensions is presented by a boolean circuit $C_{VT}$ that
	takes as input coordinates of a point in the hypercube $B=[-1,1]^n$ and outputs a
	label in $\pm[n]$ (assume $C_{VT}$ has $2n$ output gates,
	one for each label, and is syntactically constrained such that exactly one
	output gate will evaluate to \true), having the following constraints that may be enforced syntactically.
	\begin{enumerate}
		\item Dividing $B$ into $7^n$ cubelets of edge length $2/7$ using axis-aligned
		hyperplanes, all points in the same cubelet get the same label by $C_{VT}$;
		\item Interiors of antipodal boundary cubelets get opposite labels;
		\item Points on the boundary of two or more cubelets get a label belonging
		to one of the adjacent cubelets;
		\item Facets of $B$ are coloured with labels from $\pm[n]$ such
		that all colours are used, and opposite facets have opposite labels.
		For $2\leq i\leq n$ it also holds that the facet with label $i$ (resp. $-i$)
		does not intersect any cubelet having label $i$ (resp. $-i$).
		Facets coloured $\pm 1$ are unrestricted (we call them the ``panchromatic facets'').
	\end{enumerate}
	A solution consists of a polynomial number $p^C$ of points that all lie within an inverse polynomial
	distance $\delta(n)$ of each other (for concreteness, assume $\delta(n)=1/100n$).
	At least two of those points should receive equal and opposite labels by $C_{VT}$.
\end{definition}
\nvhdt\ corresponds to the problem {\sc Variant Tucker} in~\cite{FG17}; in that
paper a solution only contained 100 points, while here we use $p^C$ points.
Here we need more points since we are in $n$ dimensions,
and our analysis needs to tolerate $n$ points receiving unreliable labels.

\section{Some building-blocks and definitions}\label{sec:bb}

Here we set up some of the general structure of instances of \ch\ constructed in our reduction.
We identify some basic properties of solutions to these instances.
We define the {\em M\"{o}bius-simplex} and the manner in which a solution
encodes a point on the M\"{o}bius-simplex.
The encoding of the circuitry is covered in Section~\ref{sec:reduction}.

\begin{paragraph}{Useful quantities:} We use the following values throughout the paper.
	\begin{itemize}
		\item $\delt$ is an inverse-polynomial quantity in $n$, chosen to be substantially smaller than
		any other inverse-polynomial quantity that we use in the reduction, apart from $\epsilon$ (below).
		\item $\delta^T$ is an inverse-polynomial quantity in $n$, which is smaller than any other inverse-polynomial quantity apart from $\delt$ and is larger than $\delt$ by an inverse-polynomial amount. The quantity $\delta^T$ denotes the width of the so-called ``twisted tunnel'' (see Definition~\ref{def:twisted}).
		\item $\pgig$ denotes a large polynomial in $n$; specifically we let $\pgig=n/\delt$. The quantity $\pgig$ represents the number of sensor agents for each circuit encoder (see Definition \ref{def:sensors}).
		\item $\pmeg$ denotes a large polynomial in $n$, which is however smaller than $\pgig$ by a polynomial factor. The quantity $\pmeg$ will be used in the definition of the ``blanket-sensor agents'' (see Definition~\ref{def:blanket}) and will quantify the extent to which the cuts in the ``coordinate-encoding region'' (Definition~\ref{def:c-e-region}) are allowed to differ from being evenly spaced, before the blanket-sensor agents become active (see Section~\ref{sec:bb}). The choice of $\pmeg$ controls the value $\delta_w$ of the radius of the Significant Region (see Proposition~\ref{prop:s-r-radius}), with larger $\pmeg$ meaning larger $\delta_w$.
		\item $\varepsilon$ is the precision parameter in the Consensus-Halving solution, i.e. each agent $i$ is satisfied with a partition as long as $|\mu_i(\lplus)-\mu_i(\lminus)|\leq\varepsilon$. Henceforth, we will set $\epsilon = \delt/10$. 
	\end{itemize}
\end{paragraph}

\subsection{Basic building-blocks}\label{sec:bbb}

We consider instances $I_{CH}$ of \ch\ that have been derived from instances
$I_{VT}$ of \nvhdt\ in $n$ dimensions.
The general aim is to get any solution of such an instance $I_{CH}$ to encode a
point in $n$ dimensions that ``localises'' a solution to $I_{VT}$, by which we mean that 
from the solution of $I_{CH}$, we will be able to find a point on the $I_{VT}$ instance that can be 
transformed to a solution of $I_{VT}$ in polynomial time and fairly straightforwardly.

\begin{definition}\label{def:c-e-region}
{\bf Coordinate-encoding region (c-e region)}
Given an instance of \vhdt\ in $n$ dimensions, the corresponding instance of \ch\ has
a {\em coordinate-encoding region}, the interval $[0,n]$,
a (prefix) subinterval of $A$.
\end{definition}
The valuation functions of agents in an instance $I_{CH}$ of \ch\ obtained by our reduction from
an instance of \nvhdt\ in $n$ dimensions, will be designed in such a way that either
$n-1$ or $n$ cuts (typically $n$) must occur in the coordinate-encoding region, in any solution.
Furthermore, the distance between consecutive cuts must be close to 1
(an additive difference from 1 that is upper-bounded by an inverse polynomial),
shown in Proposition~\ref{prop:s-r-radius}.

\begin{definition}\label{def:c-e-agents}
{\bf Coordinate-encoding agents (c-e agents).}
Given an instance of \nvhdt\ in $n$ dimensions, the corresponding instance of \ch\ has
$n$ {\em coordinate-encoding agents} denoted $\{a_1,\ldots,a_n\}$.
\end{definition}
The $n$ c-e agents have associated $n$ coordinate-encoding cuts (Definition~\ref{def:c-e-cuts}).
It will be seen that the c-e cuts typically occur in the c-e region.
The c-e agents do {\em not} have any value for the coordinate-encoding
region; their value functions are only ever positive elsewhere.
In particular, they have blocks of value whose labels $\lplus/\lminus$ are affected 
by the output gates of the circuitry that is encoded to the right of the c-e region.

\begin{definition}\label{def:c-e-cuts}
{\bf Coordinate-encoding cuts (c-e cuts).}
We identify $n$ cuts as the {\em coordinate-encoding cuts}.
In the instances of \ch\ that we construct, in any (sufficiently good approximate) solution
to the \ch\ instance, all other cuts will be constrained
to lie outside the c-e region (and it will be straightforward to see that
the value functions of their associated agents impose this constraint).
A c-e cut is not straightforwardly constrained to lie in the c-e region,
but it will ultimately be proved that in any approximate solution,
the c-e cuts do in fact lie in the c-e region.
\end{definition}
%For a polynomial $\pgig(n)$ we divide the c-e region into $\pgig(n)$ equal-length intervals of length
%$\deln:=n/\pgig(n)$; denote these intervals $c_1,\ldots,c_{\pgig}$.
Recall that $\pgig= n / \delt$ from Section \ref{sec:probres}, which implies that the c-e region can be divided into $\pgig$ intervals of length $\delt$ (see also Figure \ref{fig:circuit1}).

\begin{definition}\label{def:shiftedversion}
	{\bf $\sigma$-shifted version.}
Given a value function $f$ (or measure) on the c-e region $[0,n]$, we say that another function $f'$
on the c-e region is a \emph{$\sigma$-shifted version} of $f$, when we have that $f'((x-\sigma)\mod n) = f(x)$.
\end{definition}

Recall that the circuit-encoding region (details in Section~\ref{sec:reduction}) contains $p^C$
{\em circuit-encoders}, mentioned in the following definitions.

\begin{definition}\label{def:sensors}
{\bf Sensor agents.}
Each circuit-encoder $C_i$, $i=1,\ldots,p^C$, has a set $\sensor_i$ of sensor agents.
$\sensor_i = \{s_{i,1},\ldots,s_{i,\pgig}\}$ where the $s_{i,j}$ are defined as follows.
When $i=1$, $s_{1,j}$ has value $\frac{1}{10}$ uniformly distributed over the interval
\[
\Bigl[(j-1)\delt,(j-1)\delt + \frac{\delt}{p^C} \Bigr].
\]
For $i>1$, $s_{i,j}$ is a $\frac{1}{p^C}(i-1)\delt$-shifted version of $s_{1,j}$.\\

\noindent Each sensor agent $s_{i,j}$ also has valuation outside the c-e region, in non-overlapping
intervals of the circuit-encoding region $R_i$ (see Section \ref{sec:forward}). This valuation consists of two valuation blocks of value $\frac{9}{20}$ each, with no other valuation block in between. These are exactly as described in~\cite{FG17}%(e.g. see the value of the first sensor agent $\alpha_{11}$ outside the c-e region in Figure 8 of \cite{FG17})
, see also Appendix \ref{sec:app_details} and Figure \ref{fig:sensor_example} for an illustration.

This value gadget for $s_{i,j}$ causes the $j$-th input gate in the circuit-encoder $C_i$ to be
set according to the label received by $s_i$'s block of value in the c-e region, i.e. jump to the left or to the right in order to indicate that the corresponding value-block of $s_i$ in the c-e region is labelled $\lplus$ or $\lminus$.
\end{definition}

According to the definitions above, $C_1$ has a sequence of (a large polynomial number of) sensor agents that
have blocks of value in a sequence of small intervals going from left to right of
the c-e region (see Figure \ref{fig:circuit1}).
For $1<i\leq p^C$, $C_i$ has a similar sequence, shifted slightly to the right on the c-e region (by $\delt (i-1)/p^{C}$).
For $j\in[\pgig]$, the intervals defined by the value-blocks of the sensor agents $s_{1,j}, \ldots, s_{p^C,j}$ (for $C_1, \ldots, C_{p^C}$) partition the interval $[(j-1)\delt,j\delt]$.\\

\noindent \textbf{Remark:} Note that a c-e cut may divide one of the above value-blocks held by a sensor
agent in the c-e region, and in that case the input being supplied
(using the gadget of \cite{FG17}) to its
circuit-encoder is \emph{unreliable}. However, only $n$ sensor agents may be affected in
that way, and their circuit-encoders will get ``out-voted'' by the ones that receive
valid boolean inputs. This is part of the reason why we use $p^C$ circuit-encoders in total.
More details on this averaging argument are provided in Section~\ref{sec:reduction}.

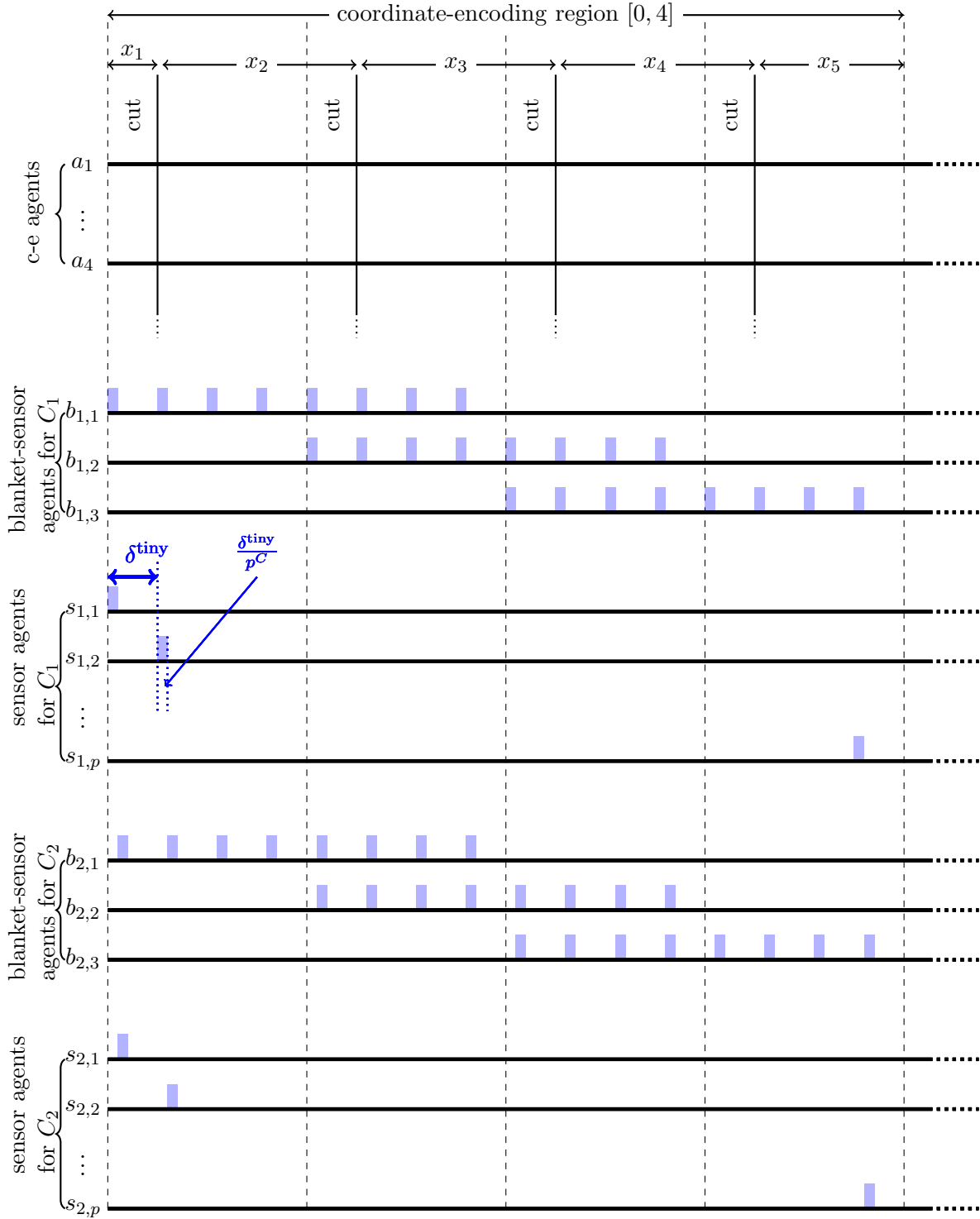
\begin{figure}
\center{
\begin{tikzpicture}[scale=0.8]
\tikzstyle{agstyle}=[ultra thick]
\tikzstyle{darrstyle}=[thick,<->,shorten >=0.1cm, shorten <=0.1cm]
\tikzstyle{cutstyle}=[black,thick]%[ultra thick,violet!60]
\tikzstyle{valuestyle}=[fill=blue,opacity=0.3]

%\draw[cyan!20] (-2,-1) grid (17,23);
\foreach \x in {0,4,8,12,16}{
\draw[dashed](\x,-1)--(\x,23);
}

\draw[thick,<->](0,23)--(16,23);\draw[thick,<->](0,22)--(1,22);
\draw[thick,<->](1.1,22)--(5,22);\draw[thick,<->](5.1,22)--(9,22);
\draw[thick,<->](9.1,22)--(13,22);\draw[thick,<->](13.1,22)--(16,22);

\node[fill=white]at(8,23){coordinate-encoding region $[0,4]$};
\node at(0.5,22.25){$x_1$};
\node[fill=white]at(3,22){$x_2$};
\node[fill=white]at(7,22){$x_3$};
\node[fill=white]at(11,22){$x_4$};
\node[fill=white]at(14.5,22){$x_5$};

\draw [thick,decorate,decoration={brace,amplitude=4pt},xshift=4pt,yshift=0pt]
      (-1,18)--(-1,20) node [midway,left,xshift=-0.5cm,yshift=1cm,rotate=90] {c-e agents}; 

\draw [thick,decorate,decoration={brace,amplitude=4pt},xshift=4pt,yshift=0pt]
      (-1,13)--(-1,15) node [midway,left,xshift=-0.5cm,yshift=2cm,rotate=90,text width=3cm]
      {blanket-sensor agents for $C_1$}; 

\draw [thick,decorate,decoration={brace,amplitude=4pt},xshift=4pt,yshift=0pt]
      (-1,8)--(-1,11) node [midway,left,xshift=-0.5cm,yshift=2cm,rotate=90,text width=2.5cm]
      {sensor agents for $C_1$}; 

%value-blocks, starting with ones on blanket-agents
\foreach \y in {0,1,2}{
  \foreach \x in {0,1,...,7}{
    \fill[valuestyle,xshift=\x cm,yshift=-1*\y cm]
    (4*\y,15)--(4*\y+0.2,15)--(4*\y+0.2,15.5)--(4*\y,15.5)--cycle;}}
\fill[valuestyle](0,11)--(0.2,11)--(0.2,11.5)--(0,11.5)--cycle;
\fill[valuestyle](1,10)--(1.2,10)--(1.2,10.5)--(1,10.5)--cycle;
\fill[valuestyle](15,8)--(15.2,8)--(15.2,8.5)--(15,8.5)--cycle;

\node at(-0.5,20){$a_1$};\node at(-0.5,19){$\vdots$};\node at(-0.5,18){$a_4$};
\node at(-0.5,15){$b_{1,1}$};\node at(-0.5,14){$b_{1,2}$};\node at(-0.5,13){$b_{1,3}$};
\node at(-0.5,11){$s_{1,1}$};\node at(-0.5,10){$s_{1,2}$};\node at(-0.5,9){$\vdots$};
\node at(-0.5,8){$s_{1,p}$};

%%%%%%%%%% modified copy for C_2 - too complicated to do it properly

\draw [thick,decorate,decoration={brace,amplitude=4pt},xshift=4pt,yshift=0pt]
      (-1,4)--(-1,6) node [midway,left,xshift=-0.5cm,yshift=2cm,rotate=90,text width=3cm]
      {blanket-sensor agents for $C_2$}; 

\draw [thick,decorate,decoration={brace,amplitude=4pt},xshift=4pt,yshift=0pt]
      (-1,-1)--(-1,2) node [midway,left,xshift=-0.5cm,yshift=2cm,rotate=90,text width=2.5cm]
      {sensor agents for $C_2$};

%value-blocks, starting with ones on blanket-agents
\foreach \y in {0,1,2}{
  \foreach \x in {0,1,...,7}{
    \fill[valuestyle,xshift=\x cm,yshift=-1*\y cm]
    (4*\y+0.2,6)--(4*\y+0.4,6)--(4*\y+0.4,6.5)--(4*\y+0.2,6.5)--cycle;}}
\fill[valuestyle,xshift=0.2cm](0,2)--(0.2,2)--(0.2,2.5)--(0,2.5)--cycle;
\fill[valuestyle,xshift=0.2cm](1,1)--(1.2,1)--(1.2,1.5)--(1,1.5)--cycle;
\fill[valuestyle,xshift=0.2cm](15,-1)--(15.2,-1)--(15.2,-0.5)--(15,-0.5)--cycle;

\node at(-0.5,6){$b_{2,1}$};\node at(-0.5,5){$b_{2,2}$};\node at(-0.5,4){$b_{2,3}$};
\node at(-0.5,2){$s_{2,1}$};\node at(-0.5,1){$s_{2,2}$};\node at(-0.5,0){$\vdots$};
\node at(-0.5,-1){$s_{2,p}$};

%%%%%%%%%%%%%%%%%%%%%%%% end modified copy for C_2

%%%% draw cuts:
\foreach \x in {1,5,9,13}{
\draw[cutstyle](\x,17)--(\x,21.8);
\draw[cutstyle,dotted](\x,17)--(\x,16.5);

\node[rotate=90]at(\x-0.5,21){cut};

%%%% draw all agent lines:
\foreach \y in {20,18,15,14,13,11,10,8,6,5,4,2,1,-1}{
  \draw[agstyle](0,\y)--(16.5,\y);\draw[agstyle,dotted](16.5,\y)--(17.5,\y);}

% delta-nano
\draw[ultra thick,blue,<->](0,11.7)--(1,11.7);
\draw[dotted,thick,blue](1,12)--(1,9);\node[blue] at(0.8,12.2){$\delt$};
\draw[dotted,thick,blue](1.2,10.5)--(1.2,9);\draw[thick,blue,->](3,11.7)--(1.15,9.5);
\node[blue] at(3,12.2){$\frac{\delt}{p^C}$};

}

\end{tikzpicture}
\caption{\small{Sensor illustration: example of $n=4$ c-e cuts representing
5 coordinates summing to 1 (a typical point in the M\"{o}bius-simplex).
Vertical lines depict the cuts, resulting in labels that alternate
between $\lplus$ and $\lminus$, starting with $\lplus$.
Shaded blocks over agents' lines indicate value-blocks of their value functions.
We only depict sensors for circuit-encoders $C_1$ and $C_2$.}}\label{fig:circuit1}}
\end{figure}

\begin{definition}\label{def:blanket}
{\bf Blanket-sensor agents.}
Each circuit-encoder $C_i$ shall have $n-1$ {\em blanket-sensor agents} $b_{i,2},\ldots,b_{i,n}$.
\begin{enumerate}
\item
In $C_1$, for each $j=2,\ldots,n$, blanket-sensor agent $b_{1,j}$ has value $1/10$ distributed over $[j-2,j]$
(see Figure~\ref{fig:circuit1}). This value consists of a sequence $$\left[(j-2),(j-2)+\frac{\delt}{p^C}\right],\cdots,\left[j-\delt,j-\delt+\frac{\delt}{p^C}\right]$$ of $2/\delt=2\pgig/n$ value-blocks, each of length
$\delt/p^C$ and of value $\frac{1}{10}\cdot(\delt/2)$.
\item In each $C_i$, $1<i\leq p^C$, and for each $j=2,\ldots,n$, the value function of $b_{i,j}$ that
lies in the c-e region is an $(i-1)\frac{\delt}{p^C}$-shifted version of $b_{1,j}$.
\item
The remaining value $9/10$ of each $b_{i,j}$ consists of 3 value-blocks
of width $\delt$ lying in a subinterval $I_{i,j}$ of the circuit-encoding region $R_i$ (see Section \ref{sec:forward}), such that:
\begin{itemize}
	\item[-] the value-blocks have values
	\[ \frac{9(1-\kappa)}{20}, \ \ \ \ 
	\frac{9\kappa}{10}, \ \ \ \ 
\frac{9(1-\kappa)}{20}
	\]
	respectively, where $\kappa = \Bigl(\frac{1}{10} \frac{\delt}{2}\Bigr)\pmeg$.
	\item[-] $I_{i,j}$ contains also value-blocks of agents for each gate that takes the value of $b_{i,j}$ as input (the feedback gadgetry, see Section \ref{sec:cea}).
\end{itemize}
\end{enumerate}
\end{definition}
The value of the blanket-sensor agents in $I_{i,j}$ is very similar to the gadget used in \cite{FG17}, %see for example the design of the blanket-sensor agent in Figure 7 of~\cite{FG17}. 
see Appendix \ref{sec:app_details} (of the present paper) and Figure \ref{fig:blanket_example}.
The structure of the blanket-sensor agents in the c-e region is shown in Figure~\ref{fig:circuit1}.

\subsubsection*{Notes on the blanket-sensors}
Each blanket-sensor agent $b_{i,j}$ has an associated cut $c(b_{i,j})$ that lies in the subinterval $I_{i,j}$. Agent $b_{i,j}$ ``monitors'' an interval of length 2, namely the interval $[j-2,j]$
within which the sequence of $2/\delt$ value-blocks lie.
If, in this interval, the number of these value-blocks labelled $\lplus$ exceeds
the number labelled $\lminus$ by at least $\pmeg$
(recall that $\pmeg$ is a large polynomial which is however polynomially smaller than $\pgig$) then (in any $\varepsilon$-approximate solution
to $I_{CH}$, where, recall, $\varepsilon=\delt/10$), the cut $c(b_{i,j})$ in $I_{i,j}$ lies in either the right-hand or the left-hand value-block,
otherwise it lies in the central value-block.
Note that these three possible positions may be converted to boolean values that influence circuit-encoder $C_i$; this was referred to as a ``bit-detection gadget'' in \cite{FG17}, see Appendix \ref{sec:app_details} for more details.

\begin{definition}[Active blanket-sensor]\label{def:bsactive}
We say that blanket-sensor $b_{i,j}$ is {\em active} if $b_{i,j}$ in fact observes a sufficiently large label discrepancy
in the c-e region, that $c(b_{i,j})$ lies in one of the two outer positions, left or right, and not in the central position. We say that $b_{i,j}$ is {\em active towards $\lplus$} if $\lplus$ is the overrepresented label,
with similar terminology for $\lminus$.

When blanket-sensor agent $b_{i,j}$ is active, it provides input to $C_i$ that causes
$C_i$ to label the value held by $a_j$ and controlled by $C_i$, to be either $\lplus$ or $\lminus$; the choice depends on 
the over-represented label in $[j-2,j]$ and the parity of the index of the blanket-sensor agent. The precise feedback mechanism to the c-e agent $a_j$ by the blanket-sensor $b_{i,j}$ is described in Section~\ref{sec:bsa}.

\end{definition}
When no blanket-sensors are active, the sequence of c-e cuts encodes a point in the
Significant Region (Definition \ref{def:sig-region}).

\bigskip
\noindent In \cite{FG17}, we worked just in two dimensions and there was just one
blanket-sensor agent for the entire c-e region, for each circuit-encoder.
Note also that there, the blanket-sensor agent had a single value-block
of length 2; here we split it into a polynomial sequence of small
value-blocks. The advantage of using a polynomial sequence of
value-blocks (which could not have been done in \cite{FG17} due to the
exponential precision requirement) is that we can argue that in all but
at most $n$ circuit-encoders, the blanket-sensor agents have value-blocks
that are not cut by the c-e cuts, so we can be precise about how big
a disparity between blocks labelled $\lplus$ and $\lminus$ cause
a blanket-sensor to be active, and for at most $n$ circuit-encoders,
we regard them as having unreliable inputs
(see Definition~\ref{def:reliable} and Observation~\ref{obs:reliable}).

\subsection{Features of solutions}

The main result of this section is Proposition~\ref{prop:s-r-radius},
that in a solution to approximate \ch\ as constructed here,
the sequence of cuts in the c-e region are ``evenly spaced'' in the sense that
the gap between consecutive cuts differs from 1 by at most an inverse-polynomial.

\begin{observation}[At most $n$ cuts in the c-e region]\label{obs:cer}
Given an instance $I_{CH}$ derived by our reduction from an instance of \nvhdt\ in
$n$ dimensions, any inverse-polynomial approximate solution of $I_{CH}$ has the property that
at most $n$ cuts lie in the coordinate-encoding region.
This is because all other cuts are associated with agents who have
at least $9/10$ of their value strictly to the right of the c-e region,
thus in a solution, those cuts cannot lie in the c-e region.
\end{observation}

\begin{definition}[Reliable input]\label{def:reliable}
We will say that a circuit-encoder {\em receives reliable input}, if no coordinate-encoding cut passes through value-blocks of its sensor agents.
\end{definition}

\begin{observation}\label{obs:reliable}
At most $n$ circuit-encoders fail to receive reliable input (by Observation~\ref{obs:cer} and the fact that sensors of distinct circuit-encoders have value in distinct intervals).

When a circuit-encoder receives reliable input, it is straightforward to interpret
the labels allocated to its sensors, as boolean values, and simulate a
circuit computation on
those values, ultimately passing feedback to the c-e agents via value-blocks
that get labelled according to the output gates of the circuit being simulated.
This is done in a conceptually similar way to that described in \cite{FG17} (e.g. see Sections 4.4.2 and 4.6 in \cite{FG17}), see also Appendix \ref{sec:app_details} of the present paper.
\end{observation}

\begin{definition}\label{def:simplex-domain}
{\bf The M\"{o}bius-simplex.}
The {\em M\"{o}bius-simplex in $n$ dimensions} consists of points $\xb$ in $\rset^{n+1}$
whose coordinates are non-negative and sum to 1.
We identify every point $(x_1,\ldots,x_n,0)$ with the point $(0,x_1,\ldots,x_n)$,
for all non-negative $x_1,\ldots,x_n$ summing to 1.
We use the following metric $d(\cdot,\cdot)$ on the M\"{o}bius-simplex,
letting $L_1$ be the standard $L_1$ distance on vectors:
\begin{equation}\label{eq:metric}
d(\xb,\xb')=\min \Bigl( L_1(\xb,\xb'),
\min_{\zb,\zb' :  \zb\equiv \zb'}(L_1(\xb,\zb)+L_1(\zb',\xb')) \Bigr)
\end{equation}
where $(0,x_1,\ldots,x_{n})\equiv(x_1,\ldots,x_n,0)$.
\end{definition}

\begin{paragraph}{How a consensus-halving solution encodes a point in the M\"{o}bius-simplex}
Let $I_{CH}$ be an instance of \ch, obtained by reduction from \nvhdt\ in $n$ dimensions,
hence having c-e region $[0,n]$. Note that, by Observation~\ref{obs:cer}, at most $n$ cuts
may lie in the c-e region.
A set of $k\leq n$ cuts of the coordinate-encoding region splits it into $k+1$ pieces.
We associate such a split with a point $\xb$ in $\rset^{n+1}$ as follows.
The first coordinate is the distance from the LHS of the consensus-halving domain
to the first cut, divided by $n$, the length of the c-e region.
For $2\leq i\leq k+1$, the $i$-th coordinate of $\xb$ is the distance between
the $i-1$-st and $i$-th cuts, divided by $n$. Remaining coordinates are 0.

If there are $n-1$ cuts in the c-e region, suppose we add a cut at either the LHS or the RHS.
These two alternative choices correspond to a pair of points that have been identified as the same point,
as described in Definition~\ref{def:simplex-domain}.
(Observation~\ref{obs:flip} makes a similar point regarding transformed coordinates.)
\end{paragraph}

\begin{observation}\label{obs:average}
Each circuit-encoder reads in ``input'' representing a point in the M\"{o}bius-simplex.
Any circuit-encoder $C_i$ ($i\in[p^C]$) behaves like $C_1$ on a point $\xb_i$,
for which (for all $i,j\in[p^C]$) $d(\xb_i,\xb_j)\leq \delt$ (recall $d$ is defined in (\ref{eq:metric})).
Consequently their collective output (the split between $\lplus$ and $\lminus$ of the
value held by the c-e agents) is the output of a single circuit-encoder averaged over
a collection of $p^C$ points in the M\"{o}bius-simplex, all within $\delt$ of each other.

This follows by inspection of the way the $p^C$ circuit-encoders differ from each other:
their sensor-agents are shifted but their internal circuitry is the same.
\end{observation}

\begin{definition}\label{def:sig-region}
{\bf The Significant Region of the M\"{o}bius-simplex $D$.}
The {\em Significant Region} of $D$ consists of all points in $D$ where no blanket-sensors are active
(where ``blanket-sensors'' and ``active'' are defined in Definition \ref{def:bsactive}).
\end{definition}

\begin{proposition}\label{prop:s-r-radius}
There is an inverse-polynomial value $\delta^w$ such that all points $\xb = (x_1,\ldots,x_{n+1})$ in the Significant Region
have coordinates $x_i$ that for $2\leq i\leq n$ differ from $1/n$ by at most $\delta^w$,
if $\xb$ is encoded by the c-e cuts of an $\varepsilon$-approximate solution to one of our
instances of \ch. (Recall that $\varepsilon = \delt/10$).

Thus, if an instance $I_{CH}$ of \ch\ (obtained using our reduction)
has a solution $S_{CH}$, then all the c-e cuts in $S_{CH}$ have the property that the distance between
two consecutive c-e cuts differs from 1 by at most some inverse-polynomial amount.
\end{proposition}
Before we proceed with the proof of the proposition, we will state a few simple lemmas that will be used throughout the proof. We start with the following definition.

\begin{definition}[Cut $\delta$-close to integer point]
	For $\ell \in \{0,\ldots,n\}$, we will say that a cut $c$ is $\delta$-\emph{close to integer point} $\ell$, if it lies in $[\ell-\delta,\ell+\delta]$. We will say that cut $c$ is \emph{$\delta$-close to an integer point} if there is some integer $\ell \in \{0,\ldots,n\}$ such that $c$ is $\delta$-close to integer point $\ell$.
\end{definition}
Intuitively, cuts that are $\delta$-close to integer points lie close
(within distance at most $\delta$) to either the endpoints or the midpoint of
some monitored interval $[j-2,j]$.

Note that, by Definition \ref{def:blanket}, a blanket-sensor agent will be active when at least $\pgig/n + \pmeg$ value-blocks of volume $\frac{1}{10}\cdot \frac{\delt}{2}$ in an interval monitored by the blanket-sensor agent receive the same label. This will happen if there is a union $\bigcup_{\ell}I_{\ell}$ of subintervals $I_\ell$ of some monitored subinterval $[j-2,j]$, for some $j\in \{2,\ldots,n\}$, which will have total length larger than $1+\delta$, where $\delta > \pmeg \cdot \delt$. This means that in $[j-2,j]$, there will be at least $\pgig/n + \pmeg$ value-blocks of volume $\frac{1}{10}\cdot \frac{\delt}{2}$ that receive the same label, since by construction, there are at least $\pgig/n + \pmeg$ such value-blocks in any interval of length at least $1+\pmeg\cdot \delt$. In such a case, the blanket-sensor agent $b_{1,j}$ is active and the set of cuts is not a solution to $I_{CH}$. In the following, we will consider $\delta$ such that $\pmeg \cdot \delt < \delta < 1/2n$.

\begin{definition}[Monochromatic interval of label $A_j$]
An interval $I$ is a \emph{monochromatic interval} if it is not intersected by any cuts (which means that it receives a single label). Additionally, if for $A_j \in \{\lplus,\lminus\}$, $I$ is labelled with $A_j$, then we will say that $I$ is a \emph{monochromatic interval of label $A_j$}.
\end{definition}
It should be clear that if any monitored interval $[j-2,j]$ has a large enough (larger than $1+\delta$) monochromatic subinterval, then the blanket-sensor agent $b_{1,j}$ is active. 

\begin{lemma}\label{lem:44twomissingcuts}
For some $\ell \in \{0,n-k\}$,with $k>1$, consider the interval $I=[\ell, \ell+k]$ of length $k$ and assume that there are at most $k-2$ cuts in this interval.
Then at least one of the blanket-sensors monitoring the subintervals in $I$ will be active.
\end{lemma}

\begin{proof}
	In $I$, there are at least $k-1$ intervals monitored by blanket-sensor agents and we only have at most $k-2$ cuts at our disposal. With $k-2$ cuts, we can partition an interval of length $k$ in at most $k-1$ intervals, the largest of which, call it $I_\text{max}$, will have length at least $1+1/k$. Since $\delta < 1/2n$, the length of $I_{max}$ is actually larger than $1+\delta$. The lemma follows then from the fact that, since the monitored intervals partition the interval $I$, $I_{max}$ will contain a monochromatic interval of length at least $1+\delta$, which will be entirely contained in some monitored interval, and the corresponding blanket-sensor agent will be active.
\end{proof}

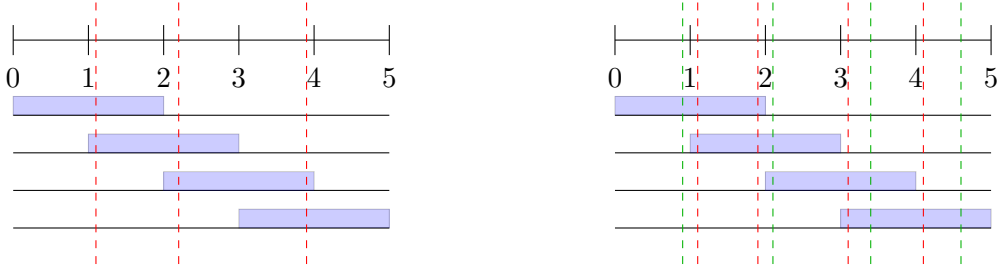
\begin{figure}
	\centering
	\begin{tikzpicture}
	% Straight line and integer points
	
	\draw (0,0) -- (5,0);
	\draw (0,0.2) -- (0,-0.2);
	\draw (1,0.2) -- (1,-0.2);
	\draw (2,0.2) -- (2,-0.2);
	\draw (3,0.2) -- (3,-0.2);
	\draw (4,0.2) -- (4,-0.2);
	\draw (5,0.2) -- (5,-0.2);
	
	\draw (8,0) -- (13,0);
	\draw (8,0.2) -- (8,-0.2);
	\draw (9,0.2) -- (9,-0.2);
	\draw (10,0.2) -- (10,-0.2);
	\draw (11,0.2) -- (11,-0.2);
	\draw (12,0.2) -- (12,-0.2);
	\draw (13,0.2) -- (13,-0.2);

	\node at (0,-0.5) {$0$};
	\node at (1,-0.5) {$1$};
	\node at (2,-0.5) {$2$};
	\node at (3,-0.5) {$3$};
	\node at (4,-0.5) {$4$};
	\node at (5,-0.5) {$5$};
	
	\node at (8,-0.5) {$0$};
	\node at (9,-0.5) {$1$};
	\node at (10,-0.5) {$2$};
	\node at (11,-0.5) {$3$};
	\node at (12,-0.5) {$4$};
	\node at (13,-0.5) {$5$};
	
	% Blanket-sensor intervals
	
	\draw (0,-1) -- (5,-1);
	\draw (0,-1.5) -- (5,-1.5);
	\draw (0,-2) -- (5,-2);
	\draw (0,-2.5) -- (5,-2.5);
	
	\draw (8,-1) -- (13,-1);
	\draw (8,-1.5) -- (13,-1.5);
	\draw (8,-2) -- (13,-2);
	\draw (8,-2.5) -- (13,-2.5);
	
	\draw[fill=blue,opacity=0.2] (0, -1) rectangle (2,-0.75);
	\draw[fill=blue,opacity=0.2] (1, -1.5) rectangle (3,-1.25);
	\draw[fill=blue,opacity=0.2] (2, -2) rectangle (4,-1.75);
	\draw[fill=blue,opacity=0.2] (3, -2.5) rectangle (5,-2.25);
	
	\draw[fill=blue,opacity=0.2] (8, -1) rectangle (10,-0.75);
	\draw[fill=blue,opacity=0.2] (9, -1.5) rectangle (11,-1.25);
	\draw[fill=blue,opacity=0.2] (10, -2) rectangle (12,-1.75);
	\draw[fill=blue,opacity=0.2] (11, -2.5) rectangle (13,-2.25);

	\draw[dashed,color=red] (1.1,0.5) -- (1.1,-3);
	\draw[dashed,color=red] (2.2,0.5) -- (2.2,-3);
	\draw[dashed,color=red] (3.9,0.5) -- (3.9,-3);
	
	\draw[dashed,color=red] (9.1,0.5) -- (9.1,-3);
	\draw[dashed,color=red] (9.9,0.5) -- (9.9,-3);
	\draw[dashed,color=red] (11.1,0.5) -- (11.1,-3);
	\draw[dashed,color=red] (12.1,0.5) -- (12.1,-3);
	
	\draw[dashed,color=darkgreen] (8.9,0.5) -- (8.9,-3);
	\draw[dashed,color=darkgreen] (10.1,0.5) -- (10.1,-3);
	\draw[dashed,color=darkgreen] (11.4,0.5) -- (11.4,-3);
	\draw[dashed,color=darkgreen] (12.6,0.5) -- (12.6,-3);
	
	\end{tikzpicture}
	\caption{The case of an interval of length $k$ being intersected by $k-2$ (left) or $k-1$ (right) cuts, here $k=5$. The monitored subintervals are depicted in blue. On the left, an interval of length $5$ is cut by only $3$ cuts.  The interval defined by the second and third cuts is of length larger than $1+1/k = 6/5$. On the right, an interval of length $5$ is cut by $4$ cuts. It is possible to achieve an approximately balanced partition, but only if all cuts are $\delta$-close to integer coordinates and specifically to midpoints of the monitored subintervals, which is indicated by the red cuts. A case where this does not happen is indicated by the green cuts, where the blanket-sensor agent of interval $[2,4]$ is active. Note that while in the figure, in both cases, the interval of length $k$ contains only full monitored intervals, the same arguments go through if it contains half intervals instead, e.g. considering the interval $[1,5]$ and $2$ cuts (left) and $3$ cuts (right).}\label{fig:44one}
\end{figure}

\begin{lemma}\label{lem:44onemissingcut}
	For some $\ell \in \{0,n-k\}$, with $k>1$, consider the interval $I=[\ell, \ell+k]$ of length $k$ and assume that there are $k-1$ cuts in this interval. Then either
	\begin{itemize}
		\item[-] each of the $k-1$ cuts in $I$ will be $\delta$-close to a different integer point and these integer points will be the midpoints of the monitored subintervals contained entirely in $I$ or 
		\item [-] at least one of the blanket-sensors monitoring the subintervals in $I$ will be active.
	\end{itemize}
\end{lemma}

\begin{proof}
	In $I$, there are at least $k-1$ intervals monitored by blanket-sensor agents and we have $k-1$ cuts at our disposal. Assume that there exists some integer point $j-1$ which is a midpoint of some monitored interval $[j-2,j] \subseteq I$ such that there exists no cut that is $\delta$-close to $j-1$. This implies that either $[j-2,j]$ is intersected by at least $2$ cuts, or the blanket-sensor agent corresponding to $[j-2,j]$ will be active. To see this, note that if there were not any cuts in $[j-2,j]$ then obviously the whole interval $[j-2,j]$ would be monochromatic and the corresponding blanket-sensor would be active. If $[j-2,j]$ was intersected by a single cut, since the cut lies at distance at least $\delta$ from the midpoint of the interval, there would exist a monochromatic subinterval of $[j-2,j]$ of length at least $1+\delta$, activating the corresponding blanket-sensor agent $b_{1,j}$.
	
	However, for the blanket-sensor agent $b_{1,j}$ to not be active, it would have to be the case that some other cut in the interval $[j-2,j]$ is also not $\delta$-close to the midpoint of one of the adjacent monitored intervals, therefore generating an imbalance in labels that has to be compensated with at least one additional cut in that interval. Given that there are only $k-1$ cuts in the interval $I$, it follows that in some monitored subinterval $[j',j'-2]$ there will be a large enough imbalance, i.e a monochromatic subinterval of length at least $1+\delta$, and the corresponding blanket-sensor agent $b_{1,j'}$ will be active. See the right-hand side of Figure \ref{fig:44one} for an illustration.
\end{proof}
We are now ready to proceed with the proof of Proposition \ref{prop:s-r-radius}.
\begin{proof}[Proof of Proposition \ref{prop:s-r-radius}]
	First, recall that by Observation \ref{obs:cer}, at most $n$ cuts can lie in the c-e region. Also recall that from Definition \ref{def:blanket}, for the circuit encoder $C_1$, the blanket-sensor agent $b_{1,j}$, $j \in \{2,\ldots,n\}$ has valuation only in the interval $[j-2,j]$ of the c-e region, i.e. it ``monitors'' the interval $[j-2,j]$. The blanket-sensor agent $b_{i,j}$ for $i \in\{2,\ldots,p^C\}$ is a $(i-1)\frac{\delt}{p^C}$-shifted version of $b_{1,j}$. We will make the argument for the blanket-sensor agents of the circuit-encoder $C_1$; the argument for any $C_i$, with $i \neq 1$ is very similar.
	%	
	%Also recall the definition of a blanket-sensor being active; $b_{1,j}$ is active when the difference in value $|\mu_{b_{1,j}}(\lplus)-\mu_{b_{1,j}}(\lminus)|$ between the labels in that interval is larger than an inverse-polynomial quantity in $n$. In that case, the blanket-sensor agent ``overrides'' the circuit and we are not at a solution, therefore 
	
	It suffices to prove that if consecutive cuts are too far apart or too close together, some blanket-sensor agent will be active. \\
	
	\noindent \textbf{Case 1: The cuts are too far apart}. First consider the case when two consecutive cuts are too far apart (by more than $1$ plus some inverse-polynomial amount $2\delta$). More formally, assume that there are two cuts $c_1$ and $c_2$ such that 
	$c_2 > c_1$ and $c_2 - c_1 > 1+2\delta$. Then, as we explain below, there is some $j \in \{2,\ldots,n\}$ such that some subinterval $I_j=[j_1,j_2] \subseteq [j-2,j]$ with $j_2-j_1 > 1+\delta$ will receive a single label, either $\lplus$ or $\lminus$. In particular, we have the following cases:
	\begin{itemize}
		\item[-] There is a $j$ such that $[c_1,c_2] \subseteq [j-2,j]$. In that case, $[c_1,c_2]$ is such a monochromatic subinterval.
		\item[-] There is a $j$ such that $[j-2,j] \subseteq [c_1,c_2]$. In that case, the whole monitored subinterval $[j-2,j]$ is such a monochromatic subinterval. 
		\item[-] For all $j$, the interval $[j-2,j]$ is intersected by at most one cut $c_\ell$, $\ell \in \{1,2\}$. Obviously, both cuts will intersect some interval, since they lie in the c-e region. Consider cut $c_1$ and let $[j-2,j]$ be an interval that is intersected by $c_1$. If $c_1$ lies in $[j-2,j-1]$, then, since there exists no other cut between $c_1$ and $c_2$ and since $c_2$ does not intersect $[j-2,j]$ by assumption, the interval $[c_1,j]$ will be a monochromatic interval of length at least $1+\delta$ and we are done. If $c_1$ lies in $[j-1,j]$, then first observe that $j \neq n$, as otherwise both cuts $c_1$ and $c_2$ would have to lie in $[n-2,n]$ violating the assumption of the case. Therefore, we can look at the interval $[j-1,j+1]$ and notice that again by the assumption of the case, since cut $c_1$ does intersect the interval $[j-1,j+1]$, we must have that $c_2 > j+1$. This is either impossible (when $j=n-1$) or otherwise $[c_1,j+1]$ is a monochromatic interval of length at least $1+\delta$, and we are done.
	\end{itemize} 
	\textbf{Case 2: The cuts are too close together}. Now consider the case when two consecutive cuts are too close together, closer than $1-2n\delta$. More formally, assume that there are two consecutive cuts $c_1$ and $c_2$ in the c-e region such that $c_2>c_1$ and $c_2-c_1<1-2n\delta$. Since the cuts are close together, there exists a monitored interval that is intersected by both $c_1$ and $c_2$ and let $[j-2,j]$ be such an interval. Notice that if there exists no other cut that intersects $[j-2,j]$, then $[j-2,c_1] \cup [c_2,j]$ is a union of subintervals of length at least $1+\delta$ that receive the same label, and we are done. Therefore there must exist at least $3$ cuts that lie in $[j-2,j]$. We consider three cases. \\
	
	\noindent \emph{There are 5 or more cuts in $[j-2,j]$}. This is an easy case to argue, as if that happens, there will be some interval, either $[0,j-2]$ or $[j,n]$ of length $k$ that is only intersected by at most $k-2$ cuts. By Lemma \ref{lem:44twomissingcuts}, some blanket-sensor agent will be active and we are done. \\
	
	\noindent \emph{There are 4 cuts in $[j-2,j]$}. Consider the intervals $[0,j-2]$ and $[j,n]$. If either $[0,j-2]$ is intersected by at most $(j-2)-2$ cuts or $[j,n]$ is intersected by at most $n-j-2$ cuts, then by Lemma \ref{lem:44twomissingcuts}, some blanket-sensor will be active and we are done. Note also for completeness that, if $j=2$ (respectively $j=n$), it is necessarily the case that $[j-2,n]$ (respectively $[0,j-2]$) is intersected by $n-4$ cuts and Lemma \ref{lem:44twomissingcuts} again applies. Therefore, we can assume that $j \in \{3,\ldots,n-1\}$, and that there are exactly $(j-2)-1$ cuts in $[0,j-2]$ and $n-j-1$ cuts in $[j,n]$. 
	
	Consider the interval $[j,n]$ without loss of generality, as the argument for $[0,j-2]$ is symmetric. By Lemma \ref{lem:44onemissingcut}, we know that the cuts in $[j,n]$ are $\delta$-close to integer points and particularly, they are $\delta$-close to the midpoints of the monitored intervals $[j,j+2], \ldots, [n-2,n]$. This implies that in the monitored subinterval $[j,j+2]$, the subinterval $[j,j+1-\delta]$ will be a monochromatic interval of label $A_j$ for some $A_{j} \in \{\lplus,\lminus\}$, 
	
	In turn, this implies that $[j-1,j]$ has a monochromatic subinterval of length at least $1-\delta$ that receives the label $A_{-j}$, where $A_{-j} \in \{\lplus,\lminus\}$ is the complementary label to $A_j$, for the blanket-sensor agent to not be activated, which is only possible if one of the $4$ cuts in $[j-2,j]$ is $\delta$-close to the integer point $j$. Propagating the effect of this cut sequence/labelling into the monitored interval $[j-2,j]$ in question, we obtain that $[j-2,j-1]$ also contains a monochromatic interval of length at least $1-\delta$ and label $A_j$, as otherwise blanket-sensor agent $b_{1,j}$ would be active. From this discussion, it follows that:
	\begin{itemize}
		\item [-] all the cuts in $[j-2,j]$ are $\delta$-close to integer coordinates within the interval and
		\item [-] there is at least one cut in $[j-2,j]$ that is $\delta$-close to the midpoint $j-1$ of the monitored interval, one cut that is $\delta$-close to the right endpoint $j$ of the monitored interval and at least one cut that is $\delta$-close to the left endpoint $j-2$ of the monitored interval, 
	\end{itemize}
	where the very last statement follows from the symmetric argument to the one developed above, for the interval $[0,j-2]$. See Figure \ref{fig:44two} for an illustration.\\
	
	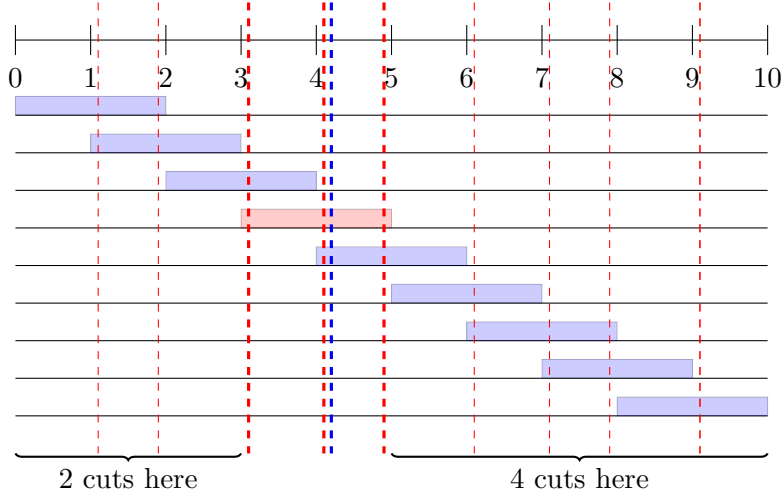
\begin{figure}
		\centering
		\begin{tikzpicture}
		
		% Straight line and integer points
		
		\draw (0,0) -- (10,0);
		\draw (0,0.2) -- (0,-0.2);
		\draw (1,0.2) -- (1,-0.2);
		\draw (2,0.2) -- (2,-0.2);
		\draw (3,0.2) -- (3,-0.2);
		\draw (4,0.2) -- (4,-0.2);
		\draw (5,0.2) -- (5,-0.2);
		\draw (6,0.2) -- (6,-0.2);
		\draw (7,0.2) -- (7,-0.2);
		\draw (8,0.2) -- (8,-0.2);
		\draw (9,0.2) -- (9,-0.2);
		\draw (10,0.2) -- (10,-0.2);
		\node at (0,-0.5) {$0$};
		\node at (1,-0.5) {$1$};
		\node at (2,-0.5) {$2$};
		\node at (3,-0.5) {$3$};
		\node at (4,-0.5) {$4$};
		\node at (5,-0.5) {$5$};
		\node at (6,-0.5) {$6$};
		\node at (7,-0.5) {$7$};
		\node at (8,-0.5) {$8$};
		\node at (9,-0.5) {$9$};
		\node at (10,-0.5) {$10$};
		
		% Blanket-sensor intervals
		
		\draw (0,-1) -- (10,-1);
		\draw (0,-1.5) -- (10,-1.5);
		\draw (0,-2) -- (10,-2);
		\draw (0,-2.5) -- (10,-2.5);
		\draw (0,-3) -- (10,-3);
		\draw (0,-3.5) -- (10,-3.5);
		\draw (0,-4) -- (10,-4);
		\draw (0,-4.5) -- (10,-4.5);
		\draw (0,-5) -- (10,-5);
		
		\draw[fill=blue,opacity=0.2] (0, -1) rectangle (2,-0.75);
		\draw[fill=blue,opacity=0.2] (1, -1.5) rectangle (3,-1.25);
		\draw[fill=blue,opacity=0.2] (2, -2) rectangle (4,-1.75);
		\draw[fill=red,opacity=0.2] (3, -2.5) rectangle (5,-2.25);
		\draw[fill=blue,opacity=0.2] (4, -3) rectangle (6,-2.75);
		\draw[fill=blue,opacity=0.2] (5, -3.5) rectangle (7,-3.25);
		\draw[fill=blue,opacity=0.2] (6, -4) rectangle (8,-3.75);
		\draw[fill=blue,opacity=0.2] (7, -4.5) rectangle (9,-4.25);
		\draw[fill=blue,opacity=0.2] (8, -5) rectangle (10,-4.75);
		
		%Cuts 
		
		\draw[dashed,color=red] (1.1,0.5) -- (1.1,-5.5);
		\draw[dashed,color=red] (1.9,0.5) -- (1.9,-5.5);
		
		\draw[dashed, very thick, color=red] (3.1,0.5) -- (3.1,-5.5);
		\draw[dashed, very thick, color=red] (4.1,0.5) -- (4.1,-5.5);
		\draw[dashed, very thick, color=blue] (4.2,0.5) -- (4.2,-5.5);
		\draw[dashed, very thick, color=red] (4.9,0.5) -- (4.9,-5.5);
		
		\draw[dashed,color=red] (6.1,0.5) -- (6.1,-5.5);
		\draw[dashed,color=red] (7.1,0.5) -- (7.1,-5.5);
		\draw[dashed,color=red] (7.9,0.5) -- (7.9,-5.5);
		\draw[dashed,color=red] (9.1,0.5) -- (9.1,-5.5);
		
		\draw[dashed,color=red] (9.1,0.5) -- (9.1,-5.5);
		
		\draw [
		thick,
		decoration={
			brace,
			mirror,
			raise=0.5cm
		},
		decorate
		] (0,-5) -- (3,-5) 
		node [pos=0.5,anchor=north,yshift=-0.55cm] {$2$ cuts here}; 
		
		\draw [
		thick,
		decoration={
			brace,
			mirror,
			raise=0.5cm
		},
		decorate
		] (5,-5) -- (10,-5) 
		node [pos=0.5,anchor=north,yshift=-0.55cm] {$4$ cuts here}; 
		\end{tikzpicture}
		\caption{The case in which there are $4$ cuts in the interval $[j-2,j]$ (shown in red), here for $j=5$. The $3$ cuts that lie in $\delta$-distance from the left endpoint, midpoint and right endpoint of $[j-2,j]$ are depicted with thick dashed red lines. The other cut in the interval (based on the position of which the different cases are considered) is depicted by a thick dashed blue line and in the particular case, it is shown to be $\delta$-close to the midpoint of the interval. Notice that the positioning of the cuts in $[0,3]$ and in $[5,10]$ is such that the cuts are $\delta$-close to integer coordinates which are the midpoints of the monitored subintervals. If that was not the case, then some subintervals would be sufficiently imbalanced and the corresponding blanket-sensor agent would be active.}
		\label{fig:44two}
	\end{figure}
	
	\noindent Now, we consider three cases with respect to the positions of the $4$ cuts in $[j-2,j]$, illustrated in Figure \ref{fig:44three}. From the discussion above, we know that three of the cuts will be $\delta$-close to the left-endpoint, midpoint and right-endpoint of $[j-2,j]$ respectively, so it suffices to consider the cases depending on the position of the fourth cut. Henceforth, we use $c_1,c_2,c_3$ to denote these three cuts, from left to right in terms of their position within the interval and $\tilde{c}$ to denote the aforementioned fourth cut.
	\begin{enumerate}[label=(\Alph*)]
		\item $\tilde{c}$ is $\delta$-close to $j-1$. In that case, assuming wlog that $\tilde{c} < c_2$, due to the parity of the cut sequence, the union of intervals $[c_1,\tilde{c}] \cup [c_2,c_3]$ contains monochromatic intervals of the same label and length at least $1+\delta$ and therefore blanket-sensor $b_{1,j}$ will be active. See the left-hand side of Figure \ref{fig:44three}.\label{case:A}
		\item $\tilde{c}$ is $\delta$-close to $j$. In that case, it is possible that $[j-2,j]$ does not contain a union of monochromatic intervals of the samel label of length at least $1+\delta$. However, by the parity of the cut sequence, in the interval $[j-1,j+1]$, now most of the interval $[j-1,j+1]$ will receive the same label, and $[j-1,j+1]$ will contain a union of monochromatic intervals of the same label of total length at least $1+\delta$, activating the blanket-sensor $b_{1,j+1}$. See the right-hand side of Figure \ref{fig:44three}.\label{case:B}
		\item $\tilde{c}$ is $\delta$-close to $j-2$. This case is symmetric to Case \ref{case:B} above. 
	\end{enumerate}
	
	\begin{figure}
		\centering
		\begin{tikzpicture}
		
		% Straight line and integer points
		
		\draw [dotted] (1,0) -- (2,0);
		\draw (2,0) -- (6,0);
		
		\draw [dotted] (9,0) -- (10,0);
		\draw (10,0) -- (14,0);
		
		\draw (2,0.2) -- (2,-0.2);
		\draw (3,0.2) -- (3,-0.2);
		\draw (4,0.2) -- (4,-0.2);
		\draw (5,0.2) -- (5,-0.2);
		\draw (6,0.2) -- (6,-0.2);
		\draw [dotted] (6,0) -- (7,0);
		
		\draw (10,0.2) -- (10,-0.2);
		\draw (11,0.2) -- (11,-0.2);
		\draw (12,0.2) -- (12,-0.2);
		\draw (13,0.2) -- (13,-0.2);
		\draw (14,0.2) -- (14,-0.2);
		\draw [dotted] (14,0) -- (15,0);
		
		\node at (2,-0.5) {$j-3$};
		\node at (3,-0.5) {$j-2$};
		\node at (4,-0.5) {$j-1$};
		\node at (5,-0.5) {$j$};
		\node at (6,-0.5) {$j+1$};
		
		\node at (10,-0.5) {$j-3$};
		\node at (11,-0.5) {$j-2$};
		\node at (12,-0.5) {$j-1$};
		\node at (13,-0.5) {$j$};
		\node at (14,-0.5) {$j+1$};
		
		% Blanket-sensor intervals

		\draw (1.5,-1.5) -- (6.5,-1.5);
		\draw (1.5,-2) -- (6.5,-2);
		\draw (1.5,-2.5) -- (6.5,-2.5);
		\draw (1.5,-3) -- (6.5,-3);
		
		\draw (9.5,-1.5) -- (14.5,-1.5);
		\draw (9.5,-2) -- (14.5,-2);
		\draw (9.5,-2.5) -- (14.5,-2.5);
		\draw (9.5,-3) -- (14.5,-3);
		
		\draw[fill=blue,opacity=0.2] (2, -2) rectangle (4,-1.75);
		\draw[fill=red,opacity=0.2] (3, -2.5) rectangle (5,-2.25);
		\draw[fill=blue,opacity=0.2] (4, -3) rectangle (6,-2.75);
		
		\draw[fill=blue,opacity=0.2] (10, -2) rectangle (12,-1.75);
		\draw[fill=red,opacity=0.2] (11, -2.5) rectangle (13,-2.25);
		\draw[fill=blue,opacity=0.2] (12, -3) rectangle (14,-2.75);

		%Cuts 
		
		\draw[dashed,color=red] (3.1,0.5) -- (3.1,-4.5);
		\draw[dashed,color=red] (4.1,0.5) -- (4.1,-4.5);
		\draw[dashed,color=red] (4.9,0.5) -- (4.9,-4.5);
		\draw[dashed,color=blue] (3.9,0.5) -- (3.9,-4.5);
		
		\draw[dashed,color=red] (11.1,0.5) -- (11.1,-4.5);
		\draw[dashed,color=red] (12.1,0.5) -- (12.1,-4.5);
		\draw[dashed,color=red] (12.9,0.5) -- (12.9,-4.5);
		\draw[dashed,color=blue] (12.8,0.5) -- (12.8,-4.5);
		\draw[dashed,color=red] (13.9,0.5) -- (13.9,-4.5);
		
		%\draw[dashed, very thick, color=red] (3.1,0.5) -- (3.1,-5.5);
		%\draw[dashed, very thick, color=red] (4.1,0.5) -- (4.1,-5.5);
		%\draw[dashed, very thick, color=blue] (4.2,0.5) -- (4.2,-5.5);
		%\draw[dashed, very thick, color=red] (4.9,0.5) -- (4.9,-5.5);
		
		%\draw[dashed,color=red] (6.1,0.5) -- (6.1,-5.5);
		%\draw[dashed,color=red] (7.1,0.5) -- (7.1,-5.5);
		%\draw[dashed,color=red] (7.9,0.5) -- (7.9,-5.5);
		%\draw[dashed,color=red] (9.1,0.5) -- (9.1,-5.5);
		
		%\draw[dashed,color=red] (9.1,0.5) -- (9.1,-5.5);
		
		% Labels
		\node at (2.5,-4) {$\lplus$};
		\node at (3.5,-4) {$\lminus$};
		\node at (4.5,-4) {$\lminus$};
		\node at (5.5,-4) {$\lplus$};
		
		\draw[->] (4,-4.7) -- (4,-4.5);
		\node at (4,-5) {$\lplus$};
		
		\node at (10.5,-4) {$\lplus$};
		\node at (11.5,-4) {$\lminus$};
		\node at (12.5,-4) {$\lplus$};
		\node at (13.5,-4) {$\lplus$};
		\node at (14.5,-4) {$\lminus$};
		
		\draw[->] (12.85,-4.7) -- (12.85,-4.5);
		\node at (12.85,-5) {$\lminus$};
		
		\end{tikzpicture}
		\caption{The two subcases of the case when there are $4$ cuts in the interval $[j-2,j]$. The three cuts $c_1,c_2$ and $c_3$ that are $\delta$-close to the integer points $j-2$, $j-1$ and $j$ in the interval are shown in red, the other cut $\tilde{c}$ is shown in blue. On the left, when $\tilde{c}$ is $\delta$-close to the midpoint $j-1$ of the interval, most of $[j-2,j]$ is coloured with the same label, here $\lminus$, by the parity of the cut sequence. On the right, $\tilde{c}$ is $\delta$-close to the right endpoint $j$ of the interval which means that, by the parity of the cut sequence, most of $[j-1,j+1]$ receives the label $\lplus$, since if there is another cut in the interval, it is constrained by the arguments of the proof to be $\delta$-close to the right endpoint $j+1$ (shown in red here).}
		\label{fig:44three}
	\end{figure}
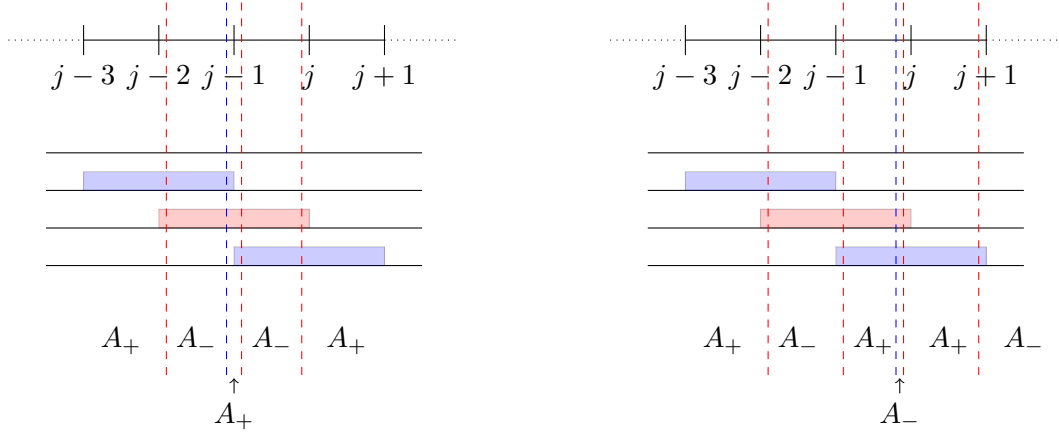
	
	\noindent \emph{There are $3$ cuts in $[j-2,j]$}. Again, considering the intervals $[0,j-2]$ and $[j,n]$ as we did in the case of $4$ cuts in $[j-2,j]$, we can now observe that one of the intervals will be intersected by at most $k-1$ cuts, where $k \in \{j-2,n-j\}$ is its length. Furthermore, if it is intersected by fewer than $k-1$ cuts, by Lemma \ref{lem:44twomissingcuts} some blanket-sensor agent will be active and we are done. Therefore, we will consider the case when one of the intervals is intersected by exactly $k-1$ cuts and let $[j,n]$ be that interval, without loss of generality, as the argument for $[0,j-2]$ is symmetric.
	
	Following exactly the same arguments as in the second and third paragraph of the case of $4$ cuts above, we can establish a very similar statement, namely that:
	\begin{itemize}
		\item [-] all the cuts in $[j-2,j]$ are $\delta$-close to integer coordinates within the interval and
		\item [-] there is at least one cut in $[j-2,j]$ that is $\delta$-close to the midpoint $j-1$ of the monitored interval, and one cut that is $\delta$-close to the right endpoint $j$ of the monitored interval.
	\end{itemize}
	
	\noindent Again, letting $c_1$ and $c_2$ denote the two cuts mentioned in the second item above from left to right in terms of their positions, we will consider some cases depending on the position of the third cut, which we will denote by $\tilde{c}$. 
	\begin{enumerate}[label=(\alph*)]
		\item $\tilde{c}$ is $\delta$-close to $j-2$. In that case, considering the intervals $[\tilde{c},c_1]$ and $[c_1,c_2]$, we observe that since the cuts $\tilde{c}$, $c_1$ and $c_2$ are $\delta$-close to the integer points $j-2,j-1$ and $j$ respectively, both intervals have length at least $1-2\delta$. However, this contradicts the assumption of the case, namely that there exists two cuts in $[j-2,j]$ that are within distance at most $1-n\delta$ from each other. See Figure \ref{fig:44four}, left-hand side.
		\item $\tilde{c}$ is $\delta$-close to $j-1$. In that case, similarly to  Case \ref{case:A} for the case of $4$ cuts, the parity of the cut sequence is such that most of $[j-2,j]$ will receive the same label and in particular $[j-2,j]$ will contain a union of monochromatic intervals of the same label with total length at least $1+\delta$, activating blanket-sensor $b_{1,j}$. See Figure \ref{fig:44four}, middle.
		\item $\tilde{c}$ is $\delta$-close to $j$. Again, similarly to Case \ref{case:B} for the case of $4$ cuts, it is possible that $[j-2,j]$ does not contain a union of monochromatic intervals of the samel label of length at least $1+\delta$. However, by the parity of the cut sequence, in the interval $[j-1,j+1]$, now most of the interval $[j-1,j+1]$ will receive the same label, and $[j-1,j+1]$ will contain a union of monochromatic intervals of the same label of total length at least $1+\delta$, activating the blanket-sensor $b_{1,j+1}$. See Figure \ref{fig:44four}, right-hand side.
	\end{enumerate}
	
	\begin{figure}
		\centering
		\begin{tikzpicture}[thick,scale=0.8, every node/.style={transform shape}]
		
		% Straight line and integer points
		
		\draw [dotted] (1,0) -- (2,0);
		\draw (2,0) -- (6,0);
		
		\draw [dotted] (8,0) -- (9,0);
		\draw (9,0) -- (13,0);
		
		\draw [dotted] (15,0) -- (16,0);
		\draw (16,0) -- (20,0);
		
		\draw (2,0.2) -- (2,-0.2);
		\draw (3,0.2) -- (3,-0.2);
		\draw (4,0.2) -- (4,-0.2);
		\draw (5,0.2) -- (5,-0.2);
		\draw (6,0.2) -- (6,-0.2);
		\draw [dotted] (6,0) -- (7,0);
		
		\draw (9,0.2) -- (9,-0.2);
		\draw (10,0.2) -- (10,-0.2);
		\draw (11,0.2) -- (11,-0.2);
		\draw (12,0.2) -- (12,-0.2);
		\draw (13,0.2) -- (13,-0.2);
		\draw [dotted] (13,0) -- (14,0);
		
		\draw (16,0.2) -- (16,-0.2);
		\draw (17,0.2) -- (17,-0.2);
		\draw (18,0.2) -- (18,-0.2);
		\draw (19,0.2) -- (19,-0.2);
		\draw (20,0.2) -- (20,-0.2);
		\draw [dotted] (20,0) -- (21,0);
		
		\node at (2,-0.5) {$j-3$};
		\node at (3,-0.5) {$j-2$};
		\node at (4,-0.5) {$j-1$};
		\node at (5,-0.5) {$j$};
		\node at (6,-0.5) {$j+1$};
		
		\node at (9,-0.5) {$j-3$};
		\node at (10,-0.5) {$j-2$};
		\node at (11,-0.5) {$j-1$};
		\node at (12,-0.5) {$j$};
		\node at (13,-0.5) {$j+1$};
		
		\node at (16,-0.5) {$j-3$};
		\node at (17,-0.5) {$j-2$};
		\node at (18,-0.5) {$j-1$};
		\node at (19,-0.5) {$j$};
		\node at (20,-0.5) {$j+1$};
		
		% Blanket-sensor intervals

		\draw (1.5,-1.5) -- (6.5,-1.5);
		\draw (1.5,-2) -- (6.5,-2);
		\draw (1.5,-2.5) -- (6.5,-2.5);
		\draw (1.5,-3) -- (6.5,-3);
		
		\draw (8.5,-1.5) -- (13.5,-1.5);
		\draw (8.5,-2) -- (13.5,-2);
		\draw (8.5,-2.5) -- (13.5,-2.5);
		\draw (8.5,-3) -- (13.5,-3);
		
		\draw (15.5,-1.5) -- (20.5,-1.5);
		\draw (15.5,-2) -- (20.5,-2);
		\draw (15.5,-2.5) -- (20.5,-2.5);
		\draw (15.5,-3) -- (20.5,-3);
		
		\draw[fill=blue,opacity=0.2] (2, -2) rectangle (4,-1.75);
		\draw[fill=red,opacity=0.2] (3, -2.5) rectangle (5,-2.25);
		\draw[fill=blue,opacity=0.2] (4, -3) rectangle (6,-2.75);
		
		\draw[fill=blue,opacity=0.2] (9, -2) rectangle (11,-1.75);
		\draw[fill=red,opacity=0.2] (10, -2.5) rectangle (12,-2.25);
		\draw[fill=blue,opacity=0.2] (11, -3) rectangle (13,-2.75);
		
		\draw[fill=blue,opacity=0.2] (16, -2) rectangle (18,-1.75);
		\draw[fill=red,opacity=0.2] (17, -2.5) rectangle (19,-2.25);
		\draw[fill=blue,opacity=0.2] (18, -3) rectangle (20,-2.75);

		%Cuts 
		
		%\draw[dashed,color=red] (3.1,0.5) -- (3.1,-4.5);
		%\draw[dashed,color=red] (4.1,0.5) -- (4.1,-4.5);
		%\draw[dashed,color=red] (4.9,0.5) -- (4.9,-4.5);
		%\draw[dashed,color=blue] (3.9,0.5) -- (3.9,-4.5);
		
		\draw[dashed,color=red] (4.9,0.5) -- (4.9,-4.5);
		\draw[dashed,color=red] (4.1,0.5) -- (4.1,-4.5);
		\draw[dashed,color=blue] (3.1,0.5) -- (3.1,-4.5);

		\draw[dashed,color=red] (11.9,0.5) -- (11.9,-4.5);
		\draw[dashed,color=red] (11.1,0.5) -- (11.1,-4.5);
		\draw[dashed,color=blue] (10.9,0.5) -- (10.9,-4.5);
		
		\draw[dashed,color=red] (18.9,0.5) -- (18.9,-4.5);
		\draw[dashed,color=red] (17.9,0.5) -- (17.9,-4.5);
		\draw[dashed,color=blue] (18.8,0.5) -- (18.8,-4.5);
		\draw[dashed,color=red] (19.9,0.5) -- (19.9,-4.5);
		
		%\draw[dashed,color=red] (11.1,0.5) -- (11.1,-4.5);
		%\draw[dashed,color=red] (12.1,0.5) -- (12.1,-4.5);
		%\draw[dashed,color=red] (12.9,0.5) -- (12.9,-4.5);
		%\draw[dashed,color=blue] (12.8,0.5) -- (12.8,-4.5);
		%\draw[dashed,color=red] (13.9,0.5) -- (13.9,-4.5);
		
		%\draw[dashed, very thick, color=red] (3.1,0.5) -- (3.1,-5.5);
		%\draw[dashed, very thick, color=red] (4.1,0.5) -- (4.1,-5.5);
		%\draw[dashed, very thick, color=blue] (4.2,0.5) -- (4.2,-5.5);
		%\draw[dashed, very thick, color=red] (4.9,0.5) -- (4.9,-5.5);
		
		%\draw[dashed,color=red] (6.1,0.5) -- (6.1,-5.5);
		%\draw[dashed,color=red] (7.1,0.5) -- (7.1,-5.5);
		%\draw[dashed,color=red] (7.9,0.5) -- (7.9,-5.5);
		%\draw[dashed,color=red] (9.1,0.5) -- (9.1,-5.5);
		
		%\draw[dashed,color=red] (9.1,0.5) -- (9.1,-5.5);
		
		% Labels
		\node at (2.5,-4) {$\lplus$};
		\node at (3.55,-4) {$\lminus$};
		\node at (4.5,-4) {$\lplus$};
		\node at (5.5,-4) {$\lminus$};
		
		%\draw[->] (4,-4.7) -- (4,-4.5);
		%\node at (4,-5) {$\lplus$};

		\node at (10.5,-4) {$\lminus$};
		\node at (11.5,-4) {$\lminus$};
		\node at (12.5,-4) {$\lplus$};

		\draw[->] (11,-4.7) -- (11,-4.5);
		\node at (11,-5) {$\lplus$};
		
		\node at (17.5,-4) {$\lminus$};
		\node at (18.35,-4) {$\lplus$};
		\node at (19.5,-4) {$\lplus$};
		\node at (20.5,-4) {$\lminus$};

		\draw[->] (18.85,-4.7) -- (18.85,-4.5);
		\node at (18.85,-5) {$\lminus$};

		\end{tikzpicture}
		\caption{The three subcases of the case when there are $3$ cuts in the interval $[j-2,j]$. The two cuts $c_1,c_2$ and that are $\delta$-close to the integer points $j-1$ and $j$ in the interval are shown in red, the other cut $\tilde{c}$ is shown in blue. On the left, when $\tilde{c}$ is $\delta$-close to the left endpoint $j-2$ of the interval, at least one of the subintervals defined by the cuts will have length at least $1-2\delta$, contradicting the assumption of the case. In the middle, $\tilde{c}$ is $\delta$-close to the midpoint $j-1$ of the interval $[j-2,j]$ and by the parity of the cut sequence, most of the interval receives the same label, here $\lminus$. Finally on the right, $\tilde{c}$ is $\delta$-close to the right endpoint $j$ of the interval which means that, by the parity of the cut sequence, most of $[j-1,j+1]$ receives the label $\lplus$, since if there is another cut in the interval, it is constrained by the arguments of the proof to be $\delta$-close to the right endpoint $j+1$ (shown in red here).}
		\label{fig:44four}
	\end{figure}
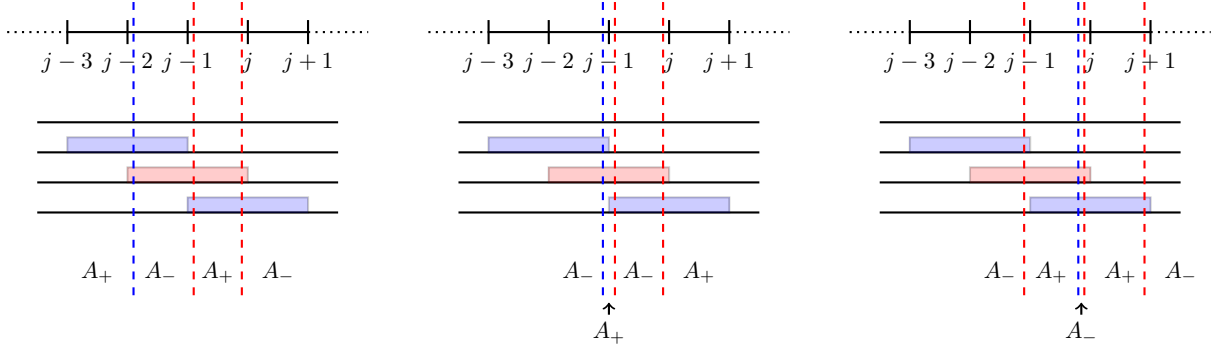
	\noindent This completes the proof.
\end{proof}

\noindent {\bf Remarks:}
Looking ahead, certain points in the Significant Region encode \nvhdt\ (namely,
the ones in the ``twisted tunnel'', Definition~\ref{def:twisted}).
The Significant Region contains the twisted tunnel, being a somewhat wider 1-dimensional ``tunnel'' of
inverse-polynomial width at most $1/p_w(n)$,
whose central axis is the set of points $(\alpha,1/n,\ldots,1/n,1/n-\alpha)$,
where the endpoints are identified together (noting Definition~\ref{def:simplex-domain}).
Topologically, the Significant Region is a high-dimensional M\"{o}bius strip.
\bigskip

\section{Reducing from \nvhdt\ to \ch}\label{sec:reduction}

In Sections~\ref{sec:forward} we give an overview of aspects of how we
construct an instance $I_{CH}$ of $\varepsilon$-\ch\ (in poly-time)
from an instance of \nvhdt, for inverse polynomial $\varepsilon$.
Section~\ref{sec:coord-sys} describes the new coordinate system
for the M\"{o}bius-simplex $D$ and establishes key properties.
Section~\ref{sec:colour-f} presents a colouring function of $D$
in terms of the coordinate system of Section~\ref{sec:coord-sys}.
Section~\ref{sec:backward} describes how to construct a purported
solution to $n$-dimensional \nvhdt\ from a solution to $\varepsilon$-\ch.
In Section~\ref{sec:prove-reduction} we prove that a solution
to \nvhdt\ that is obtained by reducing to $\varepsilon$-\ch,
solving it, and converting that solution to a solution to $n$-dimensional
\nvhdt, really is a valid solution.

\subsection{Overview of the construction of an instance of $\varepsilon$-\ch\ from
an instance of \nvhdt}\label{sec:forward}

We define the reduction from \nvhdt\ (Definition~\ref{def:nvhdt}) to $\varepsilon$-\ch.

Let $I_{VT}$ be an instance of \nvhdt\ in $n$ dimensions; let $C_{VT}$ be the
boolean circuit that represents it.
$I_{CH}$ will be the corresponding instance of \ch.
We list ingredients of $I_{CH}$ and give notation to represent them, as follows.
$A$ is the consensus-halving domain, an interval of the form $[0,poly(n)]$.
Any agent $a$ has a measure $\mu_a:A\longrightarrow\rset$
represented as a step function (thus having a polynomial number of steps).

\begin{itemize}
\item
$I_{CH}$ has $n$ {\em coordinate-encoding agents}
$a_1,\ldots,a_n$ (Definition \ref{def:c-e-agents}).
See Figure~\ref{fig:overview_general}.
\item
The consensus-halving domain $A$ of $I_{CH}$ has a
{\em coordinate-encoding region} (c-e region) (Definition \ref{def:c-e-region}) consisting of the interval $[0,n]$.
\item
$I_{CH}$ has $p^C$ {\em circuit-encoders} (Sections \ref{sec:c1}, \ref{sec:c2plus}), $C_1,\ldots,C_{p^C}$.
\begin{itemize}
\item
Each $C_i$ has a set ${\cal A}_i$ of agents (see Figure~\ref{fig:overview_general})
which includes $C_i$'s sensor agents, also circuit-encoding agents (below).
\item
Each $C_i$ has an associated circuit-encoding region $R_i$ of $A$;
each $R_i$ is an interval of polynomial length, and the $R_i$ do not intersect
with each other or with the coordinate-encoding region.
\item
${\cal A}_i$ contains a polynomial number of {\em circuit-encoding agents} (one for each gate of $C_{VT}$),
having value in $R_i$.
\item
Each $C_i$ has $\pgig$ {\em sensor agents} as defined in Definition~\ref{def:sensors} each of which
has a block of value $1/10$ in a small subinterval of the
c-e region as specified in Definition~\ref{def:sensors}, and further value in region $R_i$.
\item
Each $C_i$ has $n-1$ {\em blanket-sensor agents} as in Definition \ref{def:blanket}.
\end{itemize}
\end{itemize}

\noindent {\bf Remarks:}
We associate one cut with each agent; let $c(a)$ be the cut associated with agent $a$.
The cuts $c(a_i)$ for coordinate-encoding agents, are called the {\em coordinate-encoding cuts} (or c-e cuts).
A straightforward consequence of Proposition~\ref{prop:s-r-radius}
is that in any solution, either all $n$, or $n-1$, of the coordinate-encoding
cuts must lie in the coordinate-encoding region.
All other cuts must lie in the regions $R_i$, indeed, every cut, other than the c-e cuts,
is constrained by the value of its associated agent, to lie in a small interval that does
not overlap any other such intervals.
In the event that a c-e cut lies outside the c-e region, we refer to it as a ``stray
cut'', and while such a cut may initially appear to interfere with the functioning
of the circuitry, similarly to \cite{FG17} we have that the duplication of the circuit using $p^C$
circuit-encoders, allows the circuitry to be robust to this problem. See Appendix \ref{sec:app_details} for more details.
\bigskip

\begin{figure}
\begin{tikzpicture}[scale=1]

% agents horiz lines:
\foreach \y in {10,-10,-30,-50,-70,-90,-130,-150}{
  \draw (4pt,\y pt)--(360pt,\y pt);
}

% short vertical lines:
\foreach \y in {7,-13,-33,-53,-73,-93,-133,-153}{
  \foreach \x in {4,30,110,190,270,360}{
    \draw(\x pt,\y pt)--(\x pt,6+\y pt);
  }
}

\foreach \y in {2,-37,-77.5,-137.5}{
  \foreach \x in {17,69,149,229,309}{
    \node at (\x pt,\y pt){$\vdots$};
  }
}

\draw[fill=blue,opacity=0.2] (30pt, 13pt) rectangle (360pt,7pt);
\draw[fill=blue,opacity=0.2] (30pt, -7pt) rectangle (360pt,-13pt);

\draw[fill=blue,opacity=0.2] (4pt, -33pt) rectangle (110pt, -27pt);
\draw[fill=blue,opacity=0.2] (4pt, -53pt) rectangle (110pt, -47pt);

\draw[fill=blue,opacity=0.2] (4pt, -73pt) rectangle (30pt, -67pt);
\draw[fill=blue,opacity=0.2] (110pt, -73pt) rectangle (190pt, -67pt);

\draw[fill=blue,opacity=0.2] (4pt, -93pt) rectangle (30pt, -87pt);
\draw[fill=blue,opacity=0.2] (110pt, -93pt) rectangle (190pt, -87pt);

\draw[fill=blue,opacity=0.2] (4pt, -133pt) rectangle (30pt, -127pt);
\draw[fill=blue,opacity=0.2] (270pt, -133pt) rectangle (360pt, -127pt);

\draw[fill=blue,opacity=0.2] (4pt, -153pt) rectangle (30pt, -147pt);
\draw[fill=blue,opacity=0.2] (270pt, -153pt) rectangle (360pt, -147pt);

\node (a) at (-13pt, -107pt){$\vdots$};
\node (b) at (230pt, 30pt) {\ldots};

\draw [thick,decoration={brace,raise=5pt},decorate] (4pt,20pt) -- (30pt,20pt)node [pos=0.5,anchor=south,yshift=5pt] {\small c-e region};

\draw [thick,decoration={brace,raise=5pt},decorate] (32pt,20pt) -- (110pt,20pt)node [pos=0.5,anchor=south,yshift=5pt] {$R_1$};

\draw [thick,decoration={brace,raise=5pt},decorate] (112pt,20pt) -- (190pt,20pt)
node [pos=0.5,anchor=south,yshift=5pt] {$R_2$};
\draw [thick,decoration={brace,raise=5pt},decorate] (271pt,20pt) -- (360pt,20pt)
node [pos=0.5,anchor=south,yshift=5pt] {$R_{p^C}$};
\draw [thick,decoration={brace,raise=5pt},decorate] (33pt,40pt) -- (360pt,40pt)node [pos=0.5,anchor=south,yshift=5pt] {$R$};

\draw [thick,decoration={brace,mirror,raise=5pt},decorate] (-3pt,10pt) -- (-3pt,-10pt)
node [pos=0.5,anchor=east,xshift=-5pt] {\scriptsize{c-e agents}};
\node at(-1pt,10pt){\scriptsize{$a_1$}};\node at(-1pt,-10pt){\scriptsize{$a_n$}};

\draw [thick,decoration={brace,mirror,raise=5pt},decorate] (3pt,-27pt) -- (3pt,-50pt)node [pos=0.5,anchor=east,xshift=-5pt] {\scriptsize{$\mathcal{A}_1$}};

\draw [thick,decoration={brace,mirror,raise=5pt},decorate] (3pt,-67pt) -- (3pt,-90pt)node [pos=0.5,anchor=east,xshift=-5pt] {\scriptsize{$\mathcal{A}_2$}};

\draw [thick,decoration={brace,mirror,raise=5pt},decorate] (3pt,-130pt) -- (3pt,-150pt)
node [pos=0.5,anchor=east,xshift=-5pt] {\scriptsize{$\mathcal{A}_{p^C}$}};

\end{tikzpicture}
\caption{\small{An overview of $I_{CH}$, denoting all the different regions and the agents
of $C_1 \ldots, C_n$, as well as the coordinate-encoding agents. The highlighted areas
denote that the corresponding agent has non-zero value on these regions.}}\label{fig:overview_general}
\end{figure}
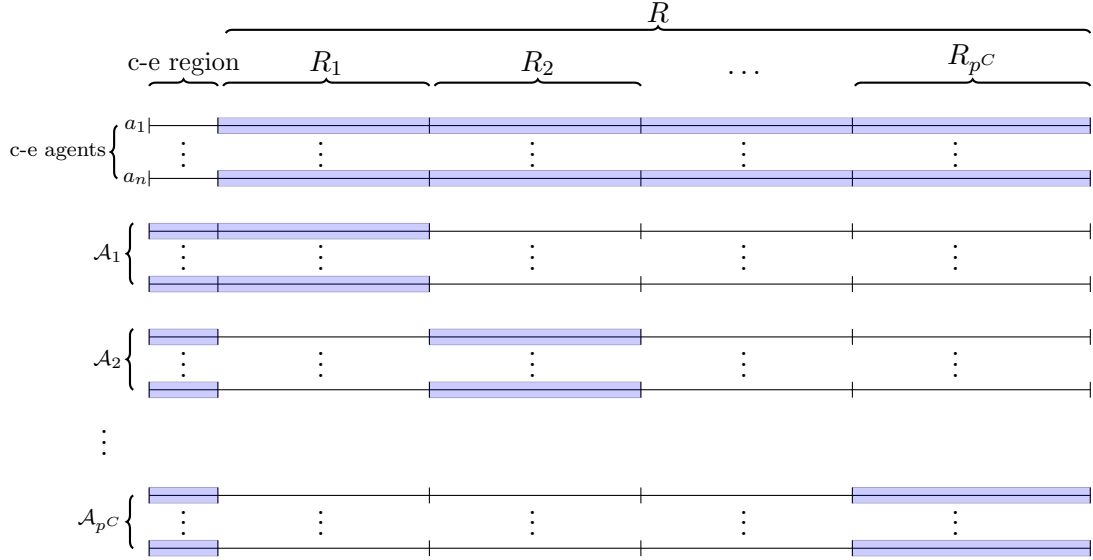

\subsubsection{Construction of $C_1$}\label{sec:c1}

Recall $C_{VT}$ is the boolean circuit in the instance $I_{VT}$ of \nvhdt.
\begin{itemize}
\item
We assume that $C_{VT}$ has $2n$ output gates
$g_1,\ldots,g_n$ and $g_{-1},\ldots,g_{-n}$ having the property
that exactly one of them will take value \true\ (this may be enforced
syntactically). $g_i$ getting value 1 (\true) means that the point at coordinates
represented by the input gets coloured $i$.
\item
$C_{VT}$ has $n\cdot{\rm polylog}(n)$ input gates, representing the coordinates of a point
in $B=[-1,1]^n$, each represented with inverse-polynomial precision.
\end{itemize}
We describe how circuit-encoder $C_1$ is derived from $C_{VT}$.
The subsequent circuit-encoders can then be specified in terms of $C_1$.
Each gate $g$ of $C_{VT}$ is simulated using a ``gate agent'' $a(g)$, as constructed in~\cite{FG17} (see Appendix \ref{sec:app_details} for a more detailed exposition).
$a(g)$'s cut $c(a(g))$ occupies a right position, or a left position, representing \true\
or \false, as a function of the cut(s) that represent boolean inputs to $g$.

The circuit-encoding agents ${\cal A}_1$ of $C_1$ thus include $2n$ gate agents whose
corresponding cuts simulate the values of the output gates of $C_{VT}$,
provided that the input represented by the c-e cuts lies in the ``Significant Region''
(Definition~\ref{def:sig-region}).
The positions of these cuts affect the labels of blocks of value held by
the $n$ coordinate-encoding agents, as detailed in Section~\ref{sec:cea}.

\begin{definition}\label{def:ref}
{\bf Reference sensor-agent.}
Noting from Definition~\ref{def:sensors} that the sensor agents for $C_i$
are denoted $\sensor_i = \{s_{i,1},\ldots,s_{i,\pgig}\}$, we let $s_{1,1}$ be the
{\em reference sensor-agent}:
Outputs produced by the circuit $C_i$ are taken with reference to the value\footnote{We are using $s_{1,1}$ to denote the boolean value taken by
	$s_{1,1}$ as well as the sensor itself.}
$s_{1,1}$,
in the sense that after simulating $C_{VT}$ we take the exclusive-or with $s_{1,1}$.
%the label allocated to its value-block in the c-e region is taken to represent \true.
\end{definition}
We used this crucial technique of Definition~\ref{def:ref} in \cite{FG17}:
it performs the task of disorienting the domain while at the same time ensuring
continuity when we move a cut from the left-hand side of the c-e region to
the right-hand side.

\begin{paragraph}{Preprocessing, prior to simulating $C_{VT}$}
For each $C_i$,
we take all $\pgig$ input bits, which appear in up to $n+1$ blocks of consecutive 1's
and 0's, and convert them into the coordinates of a point in the M\"{o}bius-simplex
(Definition~\ref{def:simplex-domain}). As noted earlier (Observation~\ref{obs:reliable}) at most $n$ circuit-encoders may receive ill-defined inputs caused
by c-e cuts cutting through value-blocks in the c-e region that belong to
their sensor agents; we simply assume that the output of those agents
is unreliable, indeed adversarially chosen.

We then perform a coordinate transformation described in Section~\ref{sec:coord-sys}.
A subset of points in the M\"{o}bius-simplex maps to a copy of the domain $B$ of the instance $I_{VT}$ (recall Definition~\ref{def:nvhdt}).
These points get their coordinates passed directly to a copy of $C_{VT}$, and the outputs of $C_{CV}$ are used to provide feedback to the c-e agents as described in Section~\ref{sec:cea} (and discussed in Observations~\ref{obs:reliable} and \ref{obs:average}). Other points get coloured in a
manner that avoids allowing bogus solutions to $I_{CH}$ (i.e. ones that do not encode solutions to $I_{VT}$).
\end{paragraph}

\subsubsection{Output gates of $C_1$, and the feedback they provide to
the coordinate-encoding agents}\label{sec:cea}

The following generalises~\cite{FG17}.
$C_{VT}$ has output gates $g_j$, $j\in\pm[n]$, with the property
that when inputs are well-defined, exactly one output gate evaluates to \true.
A circuit-encoder simulates $C_{VT}$ using the gate gadgets introduced in \cite{FFGZ18,FG17} (see Appendix \ref{sec:app_details}).
Let $x_{REF} \in\{\true,\false\}$ be the negation of the value of the reference sensor (Definition~\ref{def:ref}).
We use additional gates $g'_j$, $j\in\pm[n]$ where
\begin{itemize}
\item
if $g_j=g_{-j}=\false$ then $g'_{|j|}=\true$ and $g'_{-|j|}=\false$;
\item
if $j>0$ and $g_j=\true$ (so $g_{-j}=\false$) then $g'_j=g'_{-j}=\true \oplus x_{REF}$;
\item
if $j<0$ and $g_j=\true$ (so $g_{-j}=\false$) then $g'_j=g'_{-j}=\false \oplus x_{REF}$;
\end{itemize}
Each of the c-e agents $a_1,\ldots,a_n$ has $2$ value-blocks of value
$1/(2p^C)$ in region $R_1$, and each gate $g'_j$ of $C_1$ is able to select
the label of one of these value-blocks (recall that the boolean value at a
gate is represented by two positions that may be taken by the corresponding
cut, so that a block of value lies between these two positions.)
Figure~\ref{fig:gate-example} in Appendix \ref{sec:app_details} shows an example of how this feedback works.

\subsubsection{How $C_1$'s blanket-sensors affect the feedback mechanism}\label{sec:bsa}
Let $A_j \in \{\lplus,\lminus\}$ and let $A_{-j} \in \{\lplus,\lminus\}$, $A_{-j} \neq A_{j}$ be the complementary label. The blanket-sensor agents $b_{1,2},\ldots,b_{1,n}$ affect the output of
the circuit-encoders as follows:
\begin{enumerate}
\item\label{rule1} If none are active, the $2n$ outputs of $C_1$ are computed as
described in Section~\ref{sec:cea}.
\item\label{rule2} If $j$ is \emph{odd} and $b_{1,j}$ is active in direction $A_j$ then the output gates $g'_j,g'_{-j}$
are both set to the value that causes c-e agent $a_j$ to observe more $A_j$.
\item\label{rule3} If $j$ is \emph{even} and $b_{1,j}$ is active in direction $A_{j}$ then the output gates $g'_j,g'_{-j}$
are both set to the value that causes c-e agent $a_j$ to observe more $A_{-j}$.
\end{enumerate}
Rules \ref{rule2} and \ref{rule3} override Rule \ref{rule1}, which allocates
values that directly encode values output by $C_{VT}$. Note that the gadgetry of the circuit
can ensure that either an excess of $\lplus$ or an excess of $\lminus$ can be shown to the
corresponding c-e agent as feedback, as the circuit can convert the input value encoded by
the value gadget of the blanket-sensor agent in $R_1$ to either a ``right'' or ``left'' 
output position, depending on the parity of the index.
%Rules \ref{rule2} and \ref{rule3} are of course mutually exclusive for any $j$.
Also, if more than one blanket-sensor agent is active, they all affect their
corresponding gates. The reason for requiring blanket-sensors of different
parities to feedback different labels to the c-e agents is to be consistent
with the definition of ``consistent colours'', see Definition \ref{def:consistentcolour}.

Note that we do not define the behaviour of the blanket-sensor agents
in terms of the reference sensor. They essentially look for an imbalance between
$\lplus$ and $\lminus$ within some interval of length 2, and when they find a
sufficiently large imbalance, they force the circuit $C_1$ to show their associated c-e agent more of
the over-represented or under-represented label, depending on their parity, which can be done using gadgetry of \cite{FG17}.

\begin{paragraph}{Comment.}
Consider the operation of moving a cut from near the left-hand side of the c-e region
to the right-hand side, which corresponds to two points in the M\"{o}bius-simplex that
are close to each other via a path through the facets that have
been identified according to Definition \ref{def:simplex-domain}.
Suppose also that within the c-e region, we do not change the label of any point.
Then the blanket-sensor agents behave the same way: if some blanket-sensor agent
sees an excess of $\lplus$ in its interval then it will continue to see an excess of $\lplus$.
Regarding the (non-blanket) sensor agents, our reduction will make them ``want'' to produce opposite outputs, but
due to the flipping of $x_{REF}$, the reference sensor value, the final output
values produced by $g'_j$, $j\in\pm[n]$ are the same, and we will have continuity across this facet.
\end{paragraph}

\subsubsection{Construction of circuit-encoders $C_2,\ldots,C_{p^C}$}\label{sec:c2plus}

We next describe how the $p^C$ circuit-encoders differ from each other.
Each $C_i$ has a set of circuit-encoding agents ${\cal A}_i$, which
contains $C_i$'s sensor agents $\sensor_i$.
For $i\in[n]$ let ${\cal A}_i$ be the agents $a_{i,1},\ldots,a_{i,p}$ for
some polynomial $p$.
\begin{itemize}
\item
For all $i,j$, $\mu_{a_{i,j}}(x)=\mu_{a_{1,j}}(y)$ where $x$ and $y$ are
corresponding points in $R_i$ and $R_j$. By ``corresponding points'' here we mean
points that lie in the same distance from the left-endpoint of the respective intervals $R_i$
and $R_j$; see \cite{FG17}, Section 4.4.3 for the precise definition.
\item
For all $i,j$, all $x$ in the c-e region,
$\mu_{a_{i,j}}(x)$, is specified in Definition~\ref{def:blanket}.
\end{itemize}
The second of these items says that in the c-e region, the valuation
function of the agents that make up $C_i$ differ from those of
$C_1$ by having been shifted to the right by $\delt(i-1)$, where
this shift wraps around in the event that we shift beyond $n$
(the right-hand point of the c-e region).
In other respects, $C_i$ is an exact copy of $C_1$, save that $C_i$'s
internal circuitry lies in $R_i$ rather than $R_1$.

For each $C_i$, the c-e agents have a further 
$2n$ value-blocks of value $1/(2p^C)$ in region $R_i$, whose
labels are governed by the outputs produced by $C_i$ in the same way
as for $C_1$. Consequently we have the following observation.

\begin{observation}
The value that is labelled $\lplus$ held by any c-e agent $a_j$, is the average
of the output values that the $C_i$'s allocate to $a_j$.
If, say, all the $C_i$ receive inputs representing a point in the significant
region with label $\ell$, then $a_\ell$ observed an imbalance between
$\lplus$ and $\lminus$, but $a_j$ for $j\not=\ell$ will have $g'_j$ output
the opposite value to $g'_{-j}$, resulting in $a_j$'s value-blocks
receiving opposite labels.
\end{observation}

\subsection{An alternative coordinate system for the M\"{o}bius-simplex}\label{sec:coord-sys}

Recall that the M\"{o}bius-simplex $D$ is the $n$-simplex consisting of points $(x_1,\ldots,x_{n+1})$
whose components are non-negative and sum to 1. Furthermore, a typical point in $D$ is directly
encoded via the positions of $n$ cuts in the c-e region.

Here we specify a transformed coordinate system that is needed in order to encode instances of \nvhdt.
We will embed the hypercube-shaped domain of an instance of \nvhdt\ in a hypercube in the
transformed coordinates, and then use properties of the transformed coordinate system to
extend the labelling function to the rest of the domain in a way that does not introduce bogus solutions
(i.e. fixpoints of the extended function that lie outside the hypercube and do not encode solutions of \nvhdt).

Let $F_0$ be the set of points in $D$ of the form $(0,x_2,\ldots,x_n,0)$; thus $F_0$ is a $(n-2)$-face of $D$.
See Figure~\ref{fig:subspaces}.
For $\tau\in[0,1]$, let $\xb_\tau$ be the point
\[
\xb_\tau := \tau(1,0,\ldots,0)+(1-\tau)(0,\ldots,0,1) = (\tau,0,\ldots,0,1-\tau).
\]
(So, $\xb_0$ and $\xb_1$ are the endpoints of the 1-dimensional edge of $D$ that
is not contained in $F_0$.)
Let $D_\tau$ be the $(n-1)$-simplex consisting of convex combinations of $F_0$, and $\xb_\tau$.
Thus $D_0$ and $D_1$ are the two facets of $D$ that have been identified together
as in Definition \ref{def:simplex-domain}.

\begin{figure}
\center{
\begin{tikzpicture}[scale=0.6]

\def\vxtopx{15} \def\vxtopy{13}
\def\vxbotx{15} \def\vxboty{0}
\def\vxleftx{2} \def\vxlefty{10}
\def\vxblxtau{6} \def\vxblytau{4.667}
\def\vxblx{8} \def\vxbly{2}

%\node at(\vxblx,\vxbly){$\bullet$};

\draw[gray,ultra thick]({0.33*(\vxtopx+\vxbotx+\vxleftx)},{0.33*(\vxtopy+\vxboty+\vxlefty)})--({0.33*(\vxtopx+\vxbotx+\vxblx)},{0.33*(\vxtopy+\vxboty+\vxbly)}); %draw the axis

\draw[blue!50,fill=blue!50]({0.33*(\vxtopx+\vxbotx+\vxleftx)},{0.33*(\vxtopy+\vxboty+\vxlefty)}) circle (.8ex); %origin 0
\draw[blue!50,fill=blue!50]({0.33*(\vxtopx+\vxbotx+\vxblxtau)},{0.33*(\vxtopy+\vxboty+\vxblytau)}) circle (.8ex); %origin tau
\draw[blue!50,fill=blue!50]({0.33*(\vxtopx+\vxbotx+\vxblx)},{0.33*(\vxtopy+\vxboty+\vxbly)}) circle (.8ex); %origin 1

\node at(10.7,8){${\bf 0}_0$};
\node at(12.3,6.2){${\bf 0}_\tau$};
\node at(13,5.2){${\bf 0}_1$};

\draw[->](18,11)--(10.8,7.3);\node at(19,11.5){axis};

\draw[gray, ultra thick](2,10)--(8,2);\node at(1.5,10){$\xb_0$};\node at(1,9){$(0,...,0,1)$};
\node at(5.5,4.67){$\xb_\tau$};\node at(7.5,2){$\xb_1$};\node at(7,1){$(1,0,...,0)$};
\draw[blue!50, ultra thick](\vxleftx,\vxlefty)--(\vxtopx,\vxtopy)--(\vxbotx,\vxboty)--cycle;
\draw[white, line width=5pt,shorten >=0.4cm, shorten <=0.4cm](6,4.67)--(15,13);
\draw[blue!50, ultra thick](\vxblxtau,\vxblytau)--(\vxtopx,\vxtopy)--(\vxbotx,\vxboty)--cycle;
\draw[white,line width=5pt,shorten >=0.4cm, shorten <=0.4cm](8,2)--(15,13);
\draw[blue!50, ultra thick](\vxblx,\vxbly)--(\vxtopx,\vxtopy)--(\vxbotx,\vxboty)--cycle;

\node at(5,11.2){$D_0$};
\node at(9.9,5.8){$D_1$};
\node at(8,7.3){$D_\tau$};

\draw [decorate,decoration={brace,amplitude=4pt},xshift=0.5cm,yshift=0pt]
      (15,13) -- (15,0) node [midway,right,xshift=.1cm] {$F_0$}; 
\node[text width=2cm] at(19,5){$F_0$ is points of the form $(0,x_2,\ldots,x_n,0)$};

\end{tikzpicture}
\caption{\small{Subspaces of the M\"{o}bius-simplex $D$: $D_0$ is the triangle spanned by $\xb_0$ and $F_0$, and ${\bf 0}_0$ is its centre;
similarly for $D_\tau$ and $D_1$.}}\label{fig:subspaces}}
\end{figure}
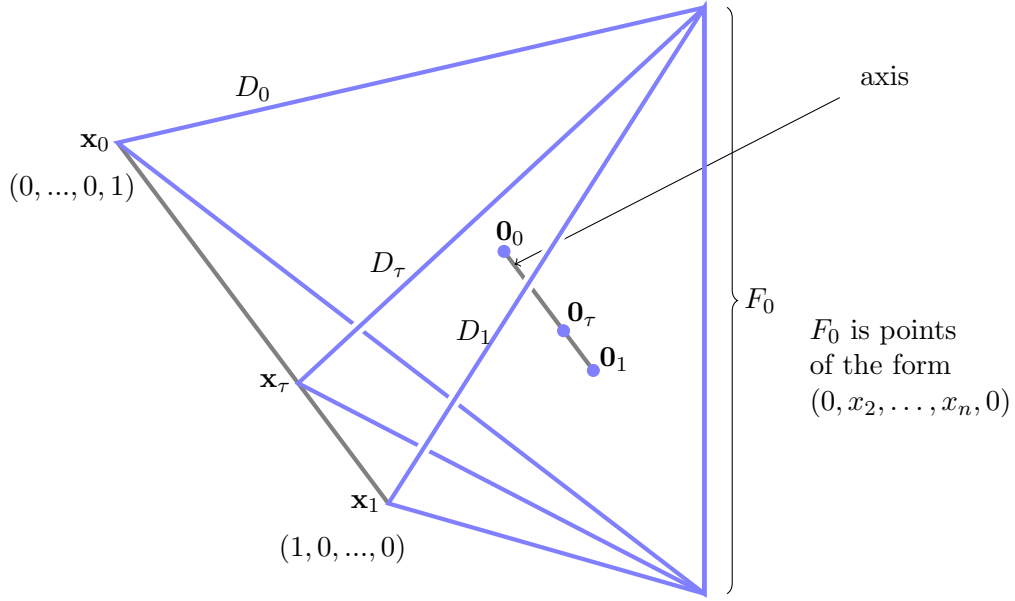

$D_\tau$ contains the point ${\bf 0}_\tau = (\tau/n,1/n,\ldots,1/n,(1-\tau)/n)$, which we regard as the origin of $D_\tau$.
The set of points $\{ {\bf 0}_\tau : 0\leq \tau \leq 1 \}$ will be referred to as the {\em axis};
it will transpire that all solutions must lie within an inverse polynomial distance from the axis
(in particular will be in the Significant Region).

We then refer to points in $D_\tau$ by means of the coordinates in a coordinate system
that itself is a linear function of $\tau$.
With respect to any fixed $\tau\in[0,1]$ we define $n-1$ vectors $(d^\tau_2,\ldots,d^\tau_n)$ as follows.
A key feature is that $(d^\tau_2,\ldots,d^\tau_n)$ form a basis of $D_\tau$ (so that with respect to the origin ${\bf 0}_\tau$,
any point in $D_\tau$ has unique coordinates).
The other key feature (Observation~\ref{obs:flip}) is that at $\tau=0$
the coordinate/directions are ``equal and opposite'' to the coordinates at $\tau=1$.
Also, the coordinate system varies suitably smoothly.\\

\noindent As a warm-up we start by considering $d^\tau_2$:
\[
d^\tau_2 := (1-\tau)(0,1,-1,0,\ldots,0)+\tau(-1,1,0,\ldots,0).
\]
$d^\tau_2$ consists of increasing the second coordinate at the expense of its
neighbours. For small $\tau$ we increase mainly at the expense of the third coordinate
and as $\tau$ increases, we increase the second coordinate more at the expense of
the first. Notice in particular that at $\tau=\frac{1}{2}$ we have
$d^\tau_2=(-\frac{1}{2},1,-\frac{1}{2},0,\ldots,0)$.\\

\noindent Generally, for $2\leq i\leq n$ we define
\begin{equation}\label{eq:dtaui}
d^\tau_i := (1-\tau)(\underbrace{0,\ldots,0}_{i-1~{\rm zeroes}},1,-1,\underbrace{0,\ldots,0}_{n-i})
       + \tau (\underbrace{0,\ldots,0}_{i-2~{\rm zeroes}},-1,1,\underbrace{0,\ldots,0}_{n-i+1})
\end{equation}
Thus, again this consists of the $i$-th coordinate increasing at the expense of
its neighbours, and we have in particular
\[
d^{\frac{1}{2}}_i = (\underbrace{0,\ldots,0}_{i-2~{\rm zeroes}},-\frac{1}{2},1,-\frac{1}{2},0,\ldots,0)
\]
For $i=2,\ldots,n$, define
\begin{equation}\label{eq:dtau-i}
d^\tau_{-i} := -d^\tau_i.
\end{equation}
Observation~\ref{obs:linearity} makes the important point that by linearity,
the vectors $d^\tau_i$, $i=2,\ldots,n$, can be used as a coordinate system to refer to points in $D_\tau$.

\begin{observation}\label{obs:linearity}
Any point $\xb$ in $D_\tau$ can be uniquely expressed as a sum
\begin{equation}\label{eq:transform}
\xb = {\bf 0}_\tau + \sum_{i=2}^n \alpha_i d^\tau_i.
\end{equation}
To see this, note first that $D_0$ is points in $D$ of the form $(0,x_2,\ldots,x_{n+1})$,
and $D_1$ is points of the form $(x_1,\ldots,x_n,0)$. Note that the observation certainly works for $\tau=0$ or $\tau=1$.
To see that it works for intermediate $\tau$, note that the vectors $d^\tau_i$ are linearly independent,
and the reason why they span $D_\tau$ is that any vector  $d^\tau_i$ is equal to $(1-\tau)$ multiplied by a vector in $D_0$,
added to $\tau$ multiplied by an equal-length vector in $D_1$. So these vectors do indeed lie in $D_\tau$.
\end{observation}

\begin{definition}\label{def:tc}
For a point $\xb\in D_\tau$ as in (\ref{eq:transform}) we say that the
{\em transformed coordinates} of $\xb$ are $(\alpha_2,\ldots,\alpha_n)$.
More generally, a point $\xb\in D$ can be expressed as $(\tau;\alpha_2,\ldots,\alpha_n)$,
where $\tau$ is chosen such that $\xb\in D_\tau$.
We use the following metric $\tilde{d}(\cdot,\cdot)$ on transformed coordinate vectors,
where similarly to (\ref{eq:metric}), $L_1$ denotes the standard $L_1$ distance on vectors.
\begin{equation}\label{eq:metric2}
\tilde{d}(\xb,\xb')=\min \Bigl( L_1(\xb,\xb'),
\min_{\zb,\zb' : \zb\equiv\zb'}(L_1(\xb,\zb)+L_1(\zb',\xb')) \Bigr)
\end{equation}
where $(0;\alpha_2,\ldots,\alpha_n)\equiv(1;-\alpha_2,\ldots,-\alpha_n)$.
\end{definition}

\begin{observation}\label{obs:flip}
With regard to Definition~\ref{def:tc}, consider two points
$\xb=(0;\alpha_2,\ldots,\alpha_n)$, $\xb'=(1;-\alpha_2,\ldots,-\alpha_n)$, that have been
equated with each other.
Assume these points are near the axis, specifically $|\alpha_j|<1/10n$ for all $j$.
Notice that
\begin{itemize}
\item the cuts in the c-e region for $\xb$ and $\xb'$ partition the c-e region in the same way.
\item (with reference to Figure~\ref{fig:dir-numbers}) when we
move from a point in $D_{1-\varepsilon}$ to a nearby point in $D_\varepsilon$,
for any $j\in\{2,\ldots,n\}$, the direction of increasing $\alpha_j$ segues {\em smoothly}
to the direction of {\em decreasing} $\alpha_{j}$.
\end{itemize}
Our assumption that $|\alpha_j|<1/10n$ ensures that cuts are fairly evenly-spaced,
and movement in any of the directions $d^\tau_j$ does not cause the cuts to cross each other.
\end{observation}
Proposition~\ref{prop:poly} says that if we perturb a point $\xb \in D$ that lies close to the axis, then the
total perturbation of the transformed coordinates of $\xb$ is polynomially related to the
total perturbation of the untransformed coordinates; $d$ and $\tilde{d}$ are polynomially related.

\begin{proposition}[Polynomial distance relation]\label{prop:poly}
There is some polynomial $p(n)$ such that
for all $\xb,\xb'\in D$ within Euclidean distance $1/10n^2$ of the axis,
letting $\tilde{\xb}$ and $\tilde{\xb}'$ be their transformed coordinates, and
letting $d,\tilde{d}$ be the metrics defined as in (\ref{eq:metric}),(\ref{eq:metric2}), we have
\[
\frac{1}{p(n)} \leq \frac{d(\xb,\xb')}{\tilde{d}(\tilde{\xb},\tilde{\xb}')} \leq p(n).
\]
\end{proposition}

\begin{proof}
In Subsection~\ref{sec:ct}, we show that for points near the axis
(i.e. within Euclidean distance $1/10n^2$ of the axis), the computation
of the coordinate transformation ---and its inverse--- have the property
that small perturbations of the input values lead to inverse-polynomial upper bounds
on the resulting perturbations of the output values.
Since we have this for both the transformation and its inverse, it follows
that there are also inverse-polynomial lower bounds on the resulting perturbations
of the output values.

The identification of transformed coordinates
$(0;\alpha_2,\ldots,\alpha_n)$ and $(1;-\alpha_2,\ldots,-\alpha_n)$
is of course equivalent to the identification of untransformed coordinates
$(0,x_1,\ldots,x_{n})$ and $(x_1,\ldots,x_n,0)$ in (\ref{eq:metric}).
If, say, $\xb$ and $\xb'$ are very close together due to being linked via
$\zb,\zb'$ for which $\zb\equiv\zb'$, then the transformed versions $\tilde{\zb},\tilde{\zb}'$
would cause $\tilde{d}(\tilde{\xb},\tilde{\xb}')$ to be very close.
So the ``polynomially related'' result for this space in which these two facets have not
been identified with each other, carries over to a ``polynomially related'' result
in which they have been identified with each other.
\end{proof}
Note that a similar result would hold if in the definitions of the metrics $d$ and $\tilde{d}$,
we replace the $L_1$ metric with, say, $L_2$ or $L_\infty$, since these are
polynomially related to $L_1$.

\bigskip
\noindent Note that Proposition~\ref{prop:poly} does not hold for all points in $D$;
the restriction to a neighbourhood of the axis is needed.
For points on $F_0$, all values of $\tau$ are equivalent, and for points close to $F_0$,
perturbed versions of them could result in large perturbations of $\tau$.\\

\noindent We show in the next section that the coordinate transformation, and its inverse, can be computed in polynomial time, for points in $D$
that are within some inverse polynomial distance from the axis.

\begin{figure}
\center{
\begin{tikzpicture}[scale=0.8]
\tikzstyle{xxx}=[dashed,thick]

\newcommand{\yytt}{15.5}
\newcommand{\yyt}{11.5}
\newcommand{\yym}{7.5}
\newcommand{\yyu}{3.5}
\newcommand{\yyuu}{-0.5}

\draw[thick,<->](0,19)--(16,19);
\node[fill=white]at(8,19){coordinate-encoding region $[0,4]$};
\draw[dashed](0,0)--(0,19);\draw[dashed](16,0)--(16,19);

\foreach \y in {\yyt+0.25,\yym+0.25,\yyu+0.25,\yytt+0.25,\yyuu+0.25}{\draw[ultra thick](0,\y)--(16,\y);}

\foreach \x in {0,4,8,12}{\draw[thick,red](\x,\yytt-0.5)--(\x,\yytt+2.5);}
\foreach \x in {1,5,9,13}{\draw[thick,red](\x,\yyt-0.5)--(\x,\yyt+2.5);}
\foreach \x in {2,6,10,14}{\draw[thick,red](\x,\yym-0.5)--(\x,\yym+2.5);}
\foreach \x in {3,7,11,15}{\draw[thick,red](\x,\yyu-0.5)--(\x,\yyu+2.5);}
\foreach \x in {4,8,12,16}{\draw[thick,red](\x,\yyuu-0.5)--(\x,\yyuu+2.5);}

\draw[->,ultra thick](3,2+\yytt)--(5,2+\yytt);
\node at(2.5,2+\yytt){$d^\tau_2$};
\draw[->,ultra thick](7,1.5+\yytt)--(9,1.5+\yytt);
\node at(6.5,1.5+\yytt){$d^\tau_3$};
\draw[->,ultra thick](11,1+\yytt)--(13,1+\yytt);
\node at(10.5,1+\yytt){$d^\tau_4$};

\draw[<-,ultra thick](0.75,2+\yyt)--(1.25,2+\yyt);\draw[->,ultra thick](4.25,2+\yyt)--(5.75,2+\yyt);
\draw[thick,dotted](1.25,2+\yyt)--(4.25,2+\yyt);
\node[fill=white] at(3,2+\yyt){$d^\tau_2$};
\draw[<-,ultra thick](4.75,1.5+\yyt)--(5.25,1.5+\yyt);\draw[->,ultra thick](8.25,1.5+\yyt)--(9.75,1.5+\yyt);
\draw[thick,dotted](5.25,1.5+\yyt)--(8.25,1.5+\yyt);
\node[fill=white] at(7,1.5+\yyt){$d^\tau_3$};
\draw[<-,ultra thick](8.75,1+\yyt)--(9.25,1+\yyt);\draw[->,ultra thick](12.25,1+\yyt)--(13.75,1+\yyt);
\draw[thick,dotted](9.25,1+\yyt)--(12.25,1+\yyt);
\node[fill=white] at(11,1+\yyt){$d^\tau_4$};

\draw[<-,ultra thick](1.5,2+\yym)--(2.5,2+\yym);\draw[->,ultra thick](5.5,2+\yym)--(6.5,2+\yym);
\draw[thick,dotted](2.5,2+\yym)--(5.5,2+\yym);
\node[fill=white] at(4,2+\yym){$d^\tau_2$};
\draw[<-,ultra thick](5.5,1.5+\yym)--(6.5,1.5+\yym);\draw[->,ultra thick](9.5,1.5+\yym)--(10.5,1.5+\yym);
\draw[thick,dotted](6.5,1.5+\yym)--(9.5,1.5+\yym);
\node[fill=white] at(8,1.5+\yym){$d^\tau_3$};
\draw[<-,ultra thick](9.5,1+\yym)--(10.5,1+\yym);\draw[->,ultra thick](13.5,1+\yym)--(14.5,1+\yym);
\draw[thick,dotted](10.5,1+\yym)--(13.5,1+\yym);
\node[fill=white] at(12,1+\yym){$d^\tau_4$};

\draw[<-,ultra thick](2.25,2+\yyu)--(3.75,2+\yyu);\draw[->,ultra thick](6.75,2+\yyu)--(7.25,2+\yyu);
\draw[thick,dotted](3.75,2+\yyu)--(6.75,2+\yyu);
\node[fill=white] at(5,2+\yyu){$d^\tau_2$};
\draw[<-,ultra thick](6.25,1.5+\yyu)--(7.75,1.5+\yyu);\draw[->,ultra thick](10.75,1.5+\yyu)--(11.25,1.5+\yyu);
\draw[thick,dotted](7.75,1.5+\yyu)--(10.75,1.5+\yyu);
\node[fill=white] at(9,1.5+\yyu){$d^\tau_3$};
\draw[<-,ultra thick](10.25,1+\yyu)--(11.75,1+\yyu);\draw[->,ultra thick](14.75,1+\yyu)--(15.25,1+\yyu);
\draw[thick,dotted](11.75,1+\yyu)--(14.75,1+\yyu);
\node[fill=white] at(13,1+\yyu){$d^\tau_4$};

\draw[<-,ultra thick](3,2+\yyuu)--(5,2+\yyuu);
\node at(2.5,2+\yyuu){$d^\tau_2$};
\draw[<-,ultra thick](7,1.5+\yyuu)--(9,1.5+\yyuu);
\node at(6.5,1.5+\yyuu){$d^\tau_3$};
\draw[<-,ultra thick](11,1+\yyuu)--(13,1+\yyuu);
\node at(10.5,1+\yyuu){$d^\tau_4$};

\node[text width=2cm]at(-1.5,\yytt+0.5){$D_\tau$\\ ($\tau \approx 0$)};
\node[text width=2cm]at(-1.5,\yyt+0.5){$D_\tau$\\ (small $\tau$)};
\node[text width=2cm]at(-1.5,\yym+0.5){$D_\tau$\\ ($\tau\approx \frac{1}{2}$)};
\node[text width=2cm]at(-1.5,\yyu+0.5){$D_\tau$\\ (large $\tau$)};
\node[text width=2cm]at(-1.5,\yyuu+0.5){$D_\tau$\\ ($\tau \approx 1$)};

\node[fill=white]at(2,\yytt-0.25){$x_2$};
\node[fill=white]at(6,\yytt-0.25){$x_3$};\node[fill=white]at(10,\yytt-0.25){$x_4$};
\node[fill=white]at(14,\yytt-0.25){$x_5$};

\node at(0.5,\yyt-0.25){$x_1$};\node[fill=white]at(3,\yyt-0.25){$x_2$};
\node[fill=white]at(7,\yyt-0.25){$x_3$};\node[fill=white]at(11,\yyt-0.25){$x_4$};
\node[fill=white]at(14.5,\yyt-0.25){$x_5$};

\node[fill=white]at(1,\yym-0.25){$x_1$};\node[fill=white]at(4,\yym-0.25){$x_2$};
\node[fill=white]at(8,\yym-0.25){$x_3$};\node[fill=white]at(12,\yym-0.25){$x_4$};
\node[fill=white]at(15,\yym-0.25){$x_5$};

\node[fill=white]at(1.5,\yyu-0.25){$x_1$};\node[fill=white]at(5,\yyu-0.25){$x_2$};
\node[fill=white]at(9,\yyu-0.25){$x_3$};\node[fill=white]at(13,\yyu-0.25){$x_4$};
\node at(15.5,\yyu-0.25){$x_5$};

\node[fill=white]at(2,\yyuu-0.25){$x_1$};\node[fill=white]at(6,\yyuu-0.25){$x_2$};
\node[fill=white]at(10,\yyuu-0.25){$x_3$};\node[fill=white]at(14,\yyuu-0.25){$x_4$};

\end{tikzpicture}
\caption{\small{The diagram shows (for $n=4$) sets of cuts (in red) that correspond to points on
the axis, for various values of $\tau$. It also shows how movements of the cuts correspond
to movement of a point in $D$ away from the axis. For example, for small $\tau$, a move
in direction $d^\tau_2$ corresponds to moving the second cut to the right and the first
only slightly to the left.
Generally, a movement in direction $d^\tau_i$ tends to increase $x_i$ at the expense of
$x_i$'s neighbours $x_{i-1}$ and $x_{i+1}$. As $\tau$ increases, the movement in direction $d^\tau_i$
tends increasingly to moving the cut to the left of interval $x_i$ to the left, as opposed to moving the
cut to the right of interval $x_i$ to the right. In the limit as $\tau$ approaches 0 from above,
the direction $d^\tau_i$ approaches the {\em negative} of the limit approached by $d^\tau_i$ 
when $\tau$ approaches 1 from below.}}\label{fig:dir-numbers}}
\end{figure}
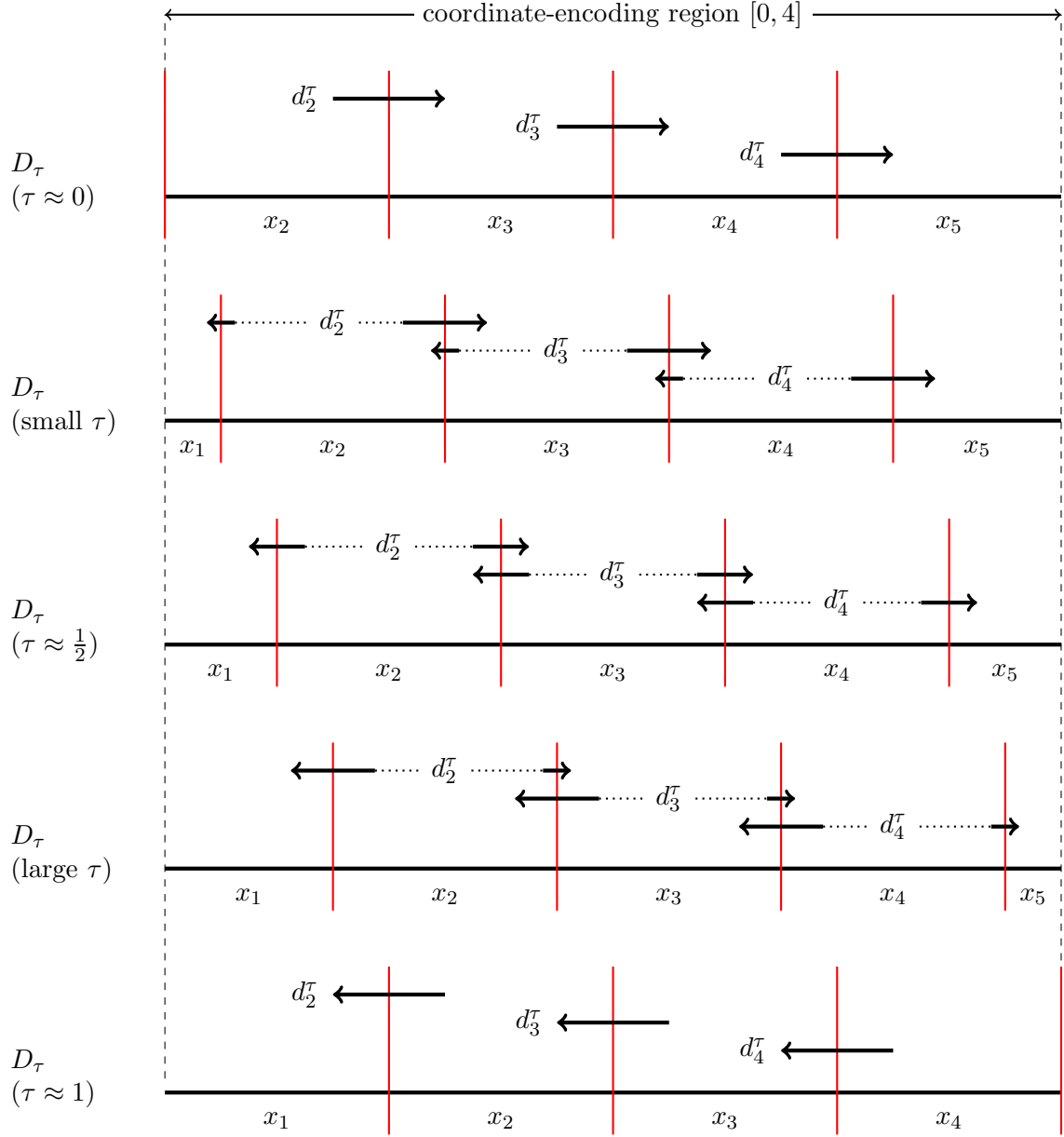

\subsubsection{Computation of the transformation, and its inverse}\label{sec:ct}
We verify here that (for points in the vicinity of the axis), our transformation may
be performed efficiently, and that small
perturbations of inputs lead to small perturbations of the outputs (in either direction).
The easy direction is the computation of $(x_1,\ldots,x_{n+1})$ from $(\tau;\alpha_2,\ldots,\alpha_n)$. 
Recall that a point $\xb$ on the original domain can be expressed in terms of the transformed coordinates $(\tau;\alpha_2,\ldots,\alpha_n)$
and the origin ${\bf 0}_\tau = (\frac{\tau}{n}, \frac{1}{n}, \ldots, \frac{1}{n},\frac{1-\tau}{n})$ as 
$\xb = {\bf 0}_\tau + \sum_{i=2}^n \alpha_i d^\tau_i$. Therefore we have:
\begin{eqnarray*}
	x_1 &=& \frac{\tau}{n} - \tau.\alpha_2 \\
	x_2 &=& \frac{1+\tau}{n} + (1-\tau)\alpha_2 - \tau \alpha_3 - x_1 \\ 
	x_3 &=& \frac{2+\tau}{n} + (1-\tau)\alpha_3 - \tau \alpha_4 - (x_1 + x_2) \\
	&\vdots& \\
	x_{n} &=& \frac{n-1+\tau}{n} + (1-\tau)\alpha_n - \sum_{i=1}^{n-1} x_i \\
	x_{n+1} &=& 1 - \sum_{i=1}^{n} x_i
\end{eqnarray*}
In the other direction, given $(x_1,\ldots,x_{n+1})$ we first compute the value of $\tau$ for the transformed
coordinate system: Note that $(x_1,\ldots,x_{n+1})$ must be a convex combination of $x_\tau$ and $F_0$ (where recall that $x_\tau=(\tau,0,\ldots,0,1-\tau)$ and $F_0$ is points of the form $(0,x_2,\ldots,x_n,0)$), therefore $\tau$ can be computed as the solution to the equation 
\[ \frac{\tau}{1-\tau} = \frac{x_1}{x_{n+1}}.\]
Note that the dependence of $\tau$ on $x_1$ and $x_{n+1}$ is not
excessively sensitive near the axis, since $x_1+x_{n+1}$ is close to $1/n$. Having computed $\tau\in[0,1]$, we 
could simply solve the equations above (the ones used to compute $x_1,\ldots,x_{n+1}$) for $\alpha_2$, $\alpha_3$ and so on successively, using the derived formulas for $\alpha_i$ for the computation of $\alpha_{i+1}$ and express each $\alpha_i$ as a function of only $\tau$ and the values $x_1, x_2, \ldots x_{n+1}$. However, in the extremal cases of $\tau=0$ and $\tau=1$, some of the $\alpha_i$ values might ``disappear''; for example, for $\tau=0$, expressing $\alpha_3$ in terms of only $\tau$, $x_1$ and $x_2$ is not possible, since  for $\tau=0$ we do not obtain a formula for $\alpha_2$ to substitute into the equation for $x_2$. To remedy this, we consider two cases:\\

\noindent \textbf{Case 1:} $\tau\geq\frac{1}{2}$. We compute $\alpha_2, \ldots, \alpha_n$ as follows:
\begin{eqnarray*}
	\alpha_2 &=& \frac{1}{n} - \frac{1}{\tau}\cdot x_1 \\
	\alpha_3 &=& \frac{1+1/\tau}{n} + \frac{(1-\tau)}{\tau}\cdot \alpha_2 - \frac{1}{\tau} \cdot (x_1+x_2)\\
	&\vdots&
\end{eqnarray*}
and so on for $\alpha_4, \ldots, \alpha_n$.\\

\noindent \textbf{Case 2:} $\tau\leq\frac{1}{2}$. We compute $\alpha_2,\ldots,\alpha_n$ starting at the opposite end:
\begin{eqnarray*}
	\alpha_n &=&  - \frac{n-1+\tau}{n(1-\tau)} + \frac{\sum_{i=1}^{n} x_i}{1-\tau} \\ 
	\alpha_{n-1} &=& \frac{\tau}{1-\tau}\cdot \alpha_n - \frac{n-2+\tau}{n(1-\tau)} + \frac{\sum_{i=1}^{n-1} x_i}{1-\tau}\\
	&\vdots&
\end{eqnarray*}	
and so on for $\alpha_{n-2}, \ldots, \alpha_2$.

Note that inverse polynomial-size perturbations of the $x_i$ lead to inverse
polynomial-size perturbations of the transformed coordinates. 
As a sanity check, note that at the boundary (points with $\tau=0$ are the same
as points with $\tau=1$), if we move a cut at the LHS to the RHS
(so $(0,x_2,\ldots,x_{n+1})$ becomes $(x_2,\ldots,x_{n+1},0)$),
it can be checked that the $\alpha_i$ get negated.

\bigskip
Note that these computations should be done with a precision (or rounding error) polynomially smaller than $\delt$.

\subsection{A (poly-time computable) partial colouring function $f:D\rightarrow \{-1,0,1\}^n$}\label{sec:colour-f}

This section defines a partial function $f:D\rightarrow \{-1,0,1\}^n$
($D$ being the $n$-dimensional M\"{o}bius-simplex (Definition~\ref{def:simplex-domain})).
$f$ is constructed in polynomial time based on an instance $I_{VT}$ of \nvhdt\ in
$n$ dimensions, defined using circuit $C_{VT}$.
$f$ is defined in the Significant Region (Definition~\ref{def:sig-region}) which is the
set of points where no blanket-sensor agents (Definition \ref{def:blanket}) are active,
thus it includes the twisted tunnel $T$ (Definition~\ref{def:twisted}).
$f$ is computable by a circuit $C$, that is used to define
the operations of the circuit-encoders in a derived instance of \ch.
Within the Significant Region, $f$ determines the outputs of the circuit-encoders.

The function $f$ maps a point $\xb$ in the Significant Region to a vector of length $n$, $e_f(\xb)$, where $e_f(\xb)_j = 1$ means that the point receives colour $j$, $e_f(\xb)_j = -1$ means that the point receives colour $-j$ and $e_f(\xb)_j = 0$ means that the point does not receives colour $j$ or $-j$. In general, $e_f(\xb)$ may have multiple non-zero entries. We will use the term \emph{the colour of $\xb$} for points that only receive a single colour (and therefore their outputs are vectors with only one non-zero entry).  

The function $f$ will have a corresponding vector-valued function $f'$ (Section~\ref{sec:borsuk-f})
that more closely represents choices of labels $\lplus$/$\lminus$ that the circuit shows to the c-e agents.
We will do this in such a way that no ``bogus'' solutions result from the transition to parts of $D$
where blanket-sensor agents are active.
By construction, there are no solutions where blanket-sensor agents are active, so all solutions occur where $f$ is defined.

In Section~\ref{sec:borsuk-f} we then define a vector-valued ``Borsuk-Ulam style'' function $F$ in terms of $f$.
Letting $I_{VT}$ be an instance of \nvhdt, $F(\xb)$ will be
approximately zero iff given $\xb$, we can derive an approximate consensus-halving solution to $I_{VT}$.
It will be shown that approximate zeroes of $F$ provide solutions to $I_{VT}$.

Recall that $D$ is the set of points $(x_1,\ldots,x_{n+1})$ whose components
are non-negative and sum to 1.
And, the ``Significant Region'' (Definition~\ref{def:sig-region}) of $D$ consists of points $(x_1,\ldots,x_{n+1})$ for
which coordinates $x_i$ ($2\leq i\leq n-1$) differ from $1/n$ by at most an inverse
polynomial $\delta^w = 1/p^w(n)$ (Proposition~\ref{prop:s-r-radius}).
($\delta^w$ represents an upper bound on the thickness of the Significant Region.)

Let $B$ be the $n$-dimensional ``box'' associated with $I_{VT}$
(recall $I_{VT}$ is represented by circuit $C_{VT}$ that maps points in $B$ to $\pm[n]$).
We embed a copy of $B$ in $D$ as follows.
Recall the way facets of $B$ are coloured in Definition~\ref{def:nvhdt}.
Let $(x_1,\ldots,x_n)$ denote a typical point in $B$, and assume that the facets
of $B$ with maximum and minimum $x_1$ (i.e. $x_1=1$ and $x_1=-1$ respectively)
are the panchromatic facets of $B$
(as in Definition~\ref{def:nvhdt}), and for $i\geq 2$ the facet of $B$ with maximum $x_i$ ($x_i=1$)
consists of points that do not have colour $i$, and the the facet of $B$
with minimum $x_i$ ($x_i=-1$) consists of points that do not have colour $-i$.

\begin{definition}\label{def:twisted}
The {\em twisted tunnel} $T$ is defined as follows.
The {\em axis} of $T$ is the set of all points ${\bf 0}_\tau$ as defined in Section~\ref{sec:coord-sys}.
The twisted tunnel is the set of all points with transformed coordinates
$(\tau;\alpha_2,\ldots,\alpha_n)$ such that for all $i$, $|\alpha_i|<\delta^T$.
Note that $\delta^T$ is an inverse polynomial quantity sufficiently small that $T$ is a
subset of the Significant Region; this is achieved since by definition, $\delta^T$ is
polynomially smaller than $\delta^w$ of Proposition~\ref{prop:s-r-radius}.
Thus, $T$ has (with respect to the transformed coordinates)
a $(n-1)$-cube-shaped intersection with any $D_\tau$.
\end{definition}
We define the behaviour of $f$ over the Significant Region (Definition \ref{def:sig-region})
in 3 stages, as follows.

\begin{enumerate}
\item{
{\bf Embedding $B$ in $D$ (recall $B=[-1,1]^n$)}

A point $\xb=(x_1,\ldots,x_n)$ in $B$ is mapped to a point $g(\xb)$ in $D$ as follows.
%$\delta^T>0$ is defined as in Definition~\ref{def:twisted} and is of inverse polynomial size.
$g(\xb)$ lies in $D_\tau$, where we choose $\tau=\frac{1}{2}+\delta^T\cdot x_1$.
Then (noting (\ref{eq:transform})) we set
$g(\xb)$ equal to ${\bf 0}_\tau + \sum_{i=2}^n \delta^T\cdot x_i d^\tau_i$
(i.e. $g(\xb)$ has transformed coordinates $(\frac{1}{2}+\delta^T x_1;\delta^T x_2,\ldots,\delta^T x_n)$).
$g(\xb)$ will receive a single colour; the colour of $g(\xb)$ --- i.e. the non-zero entry of $f(g(\xb))$ --- is set equal to the colour allocated to $\xb$ in $B$ by $I_{VT}$.
(Notice that the centre of $B$ is mapped to $(1/2n,1/n,\ldots,1/n,1/2n)$,
which is the origin of $D_{\frac{1}{2}}$, and the centre of the Significant Region.
This point has (recalling Definition~\ref{def:tc}) transformed coordinates $(\frac{1}{2};0,\ldots,0)$
where the first entry is the value of $\tau$.)
}
\item\label{item-f-strips}{
{\bf Extending $f$ to be defined on $T$}

We also colour other points in $T$ as follows --- these will also receive single colours.
Suppose $\yb$ belongs to $D_\tau$, where $\tau<\frac{1}{2}-\delta^T$ or $\tau>\frac{1}{2}+\delta^T$.
According to (\ref{eq:transform}), $\yb = {\bf 0}_\tau + \sum_{i=2}^n \alpha_i d^\tau_i$, and
$\yb$ has transformed coordinates $(\tau;\alpha_2,\ldots,\alpha_n)$.
Suppose all the $\alpha_i$ lie in the range $[-\delta^T,\delta^T]$.
Then if $\tau<\frac{1}{2}-\delta^T$, we set the colour of $\yb$ to
the colour of a point $\yb'=(\frac{1}{2}-\delta^T;\alpha_2,\ldots,\alpha_n)$.
Thus $\yb'\in D_{\frac{1}{2}-\delta^T}$, and the other transformed coordinates
(Definition \ref{def:tc}) are the same for $\yb$ and for $\yb'$.
We do a similar thing for points in $D_\tau$ for $\tau > \frac{1}{2}+\delta^T$.
That is, if $\yb$ has transformed coordinates $(\tau;\alpha_2,\ldots,\alpha_n)$ where
$\tau > \frac{1}{2}+\delta^T$ and the $\alpha_i$ are all at
most $\delta^T$ in absolute value, then $\yb$ gets the same colour as a point $\yb'$ whose
transformed coordinates are $(\frac{1}{2}+\delta^T;\alpha_2,\ldots,\alpha_n)$.
}
\item\label{item-f-sr}{
{\bf Extending $f$ to the Significant Region}

The Significant Region (Definition \ref{def:sig-region}) is points in $D$ where no blanket-sensor agents
are active, a subset of points that are close to the axis in the sense of Proposition~\ref{prop:s-r-radius}.
Consider $\xb\in D\setminus T$ with transformed coordinates $(\tau;\alpha_2,\ldots,\alpha_n)$.
\begin{enumerate}
\item
For each $j\in\{2,\ldots,n\}$ if $\alpha_j>\delta^T$ then $\xb$ gets colour $-j$;
\item
For each $j\in\{2,\ldots,n\}$ if $\alpha_j<-\delta^T$ then $\xb$ gets colour $j$;
\item
these are not mutually exclusive, $\xb$ gets at least one colour, possibly more.
\end{enumerate}
Notice that (within the subspace $D_\tau$) the side(s) of the twisted tunnel $T$ closest to $\xb$ 
is guaranteed not to be opposite to any colour of $\xb$.

For a subset $S$ of colours, let $R(S)$ be the region with colours in $S$.
We call these the ``outer regions''.
}
\end{enumerate}
Proposition~\ref{prop:strips} notes that when colour-regions meet each other at
opposite ends of $T$ (which have been identified with each other according to the
definition of the Significant Region), they will have equal and opposite colours.

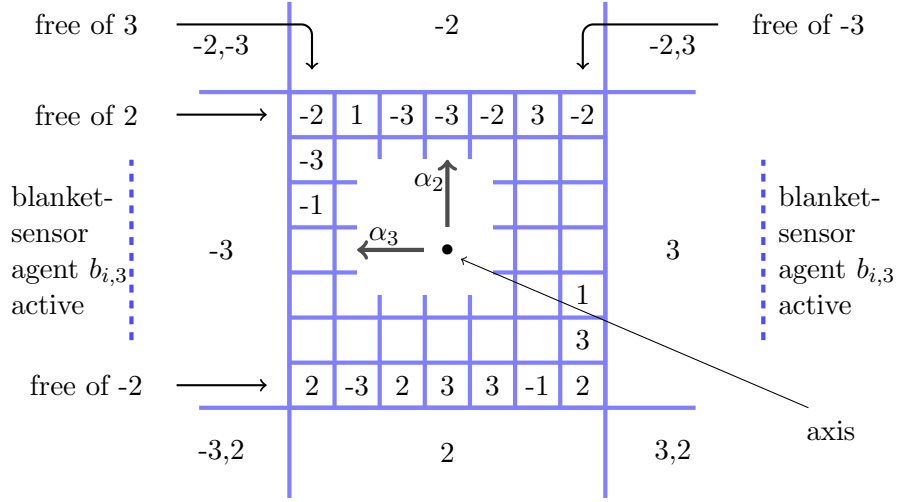
\begin{figure}
\center{
\begin{tikzpicture}[scale=0.6]

\foreach \x in {0,1,2,3,4,5,6,7}{
\draw[blue!50,ultra thick](-3.5,-0.5+\x)--(3.5,-0.5+\x);
}
\foreach \x in {-3.5,-2.5,-1.5,-0.5,0.5,1.5,2.5,3.5}{
\draw[blue!50,ultra thick](\x,-0.5)--(\x,6.5);
}

\draw[blue!50, ultra thick](-3.5,8.5)--(-3.5,-2.5);
\draw[blue!50, ultra thick](3.5,8.5)--(3.5,-2.5);
\draw[blue!50, ultra thick](-5.5,6.5)--(5.5,6.5);
\draw[blue!50, ultra thick](-5.5,-0.5)--(5.5,-0.5);

\node at(-5,7.5){-2,-3};\node at(5,7.5){-2,3};\node at(-5,-1.5){-3,2};\node at(5,-1.5){3,2};

\node at(0,8){-2};\node at(0,-1.5){2};
\node at(-5,3){-3};\node at(5,3){3};

\node at(-3,6){-2};\node at(-2,6){1};\node at(-1,6){-3};\node at(0,6){-3};\node at(1,6){-2};\node at(2,6){3};\node at(3,6){-2};
\node at(-3,5){-3};
\node at(-3,4){-1};

\node at(3,2){1};
\node at(3,1){3};
\node at(-3,0){2};\node at(-2,0){-3};\node at(-1,0){2};\node at(0,0){3};\node at(1,0){3};\node at(2,0){-1};\node at(3,0){2};

\node at(-8,6){free of 2};\draw[thick,->](-6,6)--(-4,6);
\node at(-8,0){free of -2};\draw[thick,->](-6,0)--(-4,0);

\node at(-8,8){free of 3};\draw[thick,rounded corners,->](-6,8)--(-3,8)--(-3,7);
\node at(8,8){free of -3};\draw[thick,rounded corners,->](6,8)--(3,8)--(3,7);

\draw[blue!70,ultra thick,dashed](7,1)--(7,5);
\draw[blue!70,ultra thick,dashed](-7,1)--(-7,5);
%\draw[blue!70,ultra thick,dashed] (7,0) arc (-10:10:16);
%\draw[blue!70,ultra thick,dashed] (-7,0) arc (190:170:16);

\node[text width=2cm] at(9,3){blanket-sensor agent $b_{i,3}$ active};
\node[text width=2cm] at(-8,3){blanket-sensor agent $b_{i,3}$ active};

\fill[color=white] (1,2)--(1,5)--(-2,5)--(-2,2)--cycle;
\draw[black!70, ultra thick,<->](-2,3)--(0,3)--(0,5);\node at(-1.4,3.3){$\alpha_3$};\node at(-0.4,4.5){$\alpha_2$};

\fill[color=white] (-0.5,2.5)--(0.5,2.5)--(0.5,3.5)--(-0.5,3.5)--cycle;\node at(0,3){\small $\bullet$};
\draw[->](8,-0.5)--(0.3,2.8);\node at(8.5,-1){axis};

\end{tikzpicture}
\caption{\small{Cross-section of the twisted tunnel (for $n=3$), with examples of
possible labels of regions. Note that the outer regions of the Significant Region
are not adjacent to any cubelet having an opposite colour to that outer region.
For example, the left-hand column is free of cubelets with colour $3$, and
is adjacent to the outer region with colour $-3$.}\label{fig:cross-section}}}
\end{figure}

\subsection{How to compute a solution to \nvhdt\ from a solution to \ch}\label{sec:backward}

Suppose we have a solution $S_{CH}$ to an instance $I_{CH}$ of \ch, derived by our
reduction from an instance $I_{VT}$ of \nvhdt.

Let $\xb$ be the point in the M\"{o}bius-simplex represented by the c-e cuts of $S_{CH}$.
Proposition~\ref{prop:s-r-radius} already tells us that $\xb$ must lie within some inverse-polynomial
distance of the axis, since if not, some blanket-sensor agent will be active.
We prove in Section~\ref{sec:prove-reduction} that $\xb$ must lie within,
or very close to, the twisted tunnel.

From this we identify two colour-regions that have equal and opposite colours, as follows.
Let $\xb$ have transformed coordinates $(\tau;\alpha_2,\ldots,\alpha_n)$, which
can be computed with inverse-polynomial precision from the c-e cuts.

If $\xb$ occurs within the embedded copy of $B$ (Section~\ref{sec:colour-f})
then identify two circuit-encoders $C_i$ and $C_{i'}$ that both receive reliable inputs
and have equal and opposite outputs. 
The proof of Proposition~\ref{prop:tt} tells us that this is always possible.

Now we have two points $\xb'$ and $\xb''$ within distance $\delt$ of each other,
that lie in oppositely coloured cubelets.
With respect to transformed coordinates, $\xb$ and $\xb'$ are within some
distance $\delnt$ that we can assume by Proposition~\ref{prop:poly}
to be much smaller than the widths of the cubelets and other colour-regions.
So these two cubelets are adjacent and we are done.

If $\xb$ lies in $T$ but not in the embedded copy of $B$, then we similarly
find two distinct colour-regions that are adjacent and with opposite colours.
Since these colour-regions are just extensions of the cubelets that lie
on the panchromatic facets of the embedded instance of \nvhdt, we are done.
Formally, if $S_{CH}$ represents a point with transformed coordinates
$(\tau;\alpha_2,\ldots,\alpha_n)$ where $\tau>\frac{1}{2}+\delta^T$,
we take the point $(\frac{1}{2}+\delta^T;\alpha_2,\ldots,\alpha_n)$,
and similarly,
if $S_{CH}$ represents a point with transformed coordinates
$(\tau;\alpha_2,\ldots,\alpha_n)$ where $\tau<\frac{1}{2}-\delta^T$,
we take the point $(\frac{1}{2}-\delta^T;\alpha_2,\ldots,\alpha_n)$.

\section{Completing the Proof of Theorem~\ref{thm:main}}\label{sec:prove-reduction}

Let $I_{VT}$ be an instance of \nvhdt\ in $n$ dimensions,
given in terms of circuit $C_{VT}$; suppose \ch\ instance $I_{CH}$ 
is derived from it by our reduction.
Recall that $\varepsilon=\delt/10$. We show that any $\varepsilon$-approximate solution $S_{CH}$
to $I_{CH}$ allows a solution to $I_{VT}$ to be recovered using Section~\ref{sec:backward}.

In $S_{CH}$, only the c-e cuts may lie in the c-e region (Observation~\ref{obs:cer}),
and every other cut $c(a)$ for $a$ not a c-e agent, must lie in some interval outside the c-e region
(in order for $a$'s value to be evenly split).
It follows from Proposition~\ref{prop:s-r-radius} that at least $n-1$ c-e cuts must lie in the c-e region,
and they are evenly-spaced (the gaps between them differ from 1 by an inverse polynomial).
The remaining cut may occur elsewhere, in which case it becomes what we called
a ``stray cut'' in~\cite{FG17}, and in that case, the ``double negative lemma'' of~\cite{FG17}
may be applied to prove that it has little effect on the quality of a solution:
it causes a single circuit-encoder to have unreliable output.

Recalling Observation~\ref{obs:average}, the c-e cuts of $S_{CH}$ encode
a collection of $p^C$ points $\xb_1,\ldots,\xb_{p^C}$ in the M\"{o}bius-simplex
where $d(\xb_i,\xb_j)\leq \delt$ ($d$ is the metric defined in (\ref{eq:metric})).
In transformed coordinates we have $\tilde{d}(\xb_i,\xb_j)\leq \delnt$
($\tilde{d}$ as in (\ref{eq:metric2})), and by Proposition~\ref{prop:poly}, $\delnt$
is much smaller than other inverse-polynomial quantities that we work with.
Recall also that at most $n$ of these points are incorrectly labelled since all but
$n$ of the circuit-encoders receive reliable inputs.
Alternatively, up to $n-1$ circuit-encoders receive unreliable input and one circuit-encoder
is affected by the stray cut.
We proceed by case analysis.
Note that there is an inverse-polynomial gap between the twisted tunnel and the
boundary of the Significant Region, which is much larger than $\delt$.
So the cases to consider are:
\begin{itemize}
\item[-] Most or all of the points are in the twisted tunnel. In that case it will be
proved that the procedure of Section~\ref{sec:backward} identifies a solution to \nvhdt;
see Sections \ref{sec:borsuk-f}, \ref{sec:solution}.
\item[-] Most or all are in the outer regions. In that case we are not at a solution since
the colours cannot cancel each other.
\item[-] Some of the points are outside the Significant Region. In that case we are
far from the twisted tunnel, and Section~\ref{sec:bdry1} argues that no solution
is possible here.
\item[-] All points are outside the Significant Region. Then, we are certainly far from
a solution since various blanket-sensor agents will be active, and the points are
so close together (within $\delt$) that the directions in which they are active, cannot cancel.
\end{itemize}

\subsection{A Borsuk-Ulam-style function $F:D\rightarrow[-1,1]^n$}\label{sec:borsuk-f}

Recall that $D$ denotes the $n$-dimensional M\"{o}bius-simplex (Definition~\ref{def:simplex-domain}).
We start by defining a function $f':D\rightarrow [-1,1]^n$ based on $f$ defined as in Section~\ref{sec:colour-f}.
$f'$ simulates the effect of the blanket-sensor agents as described in Section~\ref{sec:bsa}. Let $\xb\in D$.

\begin{enumerate}
\item
{\bf When no blanket-sensor agents are active at $\xb$.}
Here we are in the Significant Region.
\begin{itemize}
\item
In $T$, $f'$ behaves like $f$ in the sense that
if $f$ assigns the colour $i$ to $\xb$ (recall that it assigns a single colour to points in $T$), then set $f'(\xb):=\eb_i$ if $i>0$, $f'(\xb):=-\eb_{|i|}$ for $i<0$.
\item
Outside $T$, in outer region $R(S)$, $f'(\xb):=\sum_{i\in S,i>0} \eb_i +\sum_{i\in S,i<0} -\eb_{|i|}$.
\end{itemize}
\item
{\bf When one or more blanket-sensor agents are active.}
If the $j$-th blanket-sensor agent of $C_1$, $b_{1,j}$ is active towards $\lplus$ (respectively $\lminus$) then the
$j$-th entry of $f'(\xb)$ is set to $1$ if $j$ is odd and to $-1$ if $j$ is even (respectively, to $-1$ if $j$ is odd and $1$ is $j$ is even).
This is done for all active blanket-sensor agents, thus $f'(\xb)$ can contain multiple 1's and $-1$'s.
%\item
%{\bf At a boundary.}
%At any point $\xb$ where the value of $f$ changes ($\xb$ is the limit of sequences of points
%where $f$ can take different values), $f'$ may take any value, and we allow it be the adversarially chosen.
\end{enumerate}
The following points are similar to Observation~\ref{obs:average}:

\begin{observation}\label{obs:61}
Suppose that circuit-encoder $C_i$ (some $i\in [p^C]$) of $I_{CH}$ receives reliable inputs.
(Observation~\ref{obs:reliable} tells us that at most $n$ of them fail to receive reliable inputs.)
Then $C_i$ computes $f'$ at a point within distance $\delt$ from the $\xb\in D$ encoded
by the c-e cuts, in the sense that the value observed by each c-e agent $a_j$
that is labelled by $\lplus$, minus the amount labelled $\lminus$, restricted
to that part of $a_j$'s value that lies in $R_i$ and so is governed by the output of $C_i$,
is the $j$-th component of $f'$.
\end{observation}

This follows from the construction of Section~\ref{sec:cea} and the association of boolean
values $\true,\false$ with the labels $\lplus$ and $\lminus$.
For $\xb \in D$, $F(\xb)$ is the average of the outputs of the $C_i$; Proposition~\ref{prop:62} provides the details.

\begin{proposition}\label{prop:62}
$I_{CH}$ computes a function $F$ in the following sense.
Let $\xb$ be the point encoded by the c-e agents.
Suppose all agents other than the c-e agents have error (i.e. discrepancy between
$\lplus$ and $\lminus$ that they observe) at most $\varepsilon$.
Then the error of the c-e agents is within additive distance $1/n$ from
the average value of $f'$, averaged over a set of points all within $\delt$ of $\xb$.
\end{proposition}

\begin{proof}
We put together various observations about the way $I_{CH}$ is constructed.
Observation~\ref{obs:average} told us that the values observed by the c-e agents are the
average of a set of points all within distance $\delt$ of each other.
The additive distance $1/n$ results from the existence of up to $n$ circuit-encoders
that either fail to receive good inputs (Observation~\ref{obs:reliable}), or are affected by the stray cut,
taken in conjunction with the fact that we average over $p^C$ points,
where $p^C$ can be taken to be at least $2n^2$.
\end{proof}

\begin{proposition}\label{obs:invpoly}
$\delt$ can be chosen to be sufficiently small (but still inverse-polynomial) that
given a set of $p^C$ points $\xb_1,\ldots,\xb_{p^C}\in D$ within distance $\delt$ of each other,
when we compute their transformed coordinates $\yb_1,\ldots,\yb_{p^C}$, we have:\smallskip

\noindent Every pair of points $\yb,\yb' \in \{ \yb_1,\ldots,\yb_{p^C} \}$ has the property that they either lie in the same
colour-region, or adjacent colour-regions (where a ``colour-region'' is one of
monochromatic regions of Section~\ref{sec:colour-f}), or one of the outer regions.
\end{proposition}

\begin{proof}
In identifying which colour-region a point with transformed coordinates $\yb$ belongs to,
for the colour-regions in the twisted tunnel $T$, we compare coordinates with certain threshold values.
These threshold values differ from each other by inverse-polynomial amounts,
and the smallest difference between any pair of them is inverse-polynomial.
Applying Proposition~\ref{prop:poly}, we can keep the $\yb_i$ closer to each other than this.
(Colour-regions lie in the Significant Region, so these points lie within $1/10n^2$ of
the axis, so Proposition~\ref{prop:poly} is applicable.)

This applies also to the outer regions $R(S)$ for sets of colours $S$.
Identifying which $R(S)$ a point $\yb$ belongs to, uses the same information on comparisons of its coordinates with inverse-polynomials.
\end{proof}

\begin{paragraph}{Observations on $F$}
\begin{itemize}
\item[-] We call $F$ a Borsuk-Ulam-style function --- The suffix ``style'' is to note
that we define a kind of function that has desirable properties similar
to those of a Borsuk-Ulam function, but for example the domain of
the function is $D$ as opposed to a sphere. Also, the function $F$ is ``approximately
Lipschitz'' rather then truly continuous, which is good enough for our purposes.
\item[-] $|F(\xb)|\leq \varepsilon$ (here, $|F|$
denotes the $L_\infty$ or ``maximum'' norm of $F$) iff $\xb$ encodes an approximate \ch\ solution.
Regarding this point, $F$ is not simulating a Borsuk-Ulam function,
but rather simulating a function consisting of the difference between the values taken
by a Borsuk-Ulam function, at two antipodal points.
\end{itemize}

\end{paragraph}

\subsection{Encoding the output of $F$ with a \ch\ solution}\label{sec:solution}

\begin{proposition}\label{prop:tt}
Let $S_{CH}$ be an $\varepsilon$-approximate solution to $I_{CH}$.
Suppose that the c-e cuts of $S_{CH}$ represent a point $\xb$ that lies in the twisted tunnel.
Then we can reconstruct a solution to $I_{VT}$ in polynomial time.
\end{proposition}
Recall that $S_{CH}$, $I_{CH}$ and $I_{VT}$ and $\varepsilon$ are as introduced at
the start of Section~\ref{sec:prove-reduction}.

\begin{proof}
Observation \ref{obs:61} tells us that if a circuit-encoder $C_i$ receives reliable
inputs, it outputs the colour of a point in the M\"{o}bius-simplex that lies
within $\delt$ of $\xb$.

We note next that the feedback received by the c-e agents in $I_{CH}$ corresponds
to the average (over $i\in[p^C]$) of the feedback received by the individual circuit-encoders $C_i$.
In detail, the $i$-th coordinate of a typical point in $B=[-1,1]^n$ is obtained by taking c-e agent $a_i$, and
(given any attempt at a consensus-halving solution $S$)
subtracting $a_i$'s value for the parts of the consensus-halving
domain labelled $\lminus$ according to $S$, from those labelled $\lplus$.
The resulting point is at the centre (or origin) of $B$ iff the c-e agents have balanced allocations
of $\lplus$ and $\lminus$ (as required for a consensus-halving solution), and more
generally, a point in $[-1,1]^n$ is close to the centre of $B$ iff the c-e agents have approximately
balanced allocations of $\lplus$ and $\lminus$.

Observation~\ref{obs:reliable} tells us that at most $n$ circuit-encoders fail to receive reliable input.
If all circuit-encoders received reliable input, then the total error at a solution would be at most $\varepsilon$, i.e. the precision parameter of the $I_{CH}$ instance. However, since at most $n$ of them receive unreliable input, we might have an added discrepancy of at most $n/p^C$ when taking the average and therefore we need to get within distance $\varepsilon + \frac{2n}{p^C}$ of the centre of $B$. For this to be possible, we need some of the cancelling
to take place amongst the outputs of the circuit-encoders that received reliable inputs,
so we really can find a pair of correctly oppositely-coloured points.

Proposition~\ref{obs:invpoly} tell us that if $\xb\in D$ is the point in $D$ represented
by some \ch\ solution, then provided $\xb$ lies in the twisted tunnel,
the corresponding cluster of $p^C$ points must be mapped by the
circuit-encoders $C_i$ that mostly cancel each other out, so we find
pairs of points that belong to oppositely-labelled colour-regions,
from which Section~\ref{sec:backward} tells us how to recover
two oppositely-coloured cubelets of $I_{VT}$.
\end{proof}
It remains to rule out the possibility of $\xb$ occurring outside the twisted tunnel.

\subsection{No bogus approximate-zeroes of $F$ at boundary of Significant Region}\label{sec:bdry1}

Proposition \ref{prop:tt} tells us that
$\varepsilon$-approximate zeroes of $F$ (inputs for which $F$ has value in $[-\varepsilon,\varepsilon]^n$)
 within the twisted tunnel $T$ must encode solutions.
Around $T$, there are outer regions $R(S)$; note that if $i\in S$ then $-i \not\in S$
and moreover there is an inverse-polynomial lower bound on the distance between
any pair of points belonging to outer regions containing opposite colours.
But we have to rule out points with colour $j\in S$ being averaged with nearby
points that are ``coloured'' $-j$ due to a blanket-sensor agent.
In more detail, if $\xb\in R(S)$ and $\xb$ is within $\delt$
of $\xb'$ for which the $j$-th blanket-sensor agent is active and provides feedback corresponding to $-j$,
then we will prove that $S$ contains some other colour $k\not=j$ and no point in a $\delt$-neighbourhood
of $\xb$ activates the $k$-th blanket-sensor $b_{1,k}$ to provide feedback corresponding to $-k$. 

In the following, we will refer to cuts in the following manner: ``cut $i$'' refers to the $i$-th cut (from left to right) in the c-e region. Also, recall that the width $\delta^T$ of the twisted tunnel is smaller than any other
inverse-polynomial quantities of interest, apart from $\delt$, which itself is
smaller than all other inverse-polynomials of interest, including $\delta^T$. We provide the following definition
of a \emph{consistent colour}.

\begin{definition}[Consistent Colour]\label{def:consistentcolour}
For $\xb= (\tau;\alpha_2,\ldots,\alpha_n)$ in the Significant Region, colour $j\in\{\pm 2,\ldots,\pm n\}$, let $A_j\in \{\lplus,\lminus\}$ be the label that tends to increase in interval $[j-2,j]$ when the $|j|$-th coordinate $\alpha_j$ of $\xb$ is increased if $j>0$, or decreased if $j<0$.
($A_j$ depends on the sign and parity of $j$.) We say that $\xb$ has {\em consistent colour} $j$ if
\begin{enumerate}
\item if $j>0$ then $\alpha_j>2\delta^T$; if $j<0$ then $\alpha_{|j|}<-2\delta^T$;\label{con:1}
\item at least $\frac{1}{2}-\frac{\pmeg}{2\pgig}$ of the interval $[j-2,j]$ gets the label $A_j$.\label{con:2}
\end{enumerate}
\end{definition}
Condition \ref{con:1} says that at $\xb \in R(S)$, colour $j$ is a member of $S$ and the corresponding transformed coordinate is sufficiently far from the twisted tunnel. Condition \ref{con:2} says that we are at least some (small but significant) distance from triggering the $j$-th blanket-sensor in a direction that corresponds to excessive colour $-j$. 
In other words, for the $j$-th blanket-sensor to become active in direction $-j$, we would have to increase $A_{-j}$ by an inverse-polynomial amount.

The following proposition establishes that for points in the outer regions $R(S)$, consistent colours exist.

\begin{proposition}\label{prop:nbbsr}
Suppose $\xb = (\tau;\alpha_2,\ldots,\alpha_n)$ belongs to outer region $R(S)$ and that $\xb$ is within distance $\delt$ of the boundary of the Significant Region.
Then $\xb$ has a consistent colour in $\{\pm2,\ldots,\pm n\}$.
\end{proposition}

\begin{proof}
Let $\ell\in\arg\max_{i \in \{2,\ldots,n\}}|\alpha_i|$ be the index of 
a transformed coordinate with maximum absolute value. 
We may assume there is an inverse-polynomial quantity $\delta^+$ 
such that for any point $\xb=(\tau;\alpha_2,\ldots,\alpha_n)$ within 
distance $\delt$ of the boundary of the Significant Region, we have $|\alpha_\ell | \leq \delta^+$.
Moreover, the width $\delta^T$ of the twisted tunnel is chosen to be much smaller that $\delta^+$
(by an inverse-polynomial amount) but much larger than $\delt$ (by an inverse-polynomial amount), as explained earlier.
%\begin{enumerate}
%\item
%The upper bound of $\delta^+$ on the width of the significant region is applied in
%Case 2a: movements of cuts (within the significant region)
%are small relative to $\tau$ (which is assumed to be $\geq 1/2n$).
%\item
%The lower bound (of $\delta^-$) is applied in Case 2b(ii), when we obtain lower
%bounds for a sequence of transformed coordinate-values $\alpha_i$.
%\end{enumerate}
%Let $\ell\in\arg\max\{|\alpha_i|\}_i$.
Let 
\[
j = 
\begin{cases}
\ \ \ell, & \text{if } \alpha_\ell > 0 \\
-\ell, & \text{if } \alpha_\ell <0
\end{cases}
\]
and let $\delta=|\alpha_\ell|$,
so $\xb$ is displaced distance $\delta>0$ from the axis in direction $d^\tau_j$.
Recall that $A_j\in\{\lplus,\lminus\}$ denotes the label that increases in the
c-e region when we move in direction $d^\tau_j$. Also, let $A_{-j} \in \{+,-\}$, $A_{-j} \neq A_{j}$, be
the complementary label.
For $j>0$, this involves cuts $j$ and $j-1$ moving away from each other;
for $j<0$, this involves them moving towards each other. We consider two main cases, 
depending on the sign of $j$.

\bigskip
\noindent {\bf Case 1:} $j>0$ (i.e., $\alpha_j>0$). In this case, moving in direction $d^\tau_j$ causes cuts $j-1$ and $j$ to move away from each other; this is illustrated in Figure~\ref{fig:nbbsr}.\smallskip

\noindent We claim that $j$ is a consistent colour for $\xb$.
Note first that $\alpha_j > 2\delta^T$ and therefore Condition \ref{con:1} is satisfied, 
since $j >0$ in this case. $\alpha_j > 2\delta^T$ follows from the fact that $j \in \arg\max_{i \in \{2,\ldots,n\}}|\alpha_i|$, we are close to the boundary of the Significant Region, and the width of Significant Region is polynomially larger than that of the twisted tunnel. In order to be close to the boundary of the Significant Region, we must have moved more than $2\delta^T$ in some direction from $\bf{0}_\tau$ and by the choice of $j$, it holds that $\alpha_j > 2\delta^T$.

For Condition \ref{con:2}, recall first that $\xb$ has transformed 
coordinates $(\tau;\alpha_2,\ldots,\alpha_n)$,
and that the origin of $D_\tau$ has transformed coordinates ${\bf 0}_\tau=(\tau;0,\ldots,0)$.
The $(j-1)$-st and $j$-th cuts corresponding to the point ${\bf 0}_\tau$ are located at positions $j-2+\tau$ and $j-1+\tau$ respectively, and are shown in red in Figure~\ref{fig:nbbsr}.
Near the axis, where the cuts are evenly-spaced (see Proposition~\ref{prop:s-r-radius}),
movement in direction $d^\tau_j$ corresponds to moving the $(j-1)$-st and $j$-th cuts
(in the c-e region) away from each other. We will consider moving from ${\bf 0}_\tau$ to $\xb$ via 
a point $\xb_j$ in which we will only have increased the transformed coordinate $\alpha_j$. 

First, consider moving from ${\bf 0}_\tau$ to point $\xb_j=(\tau;0,\ldots,0,\alpha_j,0,\ldots,0)$ for $\alpha_j>0$. In this process, we move the $(j-1)$-st cut to the left by $\alpha_j \cdot\tau$
and the $j$-th cut to the right by $\alpha_j \cdot (1-\tau)$; all this takes place within
the interval $[j-2,j]$, see Figure~\ref{fig:nbbsr}. Now consider moving from $\xb_j$ to $\xb$. In this process, the $(j-1)$-st cut moves to the right by $\alpha_{j-1}\cdot(1-\tau)$
and the $(j+1)$-st cut moves to the left by $\alpha_{j+1}\cdot\tau$.
From the choice of $j$ to be  $j \in\arg\max_{i \in \{2,\ldots,n\}}|\alpha_i|$, it follows that
$\alpha_{j-1}\leq \alpha_j$ and $\alpha_{j+1}\leq \alpha_j$.
Then, there is a sub-interval of $[j-2,j]$ that contains the unit-length interval $I=[j-2+\tau+\alpha_j\cdot(1-2\tau),j-1+\tau+\alpha_j\cdot(1-2\tau)]$
which ends up coloured entirely $A_j$, implying Condition \ref{con:2}.
Overall, we obtain that $j$ is a consistent colour.\\

\noindent {\bf Case 2:} $j<0$ (i.e., $\alpha_j<0$). In this case, moving in direction $d^\tau_j$ causes cuts $|j|-1$ and $|j|$ to move towards each other; this is illustrated in Figure~\ref{fig:nbbsr2}.\\

\noindent {\bf Case 2a:} $\tau\in[1/2n,1-(1/2n)]$. In this case, all movements of the cuts, in and around the Significant Region, are in distances upper-bounded by $\delta^w$, which by Proposition~\ref{prop:s-r-radius} is smaller than $1/2n$ by an inverse-polynomial amount. This means that if we start at ${\bf 0}_\tau$ and re-set individual transformed coordinates to those of $\xb$, in any order (i.e. going through any intermediate point $\xb_j$, similarly to above), the movement of the cuts will never force them to cross integer-valued thresholds. In other words, in moving from ${\bf 0}_\tau$ to $\xb$, only the relevant cuts $j-1$ and $j$ will lie in the interval $[j-2,j]$. This case can be seen in the illustration of Figure~\ref{fig:nbbsr}, if one reverses the direction of the arrows, switches the labels $A_j$ to $A_{-j}$ and vice-versa, and substitutes $j$ by $|j|$ in the labelling of cuts. The argument establishing the existence of a consistent colour is exactly symmetric to that of Case 1 above.\\

\noindent {\bf Case 2b:} $\tau \in [0,1/2n] \cup [1-(1/2n),1]$. Here, we consider the case where $\tau\in[0,1/2n]$; the other case is similar by symmetry. This case is illustrated in Figure~\ref{fig:nbbsr2}; note that the sequence of labels $A_j/A_{-j}$ is switched to make $A_j$ the label
that increases when we move in direction $d^\tau_j$. 

Moving in direction $d^\tau_j$ causes an increase of the label $A_j$ in the interval $[|j|-2,|j|]$.
For $j$ not to be a consistent colour, we should observe an excess of the label $A_{-j}$ in this interval.
In generating cut locations from coordinates of $\xb$, the amount of $A_{-j}$ in $[|j|-2,|j|]$ can be raised in the following ways (see Figure~\ref{fig:nbbsr2}):
\begin{itemize}
	\item By increasing the transformed coordinate $\alpha_{|j|-1}$  in the negative direction,
	moving in direction $d^\tau_{-(|j|-1)}$. This causes cuts $|j|-2$ and $|j|-1$ to move towards each other and therefore importantly for us here, cut $|j|-1$ to move to the left (towards integer point $|j|-2$).
	\item By increasing the transformed coordinate $\alpha_{j+1}$ in the negative direction, moving in direction $d^\tau_{-(|j|+1)}$. This causes cuts $|j|$ and $|j|+1$ to move towards each other. 
\end{itemize}
Note that what may happen in this last case, is that cut $|j|+1$ which used to lie to the right of the integer point $|j|+2$ before moving in direction $d^\tau_{-(|j|+1)}$, now lies to the left of the integer point $|j|+2$ after the movement, therefore increasing the label $A_{-j}$ at the {\em right-hand-side} of $[|j|-2,|j|]$. We consider two more cases, depending on whether or not this is the case.\\

\noindent {\emph{Case 2b(i)}:} At $\xb$, cut $|j|+1$ is to the right of location $|j|$ or at location $|j|$.\\

\noindent There are two ways to restore the deficit of $A_{-j}$ that resulted from moving
in direction $d^\tau_j$ from ${\bf 0}_\tau$ to $\xb_j$. Moving in direction $d^\tau_{-(|j|-1)}$ moves cut $|j|-1$
to the left, and moving in direction $d^\tau_{-(|j|+1)}$ moves cut $|j|$ to the right.
(Note that the movement of cut $|j|+1$ to the left has not changed the balance of $A_j$ and $A_{-j}$ in the interval $[|j|-2,|j|]$ any further, by the assumption of the case).
Since $j$ was chosen to be in $\arg\max_{i \in \{2,\ldots,n\}}|\alpha_i|$, it is easy to verify that
\begin{itemize}
	\item Cut $|j|-1$ has moved to the left as a result of moving in direction $d^\tau_{-(|j|-1)}$ \emph{at most as much} as cut $|j|$ has moved to the left as a result of moving in direction $d^\tau_j$ (from ${\bf 0}_\tau$ to $\xb_j$).
	\item Cut $|j|$ has moved to the right as a result of moving in direction $d^\tau_{-(|j|+1)}$ \emph{at most as much} as cut $|j|-1$ has moved to the right as a result of moving in direction $d^\tau_j$ (from ${\bf 0}_\tau$ to $\xb_j$). 
\end{itemize}
Therefore a large enough subinterval of $[|j|-2,|j|]$ has been coloured with $A_j$, which means $j$ is a consistent colour.\\

\noindent {\emph{Case 2b(ii)}:} At $\xb$, cut $|j|+1$ is to the left of location $|j|$.\\

\noindent In this case, we have $\alpha_{|j|+1}<0$ and movement in direction
$d^\tau_{-(|j|+1)}$ causes cuts $|j|$ and $|j|+1$ to move towards each other.
Note also that besides the effect of the movement in direction $d^\tau_{-(|j|+1)}$, cut $|j|+1$ may move to the left due to movement in direction $d^\tau_{|j|+2}$, since such a movement would cause cuts $|j|+1$ and $|j|+2$ to move away from each other and therefore, cut $|j|+1$ to move to the left. However, the distance moved in direction $d^\tau_{|j|+2}$ is small; it is at most $\tau\cdot |\alpha_j|$, which is at most $\tau \cdot \delta^+$. % \ll \leq \tau.\delta^+ \ll \tau$. 
Therefore, we need movement at least $\tau(1-\delta^+)$ in direction $d^\tau_{-(|j|+1)}$ in order to cover the distance moved in direction $d^\tau_j$.

First, we verify that $-(|j|+1)$ satisfies Condition \ref{con:1} of Definition \ref{def:consistentcolour}, i.e. that $\alpha_{|j|+1} < -2\delta^T$ (at this point we know that $\alpha_{j+1}$ is a negative quantity). We consider two cases, depending on whether $\tau$ is ``small'' or ``large'' (relatively to the small interval $[0,1/2n]$).
\begin{itemize}
	\item In the case when $\tau<\frac{1}{4}|\alpha_j|$, the largest part of the deficit of $A_{-j}$ introduced by moving from ${\bf 0}_\tau$ to $\xb_j$ results from moving cut $|j|$ to the left. However, letting $c(|j|-1)$ denote the position of cut $|j|-1$ after this movement, the interval $[|j|-2,c(|j|-1)]$ is too small for the movement of cut $|j|-1$ in direction $d^\tau_{-(|j|-1)}$ to compensate. In other words, even if movement in direction $d^\tau_{-(|j|-1)}$ moves cut $|j|-1$ to the left endpoint of the interval $[|j|-2,|j|]$, this is not enough to make up for the deficit of $A_{-j}$ introduced from the movement in direction $d^\tau_j$. This means that cut $|j|+1$ needs to move to the left as well and in particular, it needs to move by more than $\tau/4$ to the left of location $j$. This is only possible if $\alpha_{|j|+1} < -2\delta^T$.

\item In the case when $\tau\geq \frac{1}{4}|\alpha_j|$, since $\tau$ is large enough, cut $|j|+1$ needs to move a substantial distance to the left, in order to end up positioned to the left of integer position $|j|$. In particular, it needs to move at least $\frac{1}{4}|\alpha_j|-\tau\cdot \delta^+$ to the left. This implies that Condition \ref{con:1} is satisfied for colour $-(|j|+1)$.
\end{itemize}
Now consider what needs to happen in order for the second condition to fail.
Consider the interval $[|j|-1,|j|+1]$ (which is monitored by the $(j+1)$-st blanket-sensor). Since cut $|j|+1$ is located to the left of location $|j|$
(the midpoint of this interval), there exists a subinterval of length at most $1$ labelled $A_j$,
within $[|j|-1,|j|+1]$. This means that either
\begin{itemize}
	\item[-] the colour $-(|j|+1)$ is a consistent colour and we are done, or 
	\item[-] there is an additional amount of label $A_j$ within interval $[|j|-1,|j|+1]$ and the total number of value-blocks labelled $A_j$ outnumbers that of those labelled $A_{-j}$ by at least $\pmeg$. The only way this can happen is if cut $|j|+2$
	lies to the left of the integer location $|j|+1$, and in fact, it has to lie an inverse-polynomial distance,
	at least $\frac{\pmeg}{2\pgig}$, to the left of $|j|+1$.
\end{itemize} 
In case that happens, we move on to consider interval $[j,j+2]$ and we apply the same argument.
Again, $\alpha_{j+2}$ is negative, and since cut $|j|+2$ is to the left of location $j+1$
by a margin $\frac{\pmeg}{2\pgig} < 2\delta^T$, $-(|j|+2)$ satisfies Condition \ref{con:1} to be a consistent colour. It will also satisfy Condition \ref{con:2}, unless
cut $|j|+3$ lies to the left of location $|j|+2$ by an inverse-polynomial distance, at least $\frac{\pmeg}{2\pgig}$, similarly to before.

Continuing like this, we will either find a consistent colour in some interval $[j-2,j]$ with $j<n$, or we will reach interval $[n-2,n]$. When we reach interval $[n-2,n]$, cut $n$ has had to move to the left
of integer location $n-1$ in order to prevent $-(n-1)$ from being a consistent colour (as otherwise we
would have identified a consistent colour in some already examined interval).
But then $-n$ is a consistent colour, since we have moved an inverse-polynomial
distance (at least  $\frac{\pmeg}{2\pgig} <2\delta^T$) in direction $d^\tau_{-n}$ (Condition \ref{con:1}),
and at least $1/2$ of the interval $[n-2,n]$ is coloured in a way that agrees with this (Condition \ref{con:2}).
\end{proof}

\begin{figure}
\center{
\begin{tikzpicture}[scale=0.8]
\tikzstyle{xxx}=[dashed,thick]

\draw[thick,<->](0,8)--(8,8);\node[fill=white]at(4,8){$[j-2,j]$};
\draw[ultra thick](-1,5)--(10,5);
\draw[dotted,thick](0,5)--(0,8);
\draw[dotted,thick](8,5)--(8,8);

\node at(-1,7){$A_{-j}$};\node at(3,7){$A_j$};\node at(7,7){$A_{-j}$};
\draw[red,thick](1,3)--(1,7.5);\draw[red,thick](5,3)--(5,7.5);
\node[rotate=90,color=red]at(0.7,6.5){cut ${j-1}$};\node[rotate=90,color=red]at(4.7,6.5){cut $j$};

\draw[blue,thick](2,3)--(2,7.5);\draw[blue,thick](6,3)--(6,7.5);
\node[rotate=90,color=blue]at(1.75,6.5){cut ${j-1}$};\node[rotate=90,color=blue]at(5.7,6.5){cut $j$};

\draw[<-,thick](0.5,4.3)--(1,4.3);\draw[->,thick](5,4.3)--(6.5,4.3);
\draw[dotted](1,4.3)--(5,4.3);\node[fill=white]at(3,4.3){$d^\tau_j$};

\draw[dashed,thick](0.5,5)--(0.5,2);\draw[dashed,thick](6.5,5)--(6.5,2);

\draw[->,thick](0.5,2)--(2,2);\draw[<-,thick](6,2)--(6.5,2);
\node at(1.5,1.5){$d^\tau_{j-1}$};\node at(6.5,1.5){$d^\tau_{j+1}$};

\draw[dotted,thick](2,3)--(2,0);\draw[dotted,thick](6,3)--(6,0);

\draw[<->,thick](2,1)--(6,1);\node[fill=white]at(4,1){$I$};

\end{tikzpicture}
\caption{\small{Illustration for Proposition~\ref{prop:nbbsr}, Case 1.
For $i\in[n]$, cut $i$ denotes the $i$-th (c-e) cut from the left. The cuts $j-1$ and $j$ coloured red
correspond to positions encoding part of ${\bf 0}_\tau$. The dashed lines (to the left and to the right of the
positions of the red cuts respectively) correspond to positions encoding part of point $\xb_j$, after we have only moved $\alpha_j$ in direction $d^\tau_j$. The cuts coloured blue correspond to positions resulting from subsequent movement in directions $d^\tau_{j-1}$ and $d^{\tau}_{j+1}$, which encode part of $\xb$.
In the figure, the case where the average of the movements in $d^\tau_{j-1}$, $d^\tau_j$ and $d^{\tau}_{j+1}$ forces both cuts to move to the right, compared to their original positions in the encoding of ${\bf 0}_\tau$, is shown.}\label{fig:nbbsr}}}
\end{figure}
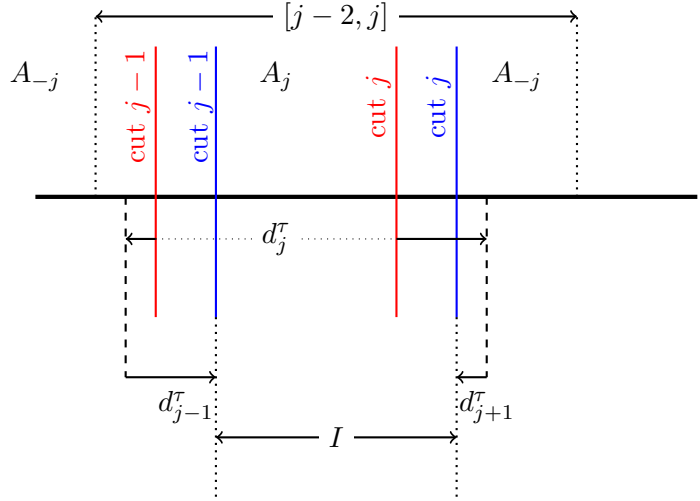

\begin{figure}
\center{
\begin{tikzpicture}[scale=0.8]
\tikzstyle{xxx}=[dashed,thick]

\draw[thick,<->](0,8)--(8,8);\node[fill=white]at(4,8){$[|j|-2,|j|]$};
\draw[ultra thick](-1,5)--(10,5);
\draw[dotted,thick](0,5)--(0,8);
\draw[dotted,thick](8,5)--(8,8);

\node at(-1.5,7){$A_j$};\node at(2,7){$A_{-j}$};\node at(7,7){$A_j$};
\draw[red,thick](0.25,3)--(0.25,7.5);\draw[red,thick](4.25,3)--(4.25,7.5);\draw[red,thick](9.75,3)--(9.75,7.5);\node[rotate=90,color=red]at(0.6,6.5){cut ${|j|-1}$};\node[rotate=90,color=red]at(3.8,6){cut $|j|$};
\node[rotate=90,color=red]at(9.5,6.5){cut ${|j|+1}$};

\draw[red,blue](-0.5,3)--(-0.5,7.5);\draw[blue,thick](3.25,3)--(3.25,7.5);\draw[blue,thick](8.75,3)--(8.75,7.5);\node[rotate=90,color=blue]at(-1,6.5){cut ${|j|-1}$};\node[rotate=90,color=blue]at(2.8,6){cut $|j|$};
\node[rotate=90,color=blue]at(8.5,6.5){cut ${|j|+1}$};

\draw[->,thick](0.25,4)--(0.5,4);\draw[->,thick](4.25,4)--(3.25,4);
\draw[dotted](0.5,4)--(3.25,4);\node[fill=white]at(2,4){$d^\tau_{-|j|}$};

\draw[dotted,thick](0.5,5)--(0.5,2);\draw[dotted,thick](3.25,5)--(3.25,2);

\draw[->,thick](0.5,3)--(-0.5,3);\draw[->,thick](2.75,3)--(3.25,3);\draw[<-,thick](8.75,3)--(9.75,3);
\node at(-0.5,2.5){$d^\tau_{-(|j|-1)}$}; \node at(3,2.5){$d^\tau_{-(|j|+1)}$};\node at(9,2.5){$d^\tau_{-(|j|+1)}$};

%\draw[dotted,thick](3,3)--(3,0);\draw[dotted,thick](7,3)--(7,0);
%\draw[<->,thick](3,1)--(7,1);\node[fill=white]at(5,1){$I$};

\end{tikzpicture}
\caption{\small{Illustration for Proposition~\ref{prop:nbbsr}, Case 2b(i). In this case, $j$ is negative and therefore 
we use $|j|$ to represent the $j$-th cut from the left. Moving in direction $d^\tau_j$ causes cuts $|j|-1$ and $|j|$ to move towards each other.
Again, the red cuts correspond to part of ${\bf 0}_\tau$ and the blue cuts correspond to the encoding of part of $\xb$, after averaging over the movement in directions $d^\tau_{-(|j|-1)}$, $d_j = d_{-|j|}$ and $d^\tau_{-(|j|+1)}$. The figure shows a case where cut $|j|-1$ has moved outside the interval $[|j|-2,|j|]$ to the left, in which case the whole subinterval $[|j|-2,c(|j|)]$ (the interval between $|j|-2$ and the position of the blue cut $|j|$) receives the label $A_{-j}$. Note however that $|j|+1$ does not intersect the interval $[|j|-2,|j|]$ and therefore there is no additional amount of $A_{-j}$ introduced to the right-hand side of $[|j|-2,|j|]$, therefore we are in Case 2b(i). The increase of $A_{-j}$ due to the movement of cut $|j|-1$ to the left is entirely compensated by the decrease of $A_{-j}$ because of the movement of cut $|j|$ to the left.}\label{fig:nbbsr2}}}
\end{figure}
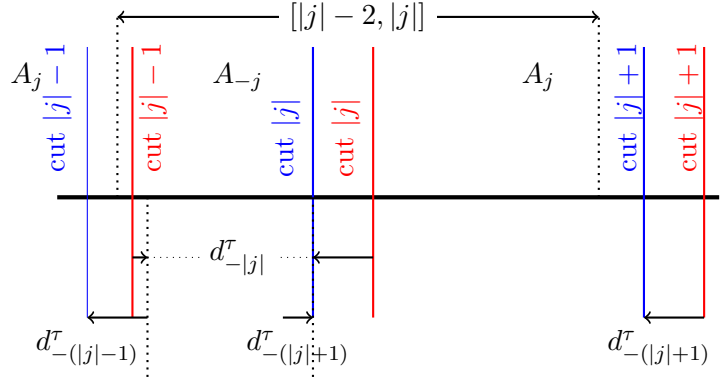

\begin{figure}
	\center{
		\begin{tikzpicture}[scale=0.8]
		\tikzstyle{xxx}=[dashed,thick]
		
		\draw[thick,<->](0,8)--(8,8);\node[fill=white]at(4,8){$[|j|-2,|j|]$};
		\draw[ultra thick](-1,5)--(10,5);
		\draw[dotted,thick](0,5)--(0,8);
		\draw[dotted,thick](8,5)--(8,8);
		
		\node at(-1.5,7){$A_j$};\node at(2,7){$A_{-j}$};\node at(7,7){$A_j$}; \node at (10,7){$A_{-j}$};
		\draw[red,thick](0.25,3)--(0.25,7.5);\draw[red,thick](4.25,3)--(4.25,7.5);\draw[red,thick](8.75,3)--(8.75,7.5);\node[rotate=90,color=red]at(0.6,6.5){cut ${|j|-1}$};\node[rotate=90,color=red]at(3.8,6){cut $|j|$};
		\node[rotate=90,color=red]at(8.5,6.5){cut ${|j|+1}$};
		
		\draw[red,blue](-0.5,3)--(-0.5,7.5);\draw[blue,thick](3.25,3)--(3.25,7.5);\draw[blue,thick](7.75,3)--(7.75,7.5);\node[rotate=90,color=blue]at(-1,6.5){cut ${|j|-1}$};\node[rotate=90,color=blue]at(2.8,6){cut $|j|$};
		\node[rotate=90,color=blue]at(7.5,6.5){cut ${|j|+1}$};
		
		\draw[->,thick](0.25,4)--(0.5,4);\draw[->,thick](4.25,4)--(3.25,4);
		\draw[dotted](0.5,4)--(3.25,4);\node[fill=white]at(2,4){$d^\tau_{-|j|}$};
		
		\draw[dotted,thick](0.5,5)--(0.5,2);\draw[dotted,thick](3.25,5)--(3.25,2);
		
		\draw[->,thick](0.5,3)--(-0.5,3);\draw[->,thick](2.75,3)--(3.25,3);\draw[<-,thick](7.75,3)--(8.75,3);
		\node at(-0.5,2.5){$d^\tau_{-(|j|-1)}$}; \node at(3,2.5){$d^\tau_{-(|j|+1)}$};\node at(9,2.5){$d^\tau_{-(|j|+1)}$};
		
		%\draw[dotted,thick](3,3)--(3,0);\draw[dotted,thick](7,3)--(7,0);
		%\draw[<->,thick](3,1)--(7,1);\node[fill=white]at(5,1){$I$};
		
		\end{tikzpicture}
		\caption{\small{Illustration for Proposition~\ref{prop:nbbsr}, Case 2b(ii). In this case, $j$ is negative and therefore we use $|j|$ to represent the $j$-th cut from the left. Moving in direction $d^\tau_j$ causes cuts $|j|-1$ and $|j|$ to move towards each other.
		Again, the red cuts correspond to part of ${\bf 0}_\tau$ and the blue cuts correspond to the encoding of part of $\xb$, after averaging over the movement in directions $d^\tau_{-(|j|-1)}$, $d_j = d_{-|j|}$ and $d^\tau_{-(|j|+1)}$. The figure shows a case where cut $|j|-1$ has moved outside the interval $[|j|-2,|j|]$ to the left, in which case the whole subinterval $[|j|-2,c(|j|)]$ (the interval between $|j|-2$ and the position of the blue cut $|j|$) receives the label $A_{-j}$. Additionally, cut $|j|+1$ has moved to the left and now intersects the interval $[|j|-2,|j|]$ introducing an  additional amount of $A_{-j}$ to the right-hand side of $[|j|-2,|j|]$. By the argument of Case 2b(ii), either $-(|j|+1)$ will be a consistent colour, or there will be some interval $[l-2,l]$ ($l >0$, possibly $[n-2,n]$) for which the overlap between $[l-2,l]$ and the cut $l+1$ will be bounded by $\pmeg/2\pgig$ and we will have a consistent colour.}\label{fig:nbbsr3}}}
\end{figure}

\begin{corollary}\label{cor:nbbsr}
A solution $S_{CH}$ to $I_{CH}$ cannot encode a point
$\xb$ in the Significant Region, within distance $\delt$ of the boundary
(where blanket-sensor agent(s) become active).
\end{corollary}

\begin{proof}
Observation~\ref{obs:61} tells us that the $k$-th component of $f'$ is the difference between $\lplus$ and $\lminus$ observed by c-e agent $a_k$, and Proposition~\ref{prop:62} tells us that all these components, averaged over a set of points within $\delt$ of $\xb$ need to be close to zero, at a solution.

Proposition \ref{prop:nbbsr} tells us that $\xb$ has some consistent colour $k$.
All points within $\delt$ of $\xb$ cause two outputs 
(gates $g'_k$, $g'_{-k}$ as defined in Section~\ref{sec:cea}) of the circuit-encoders to
represent colour $k$. This includes points where blanket-sensor agents are
active, since by the properties of consistent colours, we are at least an inverse-polynomial
distance from any point where any $b_{i,k}$ can be active in the wrong direction.
In Section~\ref{sec:bsa} the blanket-sensor agents are designed to agree with the definition of consistent colour,
Definition \ref{def:consistentcolour}.
So c-e agent $a_k$ observes a large imbalance between $\lplus$ and $\lminus$.
\end{proof}

\subsection{No bogus approximate-zeroes of $F$ due to the connecting facet}

\begin{proposition}\label{prop:strips}
Let $\xb$, $\xb'$ be points in the Significant Region having transformed coordinates
$(\tau;\alpha_2,\ldots,\alpha_n)$ and $(1-\tau;-\alpha_2,\ldots,-\alpha_n)$ respectively,
for $\tau<\frac{1}{2}-\delta^T$. Then $f(\xb)=-f(\xb')$.
\end{proposition}

\begin{proof}
The proposition extends Observation~\ref{obs:flip}. The points
$\xb$ and $\xb'$ have been coloured according to Item~\ref{item-f-strips} of Section~\ref{sec:colour-f},
and they belong to two long thin colour-regions that extend the cubelets that lie
on the panchromatic facets of the cube embedded at the centre of $T$, all the
way to the ends of $T$.
From the boundary conditions on the colouring of box $B$ in \nvhdt,
and the way $f$ is constructed above, their colours are equal and opposite.
\end{proof}
{\bf Remark:} The $\xb$, $\xb'$ in Proposition~\ref{prop:strips} will ``approach each other''
as $\tau\rightarrow 0$. That is, they correspond to sequences of \ch\ cuts where
the left-hand cut in the c-e region ``wraps around'' to the right-hand side of the c-e region.
Proposition~\ref{prop:strips} may thus seem to create Borsuk-Ulam directions that are
in conflict with each other as we cross from facet $D_0$ to $D_1$, but in fact the flip
of labels in \ch\ that occurs when we move from $D_0$ to $D_1$
will mean that they are in agreement with each other.

\bigskip
\noindent We consider the case where the set of $p^C$ points in $D$ represented
by the solution $S_{CH}$ to $I_{CH}$, contains points on opposite sides of the facets
of $D$ that have been identified with each other.
Proposition~\ref{prop:strips} tells us that
colour-regions are adjacent to colour-regions having the opposite colour.
We need to verify that for a pair $\xb,\xb'$ of points
that are close together but have opposite colours (due to lying in such
a pair of colour-regions) the same (and not opposite) feedback is provided to the c-e agents.
(So, in contrast with a pair of opposite-colour points that represent
a solution, whose feedback to the c-e agents cancel each other out.)

In reasoning about these elements $\xb,\xb'\in D$, it is helpful to depart from our convention
that the label-sequence begins with $\lplus$, and suppose that for
$\xb'$, the label-sequence begins with $\lminus$.
Suppose $\xb,\xb'$ have corresponding circuit-encoders $C_i,C_{i'}$ and assume
that $C_i$ and $C_{i'}$ receive reliable inputs, recalling that only $n$ circuit-encoders
may fail to receive reliable inputs.
Notice that if $\xb$ causes a blanket-sensor agent $b_{i,j}$ to be active
in direction $\lplus$, then $\xb'$ typically causes $b_{i',j}$ to be active
in direction $\lplus$ also (the over-represented label is fed back to
c-e agent $a_j$).

In the case that no blanket-sensor agents are active, if $\xb,\xb'$ receive
opposite colours from $C_i,C_{i'}$, then, reverting to our convention
that the shared label-sequence begins with $\lplus$, we note that
their reference-sensor agents get opposite labels, which causes $C_i$
and $C_{i'}$ to agree with each other.\smallskip

\noindent \textbf{Remark.}
For intuition, it is possibly helpful to think about the move from $\xb$ to
$\xb'$ in terms of operations on the coordinate-encoding cuts.
At $\xb$, there is a cut on the right-hand side of the c-e region, and in moving to $\xb'$
we move that cut to the left-hand side. If we move the cut while leaving the labels
of the c-e region unchanged (apart from at the ends) we expect the circuit to
behave as before, but since we have switched the roles of labels $\lplus$ and $\lminus$,
the feedback to agents $a_1,\ldots,a_n$ gets inverted.
We re-invert this feedback by reversing the colour, and hence the output of $f'$.
This is very similar (in fact, an $n$-dimensional analogue) to the handling of the
``wrap-around points'' in the sequence via the interpretation in terms of the virtual cuts in \cite{FG17}.

\section{Further work}

What is the computational complexity of $k$-thief \ns, for $k$ not a power of 2?
As discussed in \cite{Meunier14,LZ06}, the proof that it is a total search problem,
does {\em not} seem to boil down to the \ppa\ principle.
Right now, we do not even even know if it belongs to PTFNP~\cite{GP18}.
%The alert reader may notice however, that if so, a version of the 2-thief problem
%where the thieves seek a two-thirds/one-third split, should be
%``harder'' than the standard version. 

Interestingly, Papadimitriou in \cite{Pap} (implicitly) also defined a number of computational
complexity classes related to \ppa, namely \ppa-$p$, for a parameter $p \geq 2$. \ppa-$p$ is defined 
with respect to an input bipartite graph and a given vertex with degree which is not a multiple of $p$,
and the goal is to find another vertex with degree which is not a multiple of $p$ 
(it follows that \ppa=\ppa-$2$). This was done in the context of classifying the 
computational problem related to Ch{\'e}valley's Theorem from number theory, 
and it was proven that for prime $p$, \textsc{Chevally} mod $p$ is in \ppa-$p$ \cite{Pap}. 
Given the discussion above, it could possibly be the case that the the principle
associated with \ns\ for $k$-thieves is the \ppa-$k$ principle instead.

What about the computational hardness of the problem? Is 3-thief \ns\ hard for \ppa? 
At first glance, it seems like a more complicated problem, but there this is not obvious; for
example, there is no way to cause the third thief to be a dummy agent and therefore a straightforward
reduction is unlikely.
However, it is worth mentioning here that the computational equivalence between
$\varepsilon$-\ch\ and \ns\ that was proven in \cite{FG17} is actually
established between the Necklace Splitting problem for any $k$ and the
corresponding approximate $1/k$-Division problem, a generalization of
$\varepsilon$-\ch\ (see \cite{SS03}); a \ppa-hardness or \ppa-membership
result for $k>2$ for the latter problem would imply a corresponding result for \ns\ with $k>2$.
En route to either of these results, a possible approach that was suggested in \cite{SS03},
would be to define an appropriate generalisation of Tucker's Lemma and prove it \emph{constructively}
(see Section 8 in \cite{SS03}).

We have left open the questions of whether $\varepsilon$-\ch\ remains \ppa-complete
for constant $\varepsilon$, and whether \dhs\ remains \ppa-complete when coordinates
of points are given in unary. Recall that for the former problem, a \ppad-hardness result
is known from \cite{FFGZ18}; it would be quite interesting to settle this, to verify whether
it is possible for the precision parameter to play such an important role in the problem classification.

In classifying a problem as polynomial-time solvable versus \np-complete,
this is usually seen as a statement about its computational (in)tractability.
The distinction between \ppad-completeness and \ppa-completeness is one
of expressive power: we believe that \ppad-complete problems are hard,
meanwhile \ppa-complete problems are ``at least as hard'', but of
course are still in \np. The expressive power of totality principles that underpin \tfnp\ problems
is a topic of enduring interest \cite{BCEIP,GP18}; note also the related work on
Bounded Arithmetic discussed in~\cite{GP18}. Our results highlight the distinction
between computational (in)tractability and expressive power.
In analysing the relationships between these complexity classes,
it may be fruitful to focus on expressive power.

Finally, \cite{MN12} initiates an interesting experimental study of path-following algorithms
for 2-thief \ns, obtaining positive results when the number of bead colours is not too large.
However, path-following seems to be inapplicable for, say, 3 thieves.
The \ns\ problem may constitute an interesting class of challenge-instances
for SAT-solvers, now that it is known to be a very hard total search problem.

\begin{paragraph}{Acknowledgements}
We thank Alex Hollender for detailed and insightful proof-reading of earlier versions of this paper.
\end{paragraph}

\appendix

\newpage
\section*{APPENDIX}
\section{Details from \cite{FG17}.} \label{sec:app_details}

In this section, we include several details from \cite{FG17} that are being used (or extended) in the present paper as well. 

\subsection*{Bit detection gadgets}

First, the ability to detect the position of the cuts in the c-e region and feed this information to the circuit lies in the presence of gadgets developed in \cite{FFGZ18} and used in \cite{FG17}, referred to as ``bit detection gadgets'' in \cite{FG17}. A bit detection gadget consists typically of two \emph{thin} and \emph{dense} valuation blocks of relatively large height and relatively small length, situated next to each other (e.g., see the rightmost set of value-blocks in Figure \ref{fig:sensor_example} or Figure \ref{fig:gate-example}, top). These value-blocks constitute most of the agent's valuation over the related interval. The point of these gadgets is that if the discrepancy between \lplus\ and \lminus\ is (significantly) in the favour of one against the other, there will be a cut intersecting one of the two valuation blocks; which block is intersected will correspond to a $0/1$ value, i.e. a bit that indicates the ``direction'' of the discrepancy in the two labels.\\

\noindent These gadgets are used in several parts of the reduction, e.g., 
\begin{itemize}
	\item[-] at the value-blocks of the sensor agents (Definition \ref{def:sensors}) in the region $R_i$ (see Figure \ref{fig:sensor_example}) that assumes the right or left position depending on whether their small value-blocks in the c-e region are labelled \lplus or \lminus,
	\item[-] in the encoding of the circuit $C_{VT}$ of \nvhdt\, using the circuit-encoding agents $C_i$ (also see Section \ref{sec:forward}),
	\item[-] at the value-blocks of the blanket-sensor agents (Definition \ref{def:blanket}) in the region $R_i$, with the difference that in that case, a small value-block (of value $9\kappa/10$) lies between the two thin and dense valuation blocks of the bit detection gadget (see Figure \ref{fig:blanket_example}).
\end{itemize}

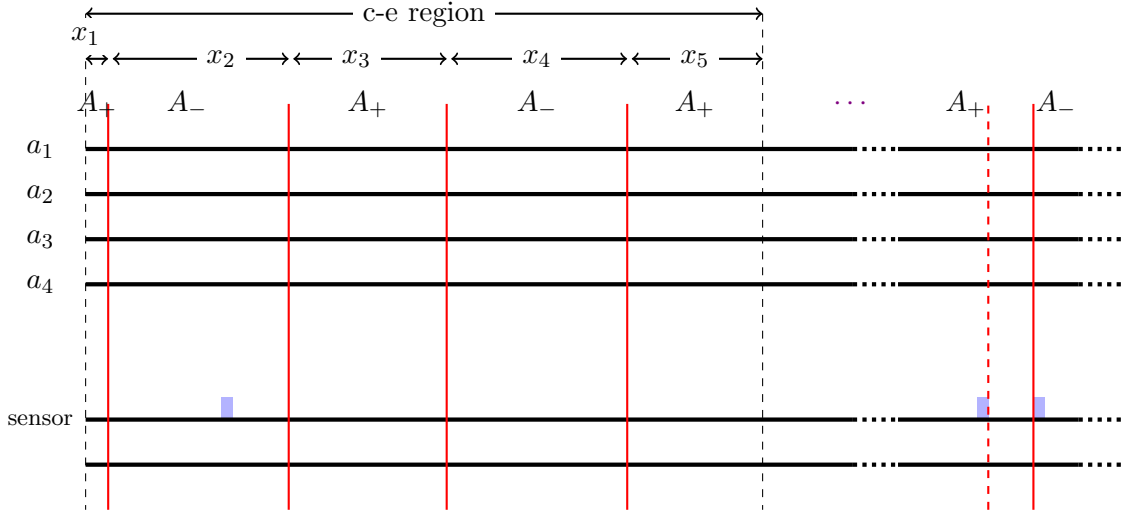
\begin{figure}
	\center{
		\begin{tikzpicture}[scale=0.6]
		\tikzstyle{line-style}=[black,ultra thick]
		\tikzstyle{cut-style}=[red,thick]
		\tikzstyle{valuestyle}=[fill=blue,opacity=0.3]
		
		\draw[dashed](0,23)--(0,12); \draw[dashed](15,23)--(15,12);
		\draw[thick,<->](0,23)--(15,23);\draw[thick,<->](0,22)--(0.5,22);
		\draw[thick,<->](0.6,22)--(4.5,22);\draw[thick,<->](4.6,22)--(8,22);
		\draw[thick,<->](8.1,22)--(12,22);\draw[thick,<->](12.1,22)--(15,22);
		
		\node[fill=white]at(7.5,23){c-e region};\node[fill=white]at(0,22.5){$x_1$};
		\node[fill=white]at(3,22){$x_2$};\node[fill=white]at(6,22){$x_3$};
		\node[fill=white]at(10,22){$x_4$};\node[fill=white]at(13.5,22){$x_5$};
		
		\fill[valuestyle,yshift=0cm](3,14)--(3,14.5)--(3.25,14.5)--(3.25,14)--cycle;

		%tri-block gadgets, including in bottom half:
		\foreach \x/\y in {20/14}{
			\fill[valuestyle,xshift=\x cm,yshift=\y cm](0,0)--(0,0.5)--(-0.25,0.5)--(-0.25,0)--cycle;
			\fill[valuestyle,xshift=\x cm,yshift=\y cm](1,0)--(1,0.5)--(1.25,0.5)--(1.25,0)--cycle;
			%\fill[valuestyle,xshift=\x cm,yshift=\y cm](0.2,0)--(0.2,0.2)--(0.8,0.2)--(0.8,0)--cycle;
		}
		
		%\fill[valuestyle,yshift=0cm](6,13)--(6,13.5)--(6.25,13.5)--(6.25,13)--cycle;
		
		\foreach \y in {20,19,18,17,14,13}{
			\draw[line-style](0,\y)--(17,\y);
			\draw[ultra thick,dotted](17,\y)--(18,\y);
			\draw[line-style](18,\y)--(22,\y);
			\draw[ultra thick,dotted](22,\y)--(23,\y);
		}
		
		\node[yshift=0cm]at(-1,20){$a_1$};\node[yshift=0cm]at(-1,19){$a_2$};
		\node[yshift=0cm]at(-1,18){$a_3$};\node[yshift=0cm]at(-1,17){$a_4$};
		\node[yshift=0cm]at(-1,14){$\footnotesize{\text{sensor}}$}; %\node[yshift=0cm]at(-1,13){$a'$};

		\foreach \x in {0.5,4.5,8,12,21}{\draw[cut-style,yshift=0cm](\x,12)--(\x,21);}
		\draw[cut-style,dashed,yshift=0cm](20,12)--(20,21);
		\node[color=black,yshift=0cm]at(0.25,21){$\lplus$};\node[color=black,yshift=0cm]at(2.25,21){$\lminus$};\node[color=black,yshift=0cm]at(6.25,21){$\lplus$};
		\node[color=black,yshift=0cm]at(10,21){$\lminus$};
		\node[color=black,yshift=0cm]at(13.5,21){$\lplus$};
		\node[color=violet,yshift=0cm]at(17,21){$\cdots$};
		\node[color=black,yshift=0cm]at(19.5,21){$\lplus$};
		
		%\node[color=violet,yshift=0cm]at(20,21.5){$v$};
		%\node[color=violet,yshift=0cm]at(21,21.5){$\lnot v$};

		\node[color=black,yshift=0cm]at(21.5,21){$\lminus$};

		%\fill[valuestyle,xshift=-4cm,yshift=-11cm](18.25,20)--(18.25,20.5)--(18.75,20.5)--(18.75,20)--cycle;
		
		%\foreach \y in {9,8,7,6,3}{
		%	\draw[line-style,dotted](0,\y)--(1,\y);
		%}
		%\foreach \x in {1,5,9,13,17}{
		%	\foreach \y in {9,8,7,6,3,2,1,0}{
			%	\draw[line-style](\x,\y)--(3+\x,\y);\draw[ultra thick,dotted](3+\x,\y)--(4+\x,\y);
		%	}}
		%	\foreach \x in {1,5,9,13,17}{
		%		\node[color=black]at(\x+0.5,10){$\lplus$};\node[color=black]at(\x+2.5,10){$\lminus$};
		%	}
			\end{tikzpicture}}
			\caption{\small{An example of how the input of a sensor agent is processed into a boolean value that will be used by the encoding of the circuit. On of the two value-blocks on the right-hand side of the picture (the bit-detection gadget) is interesected by the cut corresponding to the sensor, depending on whether the value-block on the left-hand side is labelled \lplus or \lminus. In the figure, the block is labelled \lplus and therefore the cut intersects the rightmost value-block on the right-hand side. The other (unused) option is depicted by a red dashed ine.}}\label{fig:sensor_example}
	\end{figure}
	
	\begin{figure}
		\center{
			\begin{tikzpicture}[scale=0.6]
			\tikzstyle{line-style}=[black,ultra thick]
			\tikzstyle{cut-style}=[red,thick]
			\tikzstyle{valuestyle}=[fill=blue,opacity=0.3]
			
			\draw[dashed](0,23)--(0,12); \draw[dashed](15,23)--(15,12);
			\draw[dashed](4,23)--(4,12);	\draw[dashed](8,23)--(8,12);
			%\draw[dashed](12,23)--(12,12);
			\draw[thick,<->](0,23)--(15,23);\draw[thick,<->](0,22)--(0.5,22);
			\draw[thick,<->](0.6,22)--(4.5,22);\draw[thick,<->](4.6,22)--(8,22);
			\draw[thick,<->](8.1,22)--(12,22);\draw[thick,<->](12.1,22)--(15,22);
			
			\node[fill=white]at(7.5,23){c-e region};\node[fill=white]at(0,22.5){$x_1$};
			\node[fill=white]at(3,22){$x_2$};\node[fill=white]at(6,22){$x_3$};
			\node[fill=white]at(10,22){$x_4$};\node[fill=white]at(13.5,22){$x_5$};
			
			%\fill[valuestyle,yshift=0cm](0,14)--(0,14.5)--(0.25,14.5)--(0.25,14)--cycle;
			%\fill[valuestyle,yshift=0cm](0.94,14)--(0.94,14.5)--(1.19,14.5)--(1.19,14)--cycle;
			%\fill[valuestyle,yshift=0cm](1.98,14)--(1.98,14.5)--(2.23,14.5)--(2.23,14)--cycle;
			%\fill[valuestyle,yshift=0cm](1.98,14)--(1.98,14.5)--(2.23,14.5)--(2.23,14)--cycle;
			
			\foreach \y in {1}{
				\foreach \x in {-1,0,1,2,3,4,5,6}{
					\fill[valuestyle,xshift=\x cm,yshift=-1*\y cm]
					(\y,15)--(\y+0.2,15)--(\y+0.2,15.5)--(\y,15.5)--cycle;}}
			%\fill[valuestyle](0,11)--(0.2,11)--(0.2,11.5)--(0,11.5)--cycle;
			%\fill[valuestyle](1,10)--(1.2,10)--(1.2,10.5)--(1,10.5)--cycle;
			%\fill[valuestyle](15,8)--(15.2,8)--(15.2,8.5)--(15,8.5)--cycle;
			
			%tri-block gadgets, including in bottom half:
			\foreach \x/\y in {20/14}{
				\fill[valuestyle,xshift=\x cm,yshift=\y cm](0,0)--(0,0.5)--(-0.25,0.5)--(-0.25,0)--cycle;
				\fill[valuestyle,xshift=\x cm,yshift=\y cm](1,0)--(1,0.5)--(1.25,0.5)--(1.25,0)--cycle;
				\fill[valuestyle,xshift=\x cm,yshift=\y cm](0.2,0)--(0.2,0.2)--(0.8,0.2)--(0.8,0)--cycle;
			}
			
			%\fill[valuestyle,yshift=0cm](6,13)--(6,13.5)--(6.25,13.5)--(6.25,13)--cycle;
			
			\foreach \y in {20,19,18,17,14,13}{
				\draw[line-style](0,\y)--(17,\y);
				\draw[ultra thick,dotted](17,\y)--(18,\y);
				\draw[line-style](18,\y)--(22,\y);
				\draw[ultra thick,dotted](22,\y)--(23,\y);
			}
			
			\node[yshift=0cm]at(-1,20){$a_1$};\node[yshift=0cm]at(-1,19){$a_2$};
			\node[yshift=0cm]at(-1,18){$a_3$};\node[yshift=0cm]at(-1,17){$a_4$};
			\node[yshift=0cm]at(-1,14.5){$\footnotesize{\text{blanket}}$}; %\node[yshift=0cm]at(-1,13){$a'$};
			\node[yshift=0cm]at(-1,14){$\footnotesize{\text{sensor}}$};

			\foreach \x in {0.58,12}{\draw[cut-style,yshift=0cm](\x,12)--(\x,21);}
			\draw[cut-style,yshift=0cm](4.5,12)--(4.5,21);
			\draw[cut-style,dashed,yshift=0cm](20,12)--(20,21);
			\draw[cut-style,dashed,yshift=0cm](21,12)--(21,21);
			\draw[color=blue,yshift=0cm](20.5,12)--(20.5,21);
			\node[color=black,yshift=0cm]at(0.25,21){$\lplus$};\node[color=black,yshift=0cm]at(2.25,21){$\lminus$};\node[color=black,yshift=0cm]at(6.25,21){$\lplus$};
			\node[color=black,yshift=0cm]at(10,21){$\lminus$};
			\node[color=black,yshift=0cm]at(13.5,21){$\lplus$};
			\node[color=violet,yshift=0cm]at(17,21){$\cdots$};
			\node[color=black,yshift=0cm]at(19.5,21){$\lplus$};
			
			%\node[color=violet,yshift=0cm]at(20,21.5){$v$};
			%\node[color=violet,yshift=0cm]at(21,21.5){$\lnot v$};

			\node[color=black,yshift=0cm]at(21.5,21){$\lminus$};

			%\fill[valuestyle,xshift=-4cm,yshift=-11cm](18.25,20)--(18.25,20.5)--(18.75,20.5)--(18.75,20)--cycle;
			
			%\foreach \y in {9,8,7,6,3}{
			%	\draw[line-style,dotted](0,\y)--(1,\y);
			%}
			%\foreach \x in {1,5,9,13,17}{
			%	\foreach \y in {9,8,7,6,3,2,1,0}{
			%	\draw[line-style](\x,\y)--(3+\x,\y);\draw[ultra thick,dotted](3+\x,\y)--(4+\x,\y);
			%	}}
			%	\foreach \x in {1,5,9,13,17}{
			%		\node[color=black]at(\x+0.5,10){$\lplus$};\node[color=black]at(\x+2.5,10){$\lminus$};
			%	}
			\end{tikzpicture}}
		\caption{\small{An example on how the blanket-sensor agents provide input to the circuit for their monitored intervals. Depending on the balance in labels for the value-blocks on the left-hand side, the bit-detection gadget of the blanket-sensor agent on the right-hand side assumes the leftmost, middle, or rightmost positions. In the particular example, the number of value-blocks on the left-hand side for each label is balanced, and therefore the cut on the right-hand side (shown in blue) interesects the middle value-block of the bit-detection gadget. In this case, the blanket-sensor is not active. The other two possible positions for the cut on the right-hand side, when the blanket sensor is active, are depicted by dashed red lines.}}\label{fig:blanket_example}
	\end{figure}
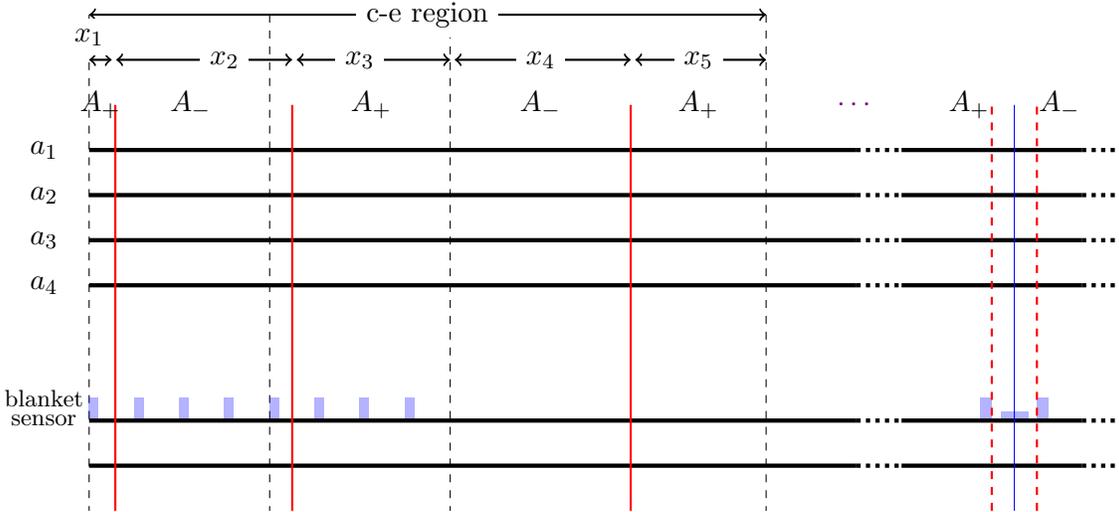

\subsection*{Sensor and blanket sensor agents}

The formal definitions of the sensor agents and the blanket sensor agents were given in the main text, see Definition \ref{def:sensors} and Definition \ref{def:blanket} respectively. Here, we explain in more detail how these agents make use of the bit-detection gadgets (which lie in the region $R_i$) to detect the positions of the cuts (for the sensor agents) or to detect large discrepancies on the volumes of the two labels in the c-e region (for the blanket sensor agents).

Starting from the sensor agents, recall that each such agent of $S_i \subset C_i$ has a small value-block (of value $1/10$) in the c-e region and its remaining value ($9/10$) lies in the circuit encoding region and particularly, in the sub-region $R_i$ (recall that the circuit-encoding region $R$ is partitioned into sub-regions $R_i$, one for each circuit encoder $C_i$, where most of the gadgetry of the encoder lies). In particular, the sensor agent has two thin blocks of value $9/20$ in the c-e region and this is precisely the bit detection gadget of the agent, as described in the previous subsection - see Figure \ref{fig:gate-example}, top, for an illustration. If the value-block on the left-hand side (in the c-e region) is labelled \lminus, then the cut on the right-hand side intersects the rightmost value-block (i.e., ``jumps'' to the right) and if it is \lplus, then it ``jumps'' to the left. 
This information is then passed on to the next level of circuit encoding agents, those that implement the pre-processing unit of the circuit (see Section \ref{sec:c1}) and the subsection following this one. In \cite{FG17}, we referred to this information as ``raw data'' (although the cut extraction mechanism there had to be more elaborate, due to the inversely-exponential precision). The pre-processing unit is responsible for converting the raw data into appropriate inputs for the circuit $C_{VT}$, which encode coordinate of points on the M\"{o}bius-simplex. These inputs are then ``propagated'' through the encoding of the circuit $C_{VT}$, to produce the appropriate labels at the output gates $g_j$, $j \in \pm [n]$, as described in th following subsection.

The blanket sensor agents use very similar bit detection gadgets in their outputs (i.e., in their value-blocks in region $R_i$), but between their thin and dense value-blocks, they have an additional small value-blocks (the block of value $9\kappa/10$ in Definition \ref{def:blanket}). This is because the blanket sensor agents needs to be able to assume three states: ``excess of \lplus'', ``excess of \lminus'' and ``(approximately) balanced labels''. The latter option corresponds to the cut associated with the blanket sensor agent intersecting the middle value-block (therefore not ``jumping'' to either side), whereas the other two options correspond to the cut ``jumping'' to either the right or the left side, where the choice depends on the over-represented label and the parity of the index of the blanket sensor agent. It is straightforward (as before) to interpret these positions as boolean values. The main idea in \cite{FG17} is that if the blanket sensor agents are active, then this information will ``over-ride'' the circuit $C_{VT}$ and generate an imbalance of labels in the feedback provided to the c-e agents directly, essentially bypassing the output gates of $C_{VT}$. We use the same principle here, but since we now have many blanket sensor agents, extra care must be taken to make sure that no artificial solutions are introduced when colouring the domain. The details on how the input from the blanket-sensors affects the feedback to the c-e agents are presented in Section \ref{sec:bsa}.

\subsection*{Encoding of the circuit $C_{VT}$ of \nvhdt.}

Before we explain how the circuit of \nvhdt\ is encoded using the circuit encoders $C_i$, we present the main building blocks for simulating boolean circuits based on bit-detection gadget, first presented in \cite{FFGZ18} and later adapted and used in \cite{FG17}.

Consider a boolean gate that is an AND, an OR, or a NOT gate, denoted $g_\land$, $g_\lor$ and $g_\neg$ respectively. Let $in_1$, $in_2$ and $out$ be intervals such that $|in_1| = |in_2| = |out| = 1$. We will encode these gates using gate gadgets, shown in Figure \ref{fig:boolean-gate-gadgets} (from \cite{FG17}).

Note that the gadget corresponding to the NOT gate only has one input, whereas the gadgets for the AND and OR gates have two inputs. In the interval $out$, each gadget has two bit detection gadgets - in the case of the NOT gate these are even, but in the case of the AND and OR gates, they are uneven. Also note that for the inputs, as well as the output of the NOT gate, the label on the left-hand side of the cut is $\lplus$ and the label on the right-hand side will be $\lminus$, whereas for the outputs to the OR and AND gate, the label on the left-hand side of the cut is $\lminus$ and the label on the right hand side is $\lplus$ . This can be achieved with the appropriate use of parity gadgets, i.e., simple valuation blocks that force cuts to lie in specific positions, only to change the parity of the cut sequence (see \cite{FG17} for more details). 

As we mentioned earlier, the bit-detection gagdets allow us to extract boolean values corresponding to the positions of the cuts in the c-e region - this is achieved via the use of the sensor agents. In the next step, these values (referred earlier as the ``raw data'') are supplied into a gadget that is referred to as the pre-processing circuit in \cite{FG17}. The role of this circuit is to convert this information to coordinates of points on the M\"{o}bius-simplex, which then will go through a coordinate transformation (see Section \ref{sec:coord-sys})  and will be used as inputs to the encoding of $C_{VT}$. The pre-processing circuit can be implemented using the boolean gate gadgets described above. The same is true for the actual circuit $C_{VT}$ as well, which can be simulated using the same set of gadgets, using the principle described in Figure \ref{fig:gates} (from \cite{FG17}). The outputs of the pre-preprocessing circuit are inputed to the input gates of $C_{VT}$ and their outputs, in turn, are inputted to the gates on the next level and this process carries on until the output gates of $C_{VT}$. 

\begin{figure}
	\centering
	\begin{minipage}{0.35\textwidth}
		\begin{tikzpicture}[scale=1,transform shape]
		\node (a_1) at (-10pt,0pt) {\small{$G_{\neg}(in_1,out)$}}; 
		\node (a_2) at (110pt, 0pt) {};
		\draw (a_1)--(a_2);
		\draw[fill=lightgray] (36pt,0pt) rectangle (57pt,5pt);
		\draw[fill=lightgray] (78pt,0pt) rectangle (84pt,20pt);
		\draw[fill=lightgray] (93pt,0pt) rectangle (99pt,20pt);
		
		\draw [
		thick,
		decoration={
			brace,
			mirror,
			raise=5pt
		},
		decorate
		] (36pt,0pt) -- (57pt,0pt)
		node [pos=0.5,anchor=north,yshift=-10pt] {\small{${in_1}$}}; 
		
		\draw [
		thick,
		decoration={
			brace,
			mirror,
			raise=5pt
		},
		decorate
		] (78pt,0pt) -- (99pt,0pt)
		node [pos=0.5,anchor=north,yshift=-10pt] {\small{$out$}};

		\draw[dashed,color=blue] (36pt,30pt) -- (36pt, -10pt);
		\draw[dashed,color=blue] (96pt,30pt) -- (96pt, -10pt);
		
		\draw[dashed,color=red] (57pt,30pt) -- (57pt, -10pt);
		\draw[dashed,color=red] (81pt,30pt) -- (81pt, -10pt);
		
		\node[color=blue] at (25pt,40pt) {\scriptsize{$\lplus$}};
		\node[color=red] at (50pt,40pt) {\scriptsize{$\lplus$}};
		
		\node[color=blue] at (90pt,40pt) {\scriptsize{$\lplus$}};
		\node[color=red] at (75pt,40pt) {\scriptsize{$\lplus$}};
		
		\end{tikzpicture}
	\end{minipage}
	\begin{minipage}{0.35\textwidth}
		\begin{tikzpicture}[scale=1,transform shape]
		\node (a_1) at (-10pt,0pt) {\small{$G_{\lor}(in_1,in_2,out)$}}; 
		\node (a_2) at (150pt, 0pt) {};
		\draw (a_1)--(a_2);
		\draw[fill=lightgray] (36pt,0pt) rectangle (57pt,5pt);
		\draw[fill=lightgray] (78pt,0pt) rectangle (99pt,5pt);
		\draw[fill=lightgray] (120pt,0pt) rectangle (126pt,20pt);
		\draw[fill=lightgray] (135pt,0pt) rectangle (141pt,25pt);
		
		\draw [
		thick,
		decoration={
			brace,
			mirror,
			raise=5pt
		},
		decorate
		] (36pt,0pt) -- (57pt,0pt)
		node [pos=0.5,anchor=north,yshift=-10pt] {\small{${in_1}$}}; 
		
		\draw [
		thick,
		decoration={
			brace,
			mirror,
			raise=5pt
		},
		decorate
		] (78pt,0pt) -- (99pt,0pt)
		node [pos=0.5,anchor=north,yshift=-10pt] {\small{$in_2$}};

		\draw [
		thick,
		decoration={
			brace,
			mirror,
			raise=5pt
		},
		decorate
		] (121pt,0pt) -- (141pt,0pt)
		node [pos=0.5,anchor=north,yshift=-10pt] {\small{$out$}};

		\draw[dashed,color=blue] (36pt,30pt) -- (36pt, -10pt);
		\draw[dashed,color=red] (37pt,30pt) -- (37pt, -10pt);
		\draw[dashed,color=red] (99pt,30pt) -- (99pt, -10pt);
		\draw[dashed,color=blue] (78pt,30pt) -- (78pt, -10pt);

		\draw[dashed,color=blue] (124pt,30pt) -- (124pt, -10pt);
		\draw[dashed,color=red] (137pt,30pt) -- (137pt, -10pt);
		
		\node[color=blue] at (25pt,40pt) {\scriptsize{$\lplus$}};
		\node[color=red] at (25pt,30pt) {\scriptsize{$\lplus$}};
		
		\node[color=red] at (95pt,40pt) {\scriptsize{$\lplus$}};
		\node[color=blue] at (70pt,40pt) {\scriptsize{$\lplus$}};
		
		\node[color=red] at (135pt,40pt) {\scriptsize{$\lminus$}};
		\node[color=blue] at (117pt,40pt) {\scriptsize{$\lminus$}};
		
		\end{tikzpicture}
	\end{minipage}
	%
	%\begin{minipage}{0.33\textwidth}
	\begin{tikzpicture}[scale=1,transform shape]
	\node at (10pt,40pt) {};
	\node (a_1) at (-10pt,0pt) {\small{$G_{\land}(in_1,in_2,out)$}}; 
	\node (a_2) at (150pt, 0pt) {};
	\draw (a_1)--(a_2);
	\draw[fill=lightgray] (36pt,0pt) rectangle (57pt,5pt);
	\draw[fill=lightgray] (78pt,0pt) rectangle (99pt,5pt);
	\draw[fill=lightgray] (120pt,0pt) rectangle (126pt,25pt);
	\draw[fill=lightgray] (135pt,0pt) rectangle (141pt,20pt);
	
	\draw [
	thick,
	decoration={
		brace,
		mirror,
		raise=5pt
	},
	decorate
	] (36pt,0pt) -- (57pt,0pt)
	node [pos=0.5,anchor=north,yshift=-10pt] {\small{${in_1}$}}; 
	
	\draw [
	thick,
	decoration={
		brace,
		mirror,
		raise=5pt
	},
	decorate
	] (78pt,0pt) -- (99pt,0pt)
	node [pos=0.5,anchor=north,yshift=-10pt] {\small{$in_2$}};
	
	\draw [
	thick,
	decoration={
		brace,
		mirror,
		raise=5pt
	},
	decorate
	] (121pt,0pt) -- (141pt,0pt)
	node [pos=0.5,anchor=north,yshift=-10pt] {\small{$out$}};
	
	\draw[dashed,color=blue] (36pt,30pt) -- (36pt, -10pt);
	\draw[dashed,color=red] (37pt,30pt) -- (37pt, -10pt);
	\draw[dashed,color=red] (99pt,30pt) -- (99pt, -10pt);
	\draw[dashed,color=blue] (78pt,30pt) -- (78pt, -10pt);

	\draw[dashed,color=blue] (124pt,30pt) -- (124pt, -10pt);
	\draw[dashed,color=red] (122pt,30pt) -- (122pt, -10pt);
	
	\node[color=blue] at (25pt,40pt) {\scriptsize{$\lplus$}};
	\node[color=red] at (25pt,30pt) {\scriptsize{$\lplus$}};
	
	\node[color=red] at (95pt,40pt) {\scriptsize{$\lplus$}};
	\node[color=blue] at (70pt,40pt) {\scriptsize{$\lplus$}};
	
	\node[color=red] at (115pt,30pt) {\scriptsize{$\lminus$}};
	\node[color=blue] at (115pt,40pt) {\scriptsize{$\lminus$}};
	\end{tikzpicture}
	%\end{minipage}	
	\caption{\small{The Boolean Gate Gadgets encoding the NOT, OR and AND gates. For visibility, the valuation blocks are not according to scale. For the NOT gate, the input has value $0.25$ and the output blocks have volume $0.375$ each. For the OR (respectively AND) gate, the input blocks have value $0.125$ each and the output blocks have value $0.3125$ and $0.4375$ (respectively $0.4375$ and $0.3125$). The cuts corresponding to pairs or triples of inputs and outputs have the same colour, and the labels on the left-side of these cuts are shown and colour-coded in the same way. For the NOT gate, when the input cut sits of the left (the blue cut), then the output cut must sit on the right (the blue cut), to compensate for the excess of $\lminus$ and oppositely for when the input cut sits of the right (the red cut). For the OR and AND gates, again the cuts corresponding to two inputs and one output have the same colour. For the OR gate, when both inputs cut sit on the left (the blue cuts), the output cut sits on the left as well, to compensate for the excess of $\lminus$ (notice that the left-hand side of the output cut is labelled $\lminus$. When one input sits on the left and the other one on the right, the inputs detect no discrepancy in the balance of labels and the output jumps to the right, because the output blocks are uneven (the red cuts). The operation of the AND gate is very similar; here the cases shown are those where the inputs are $0$ and $0$ and the ouput is $0$ (the blue cuts) and where the inputs are $0$ and $1$ and the output is still $0$ (the red cuts).}}
	\label{fig:boolean-gate-gadgets}
\end{figure}
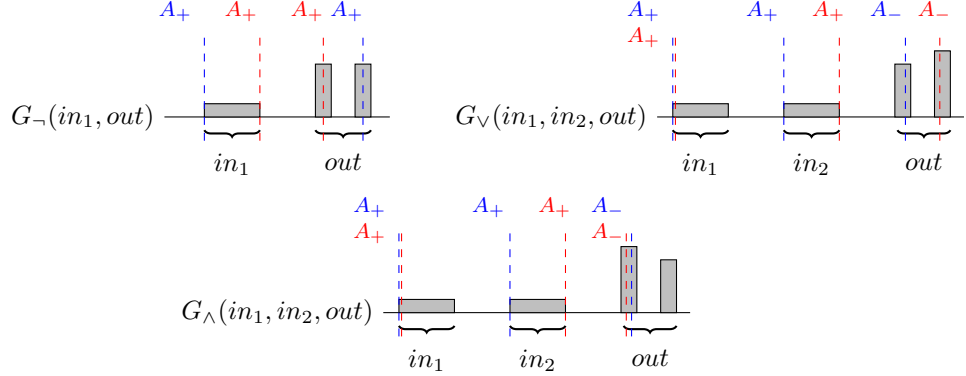

\begin{figure}
	\begin{tikzpicture}[scale=0.85, transform shape]
	%\node (up) at (0pt,20pt) {};
	\node (a_1) at (84pt,0pt) {}; 
	\node (a_2) at (485pt, 0pt) {};
	\draw (a_1)--(a_2);
	
	\node (a_1) at (84pt,-20pt) {}; 
	\node (a_2) at (485pt, -20pt) {};
	\draw (a_1)--(a_2);

	\node (a_1) at (84pt,-40pt) {}; 
	\node (a_2) at (485pt, -40pt) {};
	\draw (a_1)--(a_2);
	
	\node (a_1) at (84pt,-60pt) {}; 
	\node (a_2) at (485pt, -60pt) {};
	\draw (a_1)--(a_2);
	
	%\node (a_1) at (84pt,-100pt) {}; 
	%\node (a_2) at (485pt, -100pt) {};
	%\draw (a_1)--(a_2);
	
	%\node (a_1) at (84pt,-120pt) {}; 
	%\node (a_2) at (485pt, -120pt) {};
	%\draw (a_1)--(a_2);

	%\draw[fill=lightgray] (120pt,0pt) rectangle (130pt,5pt);
	
	\draw[fill=lightgray] (110pt,0pt) rectangle (115pt,10pt);
	\draw[fill=lightgray] (120pt,0pt) rectangle (125pt,10pt);
	
	\draw[fill=lightgray] (135pt,-20pt) rectangle (140pt,-10pt);
	\draw[fill=lightgray] (145pt,-20pt) rectangle (150pt,-10pt);
	
	%\node (a) at (160pt,-30pt) {$\ldots$};
	
	%	\draw[fill=lightgray] (170pt,-60pt) rectangle (175pt,-50pt);
	%\draw[fill=lightgray] (180pt,-60pt) rectangle (185pt,-50pt);
	
	\draw[fill=lightgray] (115pt,-40pt) rectangle (120pt,-35pt);
	\draw[fill=lightgray] (140pt,-40pt) rectangle (145pt,-35pt);
	
	\draw[fill=lightgray] (230pt,-40pt) rectangle (240pt,-25pt);
	\draw[fill=lightgray] (250pt,-40pt) rectangle (260pt,-35pt);
	
	\draw[fill=lightgray] (240pt,-60pt) rectangle (250pt,-55pt);
	
	\draw[fill=lightgray] (380pt,-60pt) rectangle (390pt,-50pt);
	\draw[fill=lightgray] (400pt,-60pt) rectangle (410pt,-50pt);

	\draw [
	thick,
	decoration={
		brace,
		raise=5pt
	},
	decorate
	] (485pt,15pt) -- (485pt,0-20pt)
	node [pos=0.5,anchor=west,xshift=8pt] {\scriptsize{Outputs of sensor agents}};
	
	\draw [
	thick,
	decoration={
		brace,
		raise=5pt
	},
	decorate
	] (485pt,-25pt) -- (485pt,-40pt)
	node [pos=0.5,anchor=west,xshift=8pt] {\scriptsize{An AND input gate gadget}};
	
	\draw [
	thick,
	decoration={
		brace,
		raise=5pt
	},
	decorate
	] (485pt,-45pt) -- (485pt,-60pt)
	node [pos=0.5,anchor=west,xshift=8pt] {\scriptsize{A NOT output gate gadget}};

	\draw[color=blue,dashed] (137.5pt,20pt) -- (137.5pt,-80pt);
	\draw[color=blue,dashed] (122.5pt,20pt) -- (122.5pt,-80pt);
	\draw[color=blue,dashed] (235.5pt,20pt) -- (235.5pt,-80pt);
	\draw[color=blue,dashed] (402.5pt,20pt) -- (402.5pt,-80pt);
	
	\draw[color=red,dashed] (147.5pt,20pt) -- (147.5pt,-80pt);
	\draw[color=red,dashed] (123.5pt,20pt) -- (123.5pt,-80pt);
	\draw[color=red,dashed] (255.5pt,20pt) -- (255.5pt,-80pt);
	\draw[color=red,dashed] (387.5pt,20pt) -- (387.5pt,-80pt);

	\end{tikzpicture}
	\caption{\small{The basic idea behind the gate-agents encoding the gates of $C_1$. The picture denotes a simplified case where two input bits from the bit-encoders are supplied to an enconding of an input AND gate of $C_i$ and the output bit of this gate is in turn supplied to the encoding of a NOT output gate of $C_i$. Note that if for example the sensor agents detect the values $1$ and $0$ respectively (the blue cuts), then the output of the AND gate is $0$ (i.e. the blue cut sits at the left of the AND gate agent's bit detector) and the output of the circuit is $1$ (again, see the blue cut that sit on the rightmost valuation block of the NOT gate agent. Similarly, if the sensor agents detect values $1$ and $1$ (the red cuts), then the output of the AND gate is $1$ and the output of the circuit is $0$.}} 
	\label{fig:gates}
\end{figure}
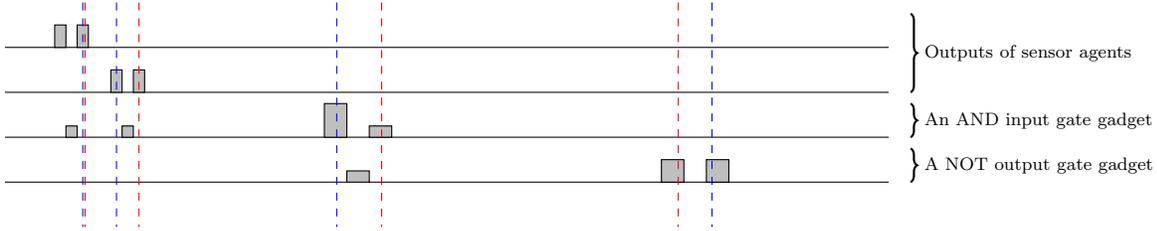

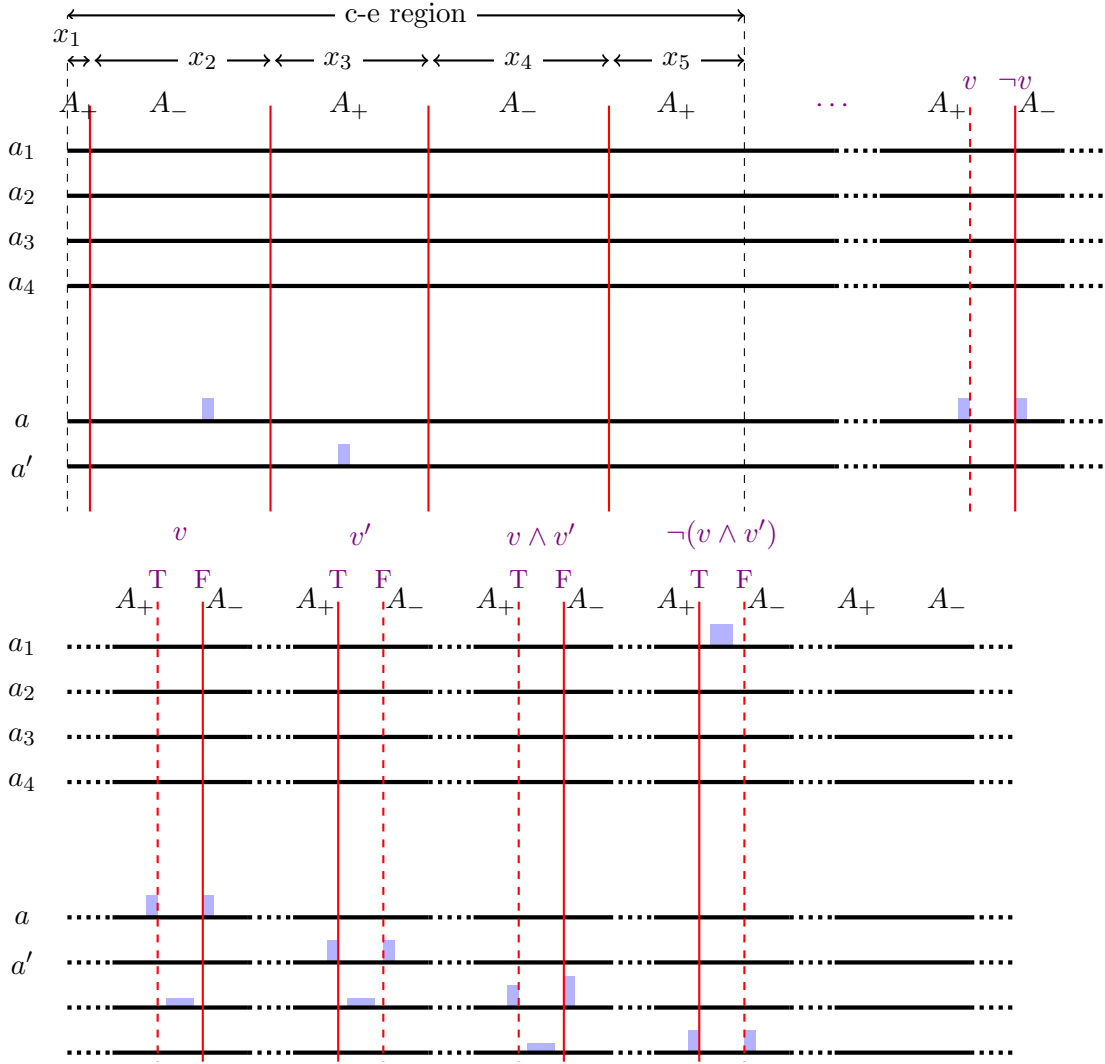
\begin{figure}
	\center{
		\begin{tikzpicture}[scale=0.6]
		\tikzstyle{line-style}=[black,ultra thick]
		\tikzstyle{cut-style}=[red,thick]
		\tikzstyle{valuestyle}=[fill=blue,opacity=0.3]
		
		\draw[dashed](0,23)--(0,12);\draw[dashed](15,23)--(15,12);
		\draw[thick,<->](0,23)--(15,23);\draw[thick,<->](0,22)--(0.5,22);
		\draw[thick,<->](0.6,22)--(4.5,22);\draw[thick,<->](4.6,22)--(8,22);
		\draw[thick,<->](8.1,22)--(12,22);\draw[thick,<->](12.1,22)--(15,22);
		
		\node[fill=white]at(7.5,23){c-e region};\node[fill=white]at(0,22.5){$x_1$};
		\node[fill=white]at(3,22){$x_2$};\node[fill=white]at(6,22){$x_3$};
		\node[fill=white]at(10,22){$x_4$};\node[fill=white]at(13.5,22){$x_5$};
		
		\fill[valuestyle,yshift=0cm](3,14)--(3,14.5)--(3.25,14.5)--(3.25,14)--cycle;

		%tri-block gadgets, including in bottom half:
		\foreach \x/\y in {2/3,6/2,14/0,20/14}{
			\fill[valuestyle,xshift=\x cm,yshift=\y cm](0,0)--(0,0.5)--(-0.25,0.5)--(-0.25,0)--cycle;
			\fill[valuestyle,xshift=\x cm,yshift=\y cm](1,0)--(1,0.5)--(1.25,0.5)--(1.25,0)--cycle;
			%\fill[valuestyle,xshift=\x cm,yshift=\y cm](0.2,0)--(0.2,0.2)--(0.8,0.2)--(0.8,0)--cycle;
		}
		
		%tri-block AND gadget, in bottom half:
		\foreach \x/\y in {10/1}{
			\fill[valuestyle,xshift=\x cm,yshift=\y cm](0,0)--(0,0.5)--(-0.25,0.5)--(-0.25,0)--cycle;
			\fill[valuestyle,xshift=\x cm,yshift=\y cm](1,0)--(1,0.7)--(1.25,0.7)--(1.25,0)--cycle;
			%\fill[valuestyle,xshift=\x cm,yshift=\y cm](0.2,0)--(0.2,0.2)--(0.8,0.2)--(0.8,0)--cycle;
		}
		
		% flat uniblock gadgets
		\foreach \x/\y in {2/1,6/1,10/0}{
			\fill[valuestyle,xshift=\x cm,yshift=\y cm](0.2,0)--(0.2,0.2)--(0.8,0.2)--(0.8,0)--cycle;
		}
		
		\fill[valuestyle,yshift=0cm](6,13)--(6,13.5)--(6.25,13.5)--(6.25,13)--cycle;
		
		\foreach \y in {20,19,18,17,14,13}{
			\draw[line-style](0,\y)--(17,\y);
			\draw[ultra thick,dotted](17,\y)--(18,\y);
			\draw[line-style](18,\y)--(22,\y);
			\draw[ultra thick,dotted](22,\y)--(23,\y);
		}
		
		\node[yshift=0cm]at(-1,20){$a_1$};\node[yshift=0cm]at(-1,19){$a_2$};
		\node[yshift=0cm]at(-1,18){$a_3$};\node[yshift=0cm]at(-1,17){$a_4$};
		\node[yshift=0cm]at(-1,14){$a$};\node[yshift=0cm]at(-1,13){$a'$};

		\foreach \x in {0.5,4.5,8,12,21}{\draw[cut-style,yshift=0cm](\x,12)--(\x,21);}
		\draw[cut-style,dashed,yshift=0cm](20,12)--(20,21);
		\node[color=black,yshift=0cm]at(0.25,21){$\lplus$};\node[color=black,yshift=0cm]at(2.25,21){$\lminus$};\node[color=black,yshift=0cm]at(6.25,21){$\lplus$};
		\node[color=black,yshift=0cm]at(10,21){$\lminus$};
		\node[color=black,yshift=0cm]at(13.5,21){$\lplus$};
		\node[color=violet,yshift=0cm]at(17,21){$\cdots$};
		\node[color=black,yshift=0cm]at(19.5,21){$\lplus$};
		
		\node[color=violet,yshift=0cm]at(20,21.5){$v$};
		\node[color=violet,yshift=0cm]at(21,21.5){$\lnot v$};

		\node[color=black,yshift=0cm]at(21.5,21){$\lminus$};

		\fill[valuestyle,xshift=-4cm,yshift=-11cm](18.25,20)--(18.25,20.5)--(18.75,20.5)--(18.75,20)--cycle;
		
		\foreach \y in {9,8,7,6,3,2,1,0}{
			\draw[line-style,dotted](0,\y)--(1,\y);
		}
		\foreach \x in {1,5,9,13,17}{
			\foreach \y in {9,8,7,6,3,2,1,0}{
				\draw[line-style](\x,\y)--(3+\x,\y);\draw[ultra thick,dotted](3+\x,\y)--(4+\x,\y);
			}}
			\foreach \x in {1,5,9,13,17}{
				\node[color=black]at(\x+0.5,10){$\lplus$};\node[color=black]at(\x+2.5,10){$\lminus$};
			}
			
			\node[color=violet]at(2.5,11.5){$v$};
			\node[color=violet]at(2,10.5){\small T};\node[color=violet]at(3,10.5){\small F};
			
			\node[color=violet]at(6.5,11.5){$v'$};
			\node[color=violet]at(6,10.5){\small T};\node[color=violet]at(7,10.5){\small F};
			
			%\node[color=violet,rotate=45]at(10.5,11){$v\land v'$};
			%\node[color=violet,rotate=45]at(11.75,11){$\lnot(v\land v')$};
			%\node[color=violet,rotate=45]at(14.5,11){$\lnot(v\land v')$};
			%\node[color=violet,rotate=45]at(15.75,11){$v\land v'$};
			
			\node[color=violet]at(10.5,11.5){$v\land v'$};
			\node[color=violet]at(10,10.5){\small T};\node[color=violet]at(11,10.5){\small F};
			
			\node[color=violet]at(14.5,11.5){$\lnot(v\land v')$};
			\node[color=violet]at(14,10.5){\small T};\node[color=violet]at(15,10.5){\small F};

			\node[yshift=0cm]at(-1,9){$a_1$};\node[yshift=0cm]at(-1,8){$a_2$};
			\node[yshift=0cm]at(-1,7){$a_3$};\node[yshift=0cm]at(-1,6){$a_4$};
			\node[yshift=0cm]at(-1,3){$a$};\node[yshift=0cm]at(-1,2){$a'$};

			\foreach \x in {3,6,11,14}{\draw[cut-style,yshift=0cm](\x,-0.2)--(\x,10);}
			\foreach \x in {2,7,10,15}{\draw[cut-style,dashed,yshift=0cm](\x,10)--(\x,-0.2);}

			\end{tikzpicture}
			\caption{\small{Example showing gate simulation: $n=4$;
					agents $a$ and $a'$ have corresponding propositional variables $v$ and $v'$
					that are two inputs to a circuit.
					$v=\false$ since $a$'s sensor-value lies in a region labelled $\lminus$, similarly
					$v'=\true$ since $a'$'s sensor-value lies in a region labelled $\lplus$.
					Gate-encoding blocks have cuts (shown in red) at two possible positions corresponding to
					\true\ and \false; a dashed-line shows the alternative position ({\em not} taken)
					by the cut (itself shown as a solid line).
					c-e agent $a_1$ receives feedback based on the conjunction of 2 input bits.}}\label{fig:gate-example}}
	\end{figure}
	
	\subsection*{Unreliable circuits and the stray cut}
	
	In the main text, we mentioned that we are guaranteed that at most $n$ cuts will lie in the c-e region (Observation \ref{obs:cer}) and at least $n-1$ cuts will lie in the c-e region, as otherwise some blanket sensor agent would be active (Proposition \ref{prop:s-r-radius}). If $n-1$ cuts lie in the c-e region, this means that one of the $n$ c-e cuts has moved into the circuit encoding region; following \cite{FG17}, we will refer to that cut as a \emph{stray cut}. As we highlighted in \cite{FG17}, the presence of a stray cut may have the following two consequences.
	
	\begin{enumerate}
		\item It intersects the circuit-encoding region $R_i$ of some circuit encoder $C_i$, for $i \in \{1,\ldots,p^C\}$.
		\item It flips the parity of the circuit encoders $C_i$, with $R_i < c$, where $c$ is the position of the stray cut in $R_{i-1}$. In other words, if
		the first cut in $R_i$ was expecting to see $\lplus$ on its left-hand side, it now sees $\lminus$ and vice-versa.
	\end{enumerate}  
	The first effect is not much of a problem, we will simply assume that the corresponding circuit is unreliable. As explained in the main text, an unreliable circuit can have an arbitrary (or even adversarial) effect on the volume of \lplus and \lminus supplied as feedback to the c-e agents, but since there will be at most $1$ unreliable circuit, its effect will be ``outvoted'' by the remaining reliable circuits (circuits that do not receive reliable inputs will also be outvoted by those that do, since there are at most $n$ of them, see Observation \ref{obs:reliable}).
	
	The parity flip that the stray cut induces on other circuit-encoders however seems potentially more worrisome., since it could flip the output of the sensor agents that detect the positions of the cuts. This is taken care of by the presence of the reference sensor agent (Definition \ref{def:ref}) which achieves the disorientation of the domain. As explained in Section \ref{sec:c1}, the outputs of circuit $C_i$ are taken with reference to the value $s_{1,1}$ of the reference senso agent,
	in the sense that after simulating $C_{VT}$ we take the exclusive-or with $s_{1,1}$; this can be implemented in the circuit using the boolean gate gadgets explained above, in a manner similar to \cite{FG17} (via the use of the XOR sub-circuit, see Section 4.4.2 in \cite{FG17}). This allows us to use a similar ``double-negative lemma'' as the one we used in \cite{FG17} (see Lemma 5.4).
	
	\newpage
	
	\section{Membership in PPA}\label{sec:inppa}
	
	We show that \dhs\ belongs to \ppa, which appears to be folklore that has not,
	to out knowledge, been written down.
	(In fact, in \cite{Pap} the problem was claimed to belong to \ppad, which is now seen to be incorrect
	subject to $\ppad\not=\ppa$, by the result in \cite{ABB}).
	Since Theorem~\ref{thm:dhs} reduces from 2-thief \ns\ to \dhs,
	it follows immediately that 2-thief \ns\ belongs to \ppa.
	In Section~\ref{sec:nsinppa} we go a bit further for \ns:
	we show that \ns\ belongs to \ppa\ whenever the number of thieves is a power of 2.

	\subsection{\dhs\ is in \ppa}\label{sec:hsinppa}
	
	We use Freund and Todd's \cite{FT81} construction of an undirected graph
	with known degree-1 vertex, based on a ``special triangulation'' of a high-dimensional $L_1$ ball
	(i.e. a high-dimensional octahedron, also known as the crosspolytope~\cite{M08}).
	Given an instance $I$ of \dhs, we show how to construct a suitable triangulation.
	$I$ contains $n$ sets of points $\{S_1,\ldots,S_n\}$ in $n$-dimensional space;
	we assume coordinates are represented
	as fractions whose numerators and denominators are give via standard binary expansions.
	(Recall that we leave it as an open problem whether \dhs\ remains
	\ppa-hard for points presented in unary. Of course, the ``in \ppa'' result
	follows immediately for that restricted version.)
	
	Based on $I$ we identify an exponentially-large collection of ``candidate hyperplanes''
	as follows, which contains a solution to $I$.
	A candidate hyperplane $H$ is represented by a gradient vector $g_H$ whose
	coordinates are assumed to be normalised to 1 ($L_1$ norm; their absolute values
	sum to 1).
	Given $g_H$, $H$ is obtained by using binary
	search to efficiently find a hyperplane with gradient $g_H$ that bisects the union of the sets $S_i$;
	it is then easy to check whether $H$ is a solution.
	Note that there exists an integer $N$ whose binary expansion
	has length polynomial in $|I|$ such that entries
	of some solution $g_H$ can be assumed to be multiples of $1/N$.
	
	Fix a point $p\in \rset^n$ such that $p$ is not contained in any candidate solution to $I$,
	e.g. $p=(1/2N,1,\ldots,1)$.
	Given any $H$, the ``positive'' side of $H$ is the side that contains $p$.
	Assume we chose $N$ large enough so that if two hyperplanes $H$ and $H'$ have
	gradients $g_{H}$ and $g_{H'}$ whose coordinates differ by at most $1/N$,
	then for any point $x$ in $I$ where they disagree, $x$ lies in $H$ or $H'$
	(and does not lie strictly on the positive side of one and the negative side of the other).
	
	For any $H$, label it as follows. Find the set $S_i$ that is most unevenly split by $H$
	breaking ties lexicographically. Label $H$ with $i$ if most of $S_i$ lies on the positive side,
	otherwise label $H$ with $-i$.
	Each hyperplane $H$ has a $g_H$ that is a point on the $L_1$-norm sphere $S^{n-1}$.
	By construction, antipodal points receive opposite labels.
	We can use these as the vertices of a special triangulation of the octahedral ball $B^n$,
	in which the origin is used as an additional point and is connected to all the
	points on $S^{n-1}$ that correspond to candidate solutions.
	The graph defined in \cite{FT81} is a degree 2 undirected graph with a single
	known degree-1 vertex (the origin), and for which any other degree-1 vertex
	represents a pair of hyperplanes that bisect all the $S_i$.

	\subsection{\ns\ is in \ppa, if the number of thieves is a power of 2}\label{sec:nsinppa}
	
	The version of the problem with two thieves, 2-thief \ns, belongs to \ppa\ since we reduced it to \dhs\ which is shown in Section~\ref{sec:hsinppa} to belong to \ppa. We can extend the membership to PPA for $k$-\ns, when $k$ is a power of $2$, by using the argument of Proposition 3.2 of Alon \cite{Alon87} (here, we take both $k$ and $l$ to be $2$). In particular, we will reduce \ns\ to $4$-\ns.
	
	We start from an instance of $4$-\ns\ (with $4$ thieves) and we regard it as an instance of \ns\ (with $2$ thieves), which we solve using an algorithm for the latter problem. The solution is a sequence of intervals defined by the endpoints of the necklace and $n$ cuts, each belonging to one of the two collections (corresponding to the two thieves), such that each collection contains exactly half of the beads of each colour. Then, we set the beads that lie in intervals belonging to each collection aside and form two new instances of \ns\ (essentially by ``gluing'' the different sub-intervals of the same collection together); note that each new instance will have an even number of beads of each colour, since we initial number of beads from each colour was a multiple of $4$. Then we run the algorithm again on the resulting instances of \ns\ to obtain a partition into $4$ collections ($2$ for each individual instance), which consitutes a partition of the $4$-\ns\ into $4$ collections according to the definition of the problem. If $n$ is the number of colours, the total number of cuts is (at most) $3n$, and therefore this partition is a solution to $4$-\ns. 
	
	The above is a Turing reduction, which can be extended straightforwardly
	to the case of $k$ a power of 2.
	We can convert such a reduction into a many-one reduction by applying
	Theorem 6.1 of Buss and Johnson~\cite{BJ12}, which shows that \ppa\ and
	some related complexity classes are closed under Turing reductions.

\end{document}